\documentclass[11pt,letterpaper]{article}
\usepackage{bbm}

\usepackage[utf8]{inputenc}
\usepackage[T1]{fontenc}
\usepackage{microtype}

\usepackage[margin=1in]{geometry}

\usepackage{amsmath,mathtools}
\usepackage{amsthm,thmtools}

\usepackage{amssymb}
\usepackage{bm,dsfont}
\usepackage{csquotes}

\usepackage{xcolor}

\usepackage{paralist}

\usepackage[small]{caption}
\usepackage{subcaption}
\usepackage{tikz}
\usepackage{pgfplots}
\pgfplotsset{compat=newest}

\usepackage{booktabs,tabularew}
\usepackage[para]{threeparttable}

\usepackage[ruled,linesnumbered,noline]{algorithm2e}

\usepackage{array}

\usepackage{comment}
\usepackage{nicefrac}

\usepackage[ruled]{algorithm2e}
\usepackage{csquotes}

\usepackage{pgfplots}
\usepgfplotslibrary{fillbetween}
\usepgfplotslibrary{groupplots}
\usetikzlibrary{arrows.meta,calc}
\pgfplotsset{compat=1.18}

\usepackage[skip=5pt plus 3pt minus 1pt,parfill=0pt]{parskip}
\makeatletter
\def\thm@space@setup{\thm@preskip=\parskip \thm@postskip=0pt
}
\makeatother

\usepackage[
 backend=biber,
 style=alphabetic,
 citetracker,
 hyperref=auto,
 maxcitenames=5,
 sortcites,
 sorting=nyt,
 maxbibnames=12,
 date=year,
 isbn=false,
 url=false,
 doi=false,
 eprint=false,
]{biblatex}

\newbibmacro{string+doiurlisbn}[1]{
  \iffieldundef{doi}{
    \iffieldundef{url}{
      #1
    }{
      \href{\thefield{url}}{#1}
    }
  }{
    \href{http://dx.doi.org/\thefield{doi}}{#1}
  }
}

\DeclareFieldFormat
[article,inbook,incollection,inproceedings,patent,thesis,unpublished]
  {title}{\usebibmacro{string+doiurlisbn}{#1}}

\usepackage[colorlinks]{hyperref}
\hypersetup{citecolor=green!60!black,linkcolor=blue!70!black,urlcolor=blue!70!black}
\usepackage[capitalize, noabbrev,nameinlink]{cleveref}
\usepackage{xurl}
\hypersetup{breaklinks=true}

\crefname{appendix}{Appendix}{Appendices}
\Crefname{appendix}{Appendix}{Appendices}

\usepackage{xspace}

\declaretheorem[numberwithin=section]{theorem}
\declaretheorem[sibling=theorem]{lemma}

\declaretheorem[sibling=theorem]{corollary}

\DeclarePairedDelimiter\ceil{\lceil}{\rceil}

\DeclareMathOperator*{\argmax}{arg\,max}

\definecolor{col1}{HTML}{56b3e9}
\definecolor{col2}{HTML}{e69f00}
\definecolor{col3}{HTML}{009e74}
\definecolor{col4}{HTML}{cc79a7}
\definecolor{col5}{HTML}{d55e00}
\definecolor{col6}{HTML}{0071b2}

\newcommand{\opt}{\mathrm{OPT}}
\newcommand{\OPT}{\mathrm{OPT}}
\newcommand{\alg}{\mathrm{ALG}}
\newcommand{\ALG}{\mathrm{ALG}}

\newcommand{\ex}{\mathbb{E}}

\newcommand{\pr}{\mathbb{P}}
\newcommand{\cD}{\mathcal{D}}

\newcommand{\GREEDY}{\textsc{Greedy}}
\newcommand{\LGREEDY}{\lambda\textnormal{-}\textsc{Greedy}}

\newcommand{\LBOOSTING}{\lambda\textnormal{-}\textsc{Boosting}}
\newcommand{\LFREEZE}{\lambda\textnormal{-}\textsc{Freeze}}
\newcommand{\PGREEDY}{\textsc{ParallelGreedy}}

\newcommand{\Var}{\mathrm{Var}}

\newcommand{\drm}{\mathrm{d}}

\newcommand{\Ecal}{\mathcal{E}}
\newcommand{\Fcal}{\mathcal{F}}

\newcommand{\Ucal}{\mathcal{U}}

\newcommand{\Ebb}{\mathbb{E}}

\newcommand{\Nbb}{\mathbb{N}}
\newcommand{\Pbb}{\mathbb{P}}

\newcommand{\Rbb}{\mathbb{R}}

\newcommand{\one}{\mathbbm{1}}
\newcommand{\oneb}{\boldsymbol{\mathbbm{1}}}

\usepackage[colorinlistoftodos,prependcaption,textsize=scriptsize]{todonotes}

\newcommand{\romain}[1]{\todo[backgroundcolor=green!25,bordercolor=green]{R:~#1}}

\title{Online Scheduling with a Stochastic Signal} 

\author{Romain Cosson\thanks{Courant Institute, New York University, USA.}\and Jingwei Li\thanks{IEOR, Columbia University, USA.} \and Alexander Lindermayr\thanks{Institut für Mathematik, Technische Universität Berlin, Germany. Part of this work was done while the author was a research fellow at the
Simons Institute for the Theory of Computing for the program on ``Algorithmic Foundations for Emerging Computing
Technologies'' in Fall 2025.} \and Jens Schlöter\thanks{Centrum Wiskunde \& Informatica (CWI), The Netherlands.}}

\date{}

\begin{document}

\maketitle

\begin{abstract}
Nonclairvoyant scheduling is a fundamental online model in which 
processing times are initially unknown to the scheduler. 
Unfortunately, for important objectives such as total completion time and makespan, worst-case analysis yields pessimistic guarantees: 
every nonclairvoyant algorithm has a competitive ratio of at least
$2$ for these objectives.

Recent work introduced $\varepsilon$-clairvoyance, 
where a scheduler receives a signal once an $\varepsilon$-fraction of a job remains (FOCS'25, NeurIPS'25). This model avoids giving the algorithm
a priori predictions as done in learning-augmented algorithms, 
a practice that is often hard to justify in practice. However, existing
algorithms and analyses rely crucially on signal times being precise,
an assumption hardly justifiable in applications such as task profiling.

We introduce stochastic clairvoyance, a beyond-worst-case model 
in which each job emits a randomly timed signal during its execution, 
drawn from a distribution over its processing length. For this model, 
we design new online scheduling algorithms whose 
competitive ratios are strictly below $2$ for minimizing total
completion time and makespan. On the technical side, we prove a new black-box theorem that converts bounds on expected pairwise job delays into competitive
guarantees via a continuous amortized charging argument.

Our results show that stochastic clairvoyance is not merely a curiosity: 
it yields robust improvements across different scheduling objectives and
machine environments. More broadly, stochastic clairvoyance suggests a new direction in beyond-worst-case analysis for online algorithms, 
and builds a bridge between learning-augmented algorithms and stochastic information models.
\end{abstract}

\thispagestyle{empty}
\newpage
\setcounter{tocdepth}{2}
\tableofcontents
\thispagestyle{empty}
\newpage
\setcounter{page}{1}

\section{Introduction}

We study fundamental online scheduling problems in a new information model.
We focus on 
preemptive
machine scheduling problems, where all jobs are available at time $0$ and an algorithm must 
choose at each time $t$ a feasible allocation of $n$ jobs to $m$ parallel identical machines. 
An equivalent description is that, at any time $t$, an algorithm must select a rate $y_j(t) \in [0,1]$ for each job $j$ such that $\sum_j y_j(t) \leq m$. A schedule is feasible if each job $j$ receives its processing requirement $p_j$, that is, $\int_0^\infty y_j(t) \mathrm{d}t \geq p_j$.

In many applications, the assumption that all jobs' processing requirements are known upfront is too strong.
This motivated the study of nonclairvoyant algorithms~\cite{MotwaniPT94,FeldmannST91,ShmoysWW91}. In this form of online scheduling~\cite{PruhsST04}, an algorithm has to build a schedule by irrevocably selecting rates $y(t)$ over time $t$ and only learns about a job's processing requirement $p_j$ when the job completes; previous allocations cannot be undone.
Nonclairvoyant scheduling is well-studied 
for a variety of objectives such as minimizing the
sum of (weighted) completion times~\cite{MotwaniPT94,BeaumontBEM12,ImKM18,JLM26},
sum of (weighted) flow times~\cite{Edmonds00,KalyanasundaramP00,BecchettiL04,ImKMP14,ImKM18},
or makespan~\cite{ShmoysWW95}.
While these model important quality-of-service
objectives in many practical applications~\cite{PinedoBook},
nonclairvoyant algorithms perform poorly due to their pessimistic nature: for total completion time and makespan,
every nonclairvoyant algorithm has a competitive ratio of at least $2$ 
\cite{MotwaniPT94,ShmoysWW95}, which is the worst-case ratio between an algorithm's objective value and the optimum objective value in hindsight.

The assumption of nonclairvoyance is very strong: the algorithm does not know the processing time of a job until it is completed.
This made it a perfect use case
for beyond-worst-case models~\cite{MitzenmacherV22,R2020}.
Many works study learning-augmented algorithms \cite{PurohitSK18,LindermayrM25permutation,LassotaLMS23,DBLP:conf/nips/BenomarP23,BenomarP24nonclarivoyant,ImKQP23,DinitzILMV22,BampisKLP26,AzarLT21,AzarLT22,BenomarCLS25,GuptaKPW26}\footnote{We refer to \cite{alps} for an organized list of works in this area.} with the caveat that predictions are assumed to be accessible upfront. This is a questionable assumption in many practical applications: how can predictions be inferred when no information about the instance or jobs, such as log or profiling data, is initially available?

Motivated by this question, \cite{GuptaKLSY25,BenomarCLS25} considered $\varepsilon$-clairvoyance: an algorithm receives a \emph{signal} once an $\varepsilon$-fraction of the job remains. They showed that one can break the nonclairvoyant lower bounds for any $\varepsilon > 0$ by proving a tight $2-\varepsilon$ guarantee for total completion time. 
The $\varepsilon$-clairvoyant information model has surprisingly many analogs in practice,
for example,
in maintenance or health scheduling~\cite{MSK95maintenance-statistic,Tian2011ConditionBM},
in warehouse scheduling~\cite{MLKS19warehouse,Cao2017RealtimeWarehouse},
in ride-hailing~\cite{Wang2018ridehailing,Wang2018LearningTE},
and in computing systems~\cite{jajoo2022slearn,zaharia2008improving}.
However, in all of these applications, we clearly cannot
assume that signals are \emph{accurate} and \emph{uniform} across all jobs.
Instead, historical data naturally yields distributions over potential signal times.

We introduce \emph{stochastic clairvoyance}
to capture such distributions.
In our model, a signal does not arrive at a fixed fraction of the processing, but at a random time sampled from a distribution over the job's processing requirement.
We develop and analyze algorithms for this model that 
significantly go beyond the nonclairvoyant worst-case lower bounds. Our results push the boundary of what amount of
information suffices for online algorithms to achieve strong provable performance guarantees.

\subsection{Main Results}

We show that even a \emph{single uniformly distributed signal} suffices 
to overcome the classic nonclairvoyant lower bounds for minimizing the total (weighted) completion time or makespan on parallel identical machines.
Beyond online scheduling, we believe that the framework of having a ``minimal'' stochastic information model on top of an adversarial competitive analysis 
can be naturally applied to many online problems and thus opens
an exciting new research direction beyond algorithms with predictions. 
It specifically builds a novel bridge between \emph{learning-augmented algorithms} and \emph{stochastic combinatorial optimization}.

\paragraph{Model.}
We first formally introduce stochastic clairvoyance.
Each job $j$ is associated with a 
distribution $\cD_j$ over $[0,p_j]$.
After the adversary fixes $p_1,\ldots,p_n$, a signal value $S_j$ is sampled from $\cD_j$ independently for each job $j$. 
At the beginning of the instance, the algorithm knows neither $S_j$ nor $p_j$.
When it has processed job $j$ by $S_j$ units, it receives a signal from $j$, and thus learns about the value of $S_j$.
We further consider a \emph{strong} variant of stochastic clairvoyance where the algorithm learns about $p_j$ when it receives $j$'s signal. In the default (weak) variant, $p_j$ is only revealed when the job completes, as in the nonclairvoyant setting.
In the first part of this work, we focus on \emph{uniform} stochastic clairvoyance, that is, $\cD_j = \mathcal{U}([0,p_j])$ is the uniform distribution over $[0,p_j]$.

\begin{table}[tb]
    \centering
    \caption{Overview of our results for minimizing total completion times on a single machine. }
    \begin{tabular}{l|cc}
    \toprule
       Comp. ratio  & UB & LB \\
    \midrule
        weak uniform signal   & 1.752   & 1.6 \\
    strong uniform signal & 1.59999 & 1.2 \\
$k$ weak uniform signals& 
    $1+O(k^{-1/3})$ & $1+\Omega(k^{-1/3})$\\
    \bottomrule
    \end{tabular}
    \label{tab:results}
\end{table}

\paragraph{Algorithmic Results.}
We next present an overview of our results. We first focus on minimizing the total completion time with uniform stochastic clairvoyance and obtain the following results:
\begin{itemize}
    \item The algorithm $\GREEDY$, which runs \emph{Round-Robin (RR)} and completes a job whenever a signal comes, has a competitive ratio of $\nicefrac{16}{9} \approx 1.778$ on a single machine (\Cref{thm:greedy}) and on parallel identical machines (\Cref{thm:greedy:m:machines}).
    \item There exists an improved $1.752$-competitive algorithm on a single machine (\Cref{thm:lambda-greedy}).
    \item Every (randomized) algorithm has a competitive ratio of at least $1.6$ for weak uniform stochastic clairvoyance and of at least $1.2$ for the strong variant, even on a single machine.
    \item For strong uniform stochastic clairvoyance, there exists a $1.59999$-competitive algorithm on a single machine (\Cref{thm:alg-strong}).
    This separates weak and strong uniform stochastic clairvoyance.
\end{itemize}

We next consider the more general \textit{Beta} stochastic clairvoyance where $S_j = p_j X_j$ for a random variable $X_j \sim \mathrm{Beta}(a,b)$, with $a,b>0$. Note that the Beta distribution with $a=b=1$ corresponds to the uniform distribution over $[0,1]$.
\begin{itemize}
    \item For small $a,b \geq 1$, we present bounds on the competitive ratio of $\GREEDY$ below $2$ (\Cref{sec:greedy-beta}). 
   \item  For $\mathrm{Beta}(k^{2/3},k)$-stochastic clairvoyance, $\GREEDY$ has a competitive ratio of $1+ O(k^{-1/3})$ for all $k \geq 1$, and this dependence on $k$ is optimal (\Cref{thm:beta-asymptotic-ub}).
\end{itemize}
While this paper is centered on scheduling with a single stochastic signal, we briefly consider the setting where each job emits $k$ uniform signals. We apply our results on the Beta distribution to recover the $1+\Theta(k^{-1/3})$ asymptotic rate of the \textit{Poisson progress bar} in \cite{BenomarCLS25} (\Cref{thm:beta-asymptotic-ub}), albeit for a \textit{uniform progress bar}. This is further discussed in \Cref{sec:asymptotics}.

We also study the makespan objective on parallel identical machines. 
\begin{itemize}
    \item For strong uniform stochastic clairvoyance, there exists a $\nicefrac{13}{8}$-competitive algorithm
    for minimizing the makespan on parallel identical machines (\Cref{thm:makespan}).
\end{itemize}
This shows that stochastic
clairvoyance is not a problem-specific trick or curiosity of the total completion time objective.
Rather, we believe that it points to a new layer of beyond-worst-case algorithm design for scheduling problems and beyond.

\subsection{Further Related Work}

\paragraph{Nonclairvoyant scheduling.}
This online model was first studied by Shmoys, Wein, and Williamson \cite{ShmoysWW91}
and Feldmann, Sgall, and Teng~\cite{FeldmannST91}.
\cite{ShmoysWW95} show a tight competitive ratio of $2$ for makespan on parallel identical machines,
an $\Omega(\log m)$ lower bound on $m$ uniformly related machines, and a matching $O(\log m)$ upper bound on $m$ unrelated machines.
Motwani, Phillips, and Torng~\cite{MotwaniPT94} initiated the study of min-sum objectives and showed that Round-Robin is an optimal online algorithm for minimizing total completion time on parallel identical machines, with competitive ratio $2$.
For total flow time, where jobs arrive online at release dates $r_j$ and the objective is to minimize $\sum_j(C_j-r_j)$,
they gave nonconstant lower bounds for randomized ($\Omega(\log n)$) and deterministic ($\Omega(n^{1/3})$) algorithms.

After these initial works, strong competitive nonclairvoyant algorithms have been given for the weighted completion-time objective in heterogeneous scheduling environments~\cite{DengGBL00,BeaumontBEM12,ImKM18,JLM26}
and for scheduling with precedence constraints~\cite{JagerW24,GargGKS19,LassotaLMS23}.
For total flow time, an optimal $O(\log n)$-competitive randomized algorithm exists~\cite{BecchettiL04,KalyanasundaramP03}.
This objective is also well studied under $(1+\varepsilon)$-speed augmentation, where constant competitive ratios are possible~\cite{KalyanasundaramP00}, even for weighted flow time in heterogeneous settings with speedup curves; see, e.g.,~\cite{ImKMP14,ChekuriGKK04,EdmondsP12,BansalKN14}.
We refer to~\cite{PruhsST04} for a survey on online scheduling.

\paragraph{Learning-augmented algorithms.}
Integrating predictions of unknown quality into algorithm design has received significant attention in the last couple of years.
The majority of learning-augmented works consider online problems; see, e.g.,~\cite{LykourisV21,PurohitSK18,BamasMS20,AntoniadisCEPS23}.
However, the framework can also be used to speed up algorithms~\cite{DinitzILMV21,ChenSVZ22,SakaueO22a,DaviesMVW23,BaiC23,KraskaBCDP18,LinLW22,McCauleyMN023,BrandFNP24,HenzingerSSY24}, improve approximation algorithms~\cite{CohenAdded24approximation}, or design mechanisms~\cite{AgrawalBGOT22,GkatzelisKST22,XuL22,BalkanskiGT23}.
For an introduction and overview, we refer to~\cite{MitzenmacherV22,alps}.

Scheduling problems are among the most studied problems in the learning-augmented framework.
Nonclairvoyant scheduling with predictions has been studied for minimizing total completion time~\cite{PurohitSK18,WeiZ20,ImKQP23,BampisDKLP22,EliasKMM24,DBLP:conf/nips/BenomarP23,LindermayrM25permutation}, makespan~\cite{Zhao0Z22,BampisKLP23}, and flow time~\cite{AzarLT21,AzarLT22,ZhaoLZ22,AzarPT22flowtime,GuptaKPW26}.
Predictions have also been used to improve the competitive ratio when jobs arrive online over time with known processing times.
\cite{LattanziLMV20} initiated this line for makespan minimization with restricted assignment, which was later extended and generalized by~\cite{LavastidaM0X21,LiX21,CohenP23}.
Other scheduling problems studied with predictions include online speed scaling~\cite{BamasMRS20,AntoniadisGS22,BalkanskiP0W23}, speed-robust~\cite{BalkanskiOSW23} and speed-oblivious scheduling~\cite{LindermayrMR23}, and interval scheduling~\cite{BoyarFKL23}.

This line is orthogonal to stochastic clairvoyance: learning-augmented algorithms are usually given predictions before decisions are made and are evaluated as a function of prediction error, whereas in our model no job-size prediction is provided upfront and the only additional information is an online signal generated during processing.

\paragraph{Stochastic scheduling.}
Stochastic scheduling provides another classical way to interpolate between
clairvoyant and nonclairvoyant scheduling. In finite-batch stochastic scheduling,
all jobs are available from the beginning, but their processing times are random
variables drawn from known distributions; the goal is typically to optimize an
expected objective over nonanticipative policies in polynomial time. This model dates back at least to
the work of Rothkopf~\cite{Rothkopf1966} and was further developed for single- and
parallel-machine settings in, e.g., \cite{Weiss1980,Weber1982,Weber1986};
see also \cite{PinedoBook} for background. Stochastic scheduling is different
from stochastic clairvoyance: in stochastic scheduling, the processing times
themselves are random and their distributions are part of the model input, whereas
in stochastic clairvoyance the processing times are fixed adversarially and only
the time at which information is revealed through a signal is random. Thus, our
model randomizes the arrival of information rather than the job sizes.

\subsection{Notations and Definitions}

There are $n$ jobs with processing times that satisfy w.l.o.g.\ $p_1 \leq \ldots \leq p_n$. An algorithm
chooses over time $t$ a feasible allocation of $n$ jobs to $m$ parallel identical machines, that is, it must select a rate $y_j(t) \in [0,1]$ for each job $j$ such that $\sum_j y_j(t) \leq m$. A schedule is feasible if each job $j$ receives its processing requirement $p_j$, that is, $\int_0^\infty y_j(t) \mathrm{d}t \geq p_j$.
The elapsed time of job $j$ at time $t$ is $e_j(t) := \int_0^t y_j(t') \mathrm{d}t'$.
The completion time $C_j$ of job $j$ is the earliest time $t$ that satisfies $e_j(t) \geq p_j$. 
The expected total completion time is equal to $\Ebb[\sum_{j=1}^n C_j]$ and the expected makespan is equal to $\Ebb[\max_{j \in [n]} C_j]$, where the expectation is over the algorithm's and signal's randomness. 
For a fixed objective and algorithm, we denote the algorithm's objective value by $\alg$ and the optimum objective value by $\OPT$.
We say that an algorithm is $c$-competitive if $\Ebb[\alg] \leq c \cdot \opt$ for all instances; we call the smallest $c$ such that an algorithm is $c$-competitive its \emph{competitive ratio}.

\subsection{Technical Contribution}\label{sec:overview}

The main technical contribution of this work
is a new general-purpose theorem for total (weighted) completion time---that we dub the \textit{Global Delay Theorem}---that turns an arbitrary concave bound on (expected) pairwise job delays into a strong competitive guarantee.

\paragraph{Background: Local Delay Theorem.}
The delay decomposition is an established tool for analyzing scheduling algorithms~\cite{MotwaniPT94,LindermayrM25permutation,BenomarCLS25,JLM26,GuptaKLSY25,AlbersE20,LiuLWZ23,GongCH24,DogeasEL24}
on a single machine and on parallel identical machines. Fix two jobs $i$ and $j$.
Let $d^{\ALG}(i,j)$ be the delay that job $i$ inflicts on job $j$, i.e., the amount of work that job $i$ receives before job $j$'s completion (note that $\forall i: d^{\ALG}(i,i) = p_i$). Using $P = \sum_j p_j$, the total completion time of an algorithm can be rewritten as the sum of all pairwise delays 
\(
\ALG = P + \sum_{i<j} (d^{\ALG}(i,j) + d^{\ALG}(j,i)),
\)
so 
any bound of the form $d^{\alg}(i,j) + d^{\alg}(j,i) \leq \alpha \cdot (d^{\opt}(i,j) + d^{\opt}(j,i))$ yields an $\alpha$-competitive guarantee. By linearity of expectation, this in particular holds for expected pairwise delays. This is the essence of the classical Local Delay Theorem~\cite{MotwaniPT94}, where the key observation is that the optimal algorithm satisfies $d^{\opt}(i,j) + d^{\opt}(j,i)= \min \{p_i,p_j\}$ because it schedules jobs in order of their processing times and thus delays only longer jobs with shorter ones.
\begin{theorem}[Local Delay Theorem \cite{MotwaniPT94}]
\label{lem:weak-delay}
For any single-machine nonidling scheduling algorithm $\alg$, if for any two distinct jobs $i$ and $j$ with $p_i \leq p_j$ it holds that 
\[
\ex \big[ d^\alg(i,j) + d^{\alg}(j,i) \big] \leq \alpha \cdot p_i
\]
for some constant $\alpha \geq 1$,
then $\alg$ is $\alpha$-competitive for minimizing the total completion time.
\end{theorem}
The textbook application is showing the $2$-competitiveness of the Round-Robin algorithm, by noting that $d^{\mathrm{RR}}(i,j) + d^{\mathrm{RR}}(j,i)= 2\min \{p_i,p_j\}$ \cite{MotwaniPT94}. 
Unfortunately, the Local Delay Theorem is insufficient for showing an improved guarantee for $\GREEDY$, because we prove that it has an expected pairwise delay 
$\ex[d^\alg(i,j) + d^\alg(j,i)] = 2 p_i - \frac23  p_i^2 / p_j$ when $p_i \leq p_j$.
This expression cannot be bounded by a linear function of $p_i$ that would be better than $2 p_i$, which is tight for  $p_j \to \infty$.

\paragraph{A first improvement via amortized charging.}
We next give a high-level sketch of how we can break the factor of $2$ for $\GREEDY$, and then explain how we derive our Global Delay Theorem from this idea.
First note that due to the decomposition above, if we show that $\sum_{i<j} p_i^2 / p_j$ is at least a $\rho$-fraction of $\sum_{i<j} p_i$, then $\GREEDY$ is $(2 - \frac23\rho)$-competitive.
To this end, let $A_i = \sum_{j=i+1}^n p_i^2 / p_j$ and $O_i = \sum_{j=i+1}^n p_i = (n-i)p_i$, which we call \emph{row} $i$ of $\alg$ and $\opt$, respectively. 
We now do an \emph{amortized charging argument}. 
On a high level, we decompose each $A_i$ into a credit for row $i$ and a charge to a higher row $j_i$, denoted by $\mathrm{credit}(i)$ and $\mathrm{charge}(j_i)$, 
and then show that the total credit minus total charge is  
$\Omega(\sum_{i} O_i)$.
The crux is to find such a suitable charging scheme across rows.
We bound the first $m_i = \lfloor \frac12 (n-i) \rfloor$ terms $p_i^2/p_j$ after $i$ in $A_i$ from below by their smallest contribution, giving
\[
        A_i 
        = \sum_{j=i+1}^n \frac{p_i^2}{p_j}
        \geq \sum_{j=i+1}^{i+m_i}  \frac{p_i^2}{p_j}
        \geq  m_i \frac{p_i^2}{p_{i+m_i}}
        \ge \underbrace{2\lambda m_i p_i}_{\mathrm{credit}(i)} - \underbrace{\lambda^2 m_i p_{i+m_i}}_{\mathrm{charge}(i+m_i)}
\]
using Young's inequality with constant $\lambda \geq 0$ in the final step.
Note that $\mathrm{credit}(i) \geq \lambda O_i$ for each row $i$. Thus, it remains to argue that the total charge is small. 
The key is that only two distinct rows $i$ map their charge to row $j = i+m_i$, namely $2j-n$ and $2j-n+1$.
Thus, the total charge is
    \[
         \sum_{i=1}^n \mathrm{charge}(i+m_i) = \lambda^2 \sum_{i=1}^n m_i p_{i+m_i} 
         = \lambda^2 \sum_{j=1}^n p_{j} \sum_{i: i+m_i = j} m_i 
         \leq \lambda^2 \sum_{j=1}^n p_{j} \cdot 2(n-j) = 2\lambda^2 \sum_{i=1}^n O_i \ ,
    \]
while the total credit is at least $\lambda \sum_{i} O_i$.
Therefore, choosing $\lambda = \nicefrac14$ gives $\sum_i A_i \geq (\frac14 - 2 \frac{1}{16}) \sum_{i=1}^n O_i = \frac18 \sum_{i=1}^n O_i$, which yields a competitive bound of $(2 - \frac23 \cdot \frac18) = \nicefrac{23}{12}$, as desired.

\paragraph{New Global Delay Theorem.}
Our final argument to reach $\nicefrac{16}{9}$ for $\GREEDY$ builds  on this idea.
A key improvement is to embed the list of processing times in a nondecreasing function $f:[0,1] \to [0, \infty)$, 
where $f(x) = p_i$ for $x \in [(i-1)/n,i/n]$. 
Then, instead of a single midpoint for the sum over $\{i,\ldots,n\}$ as above, we slice $[x,1]$ into infinitesimally small slices using a variable transformation $y = x + \tau(1-x)$ and then 
perform the above argument for each  $\tau$ with a separately optimized value of $\lambda$.
Moreover, we allow an abstract normalized pairwise delay $\frac{1}{p_i} \ex[d^{\alg}(i,j)+d^{\alg}(j,i)] \leq g (p_i / p_j)$ where $g : [0,1] \to [0,\infty)$ is nonincreasing and concave (instead of only affine as for $\GREEDY$).
These ingredients yield the following key equality (\Cref{lem:f-bound}):
\begin{equation}
    \sup_{f\in\mathcal F}
    \frac{
        \int_0^1\int_x^1
        f(x) \cdot g\!\left(\frac{f(x)}{f(y)}\right)
        \,\mathrm{d}y\,\mathrm{d}x
    }{
        \int_0^1\int_x^1 f(x)\,\mathrm{d}y\,\mathrm{d}x
    }
    =
    \int_0^1 g(t^2)\,\mathrm{d}t \label{eq:key}
\end{equation}
where $\mathcal F$ is the set of all nondecreasing functions $f:[0,1] \to [0, \infty)$.
Both the upper bound and the lower bound in \eqref{eq:key} require careful real analysis and convex analysis arguments.
In the asymptotic regime where, after renormalization by $\frac{1}{n^2}$, it holds that
    \begin{align*}
        \ALG \rightarrow \int_0^1\int_x^1
            f(x) \cdot g\!\left(\frac{f(x)}{f(y)}\right)
            \,\mathrm{d}y\,\mathrm{d}x
        \qquad \text{and}
        \qquad
        \OPT\rightarrow \int_0^1\int_x^1 f(x)\,\mathrm{d}y\,\mathrm{d}x \ ,
    \end{align*}
\Cref{eq:key} immediately implies our novel Global Delay Theorem.

\begin{restatable}[Global Delay Theorem]{theorem}{thmStrongDelay}\label{lem:strong-delay}
Let $g : [0,1] \to [1,\infty)$ be nonincreasing and concave.
For any single-machine nonidling scheduling algorithm $\alg$, 
if for any jobs $i \neq j$ with $p_i \leq p_j$ it holds that 
\[
\ex \big[ d^\alg(i,j) + d^{\alg}(j,i) \big] \leq p_i \cdot g \bigg(\frac{p_i}{p_j} \bigg)
\]
then $\alg$ is $\rho$-competitive for minimizing the total completion time, where $\rho = \int_0^1 g \big( t^2 \big) \mathrm{d} t$. Furthermore, if the delay bound is achieved with equality, then the competitive ratio is tight.
\end{restatable}

To obtain this non-asymptotic statement in full generality,  
other technical steps are required
in the proof of \Cref{lem:strong-delay}. 

\paragraph{A framework for applying the Global Delay Theorem.}
We now describe a framework for applying \Cref{lem:strong-delay} whenever provided a nonincreasing and concave bound $g$ on the normalized (expected) pairwise delay. For $\GREEDY$, the exact normalized expected pairwise delay $g(x) = 2 - \frac{2}{3}x$ is already nonincreasing and affine, thus concave. 
Hence we get a competitive ratio equal to $\int_0^1 g \big( t^2 \big) \mathrm{d} t = 2 - \frac{2}{3} \cdot \frac{1}{3} = \nicefrac{16}{9}$.
In contrast to~\GREEDY, expected pairwise delays of different algorithms do not necessarily 
have the form required in~\Cref{lem:strong-delay}; in fact, the expected pairwise delays of all of our other algorithms do not. 
Therefore, we use
the following proof framework in \Cref{thm:lambda-greedy,thm:alg-strong,thm:beta-case2-U}:
\begin{enumerate}
    \item Compute the exact delays $\ex \big[ d^\alg(i,j) + d^{\alg}(j,i) \big] = p_i \cdot g(x)$ as a function of $x =\nicefrac{p_i}{p_j}$.
    \item Find a concave and nonincreasing majorant $\hat{g}$ with $\hat{g}(x) \ge g(x)$ for $x \in [0,1]$.
    \item Apply the Global Delay Theorem to $\hat{g}$.
\end{enumerate}

\section{Global Delay Theorem}\label{sec:global-delay}

We first prove our new Global Delay Theorem (\cref{lem:strong-delay}), which we use throughout the
paper. 
The main technical ingredient is the following variational integral identity over the class of nondecreasing positive functions,
which may be of independent interest.

\begin{theorem}[Integral Equality]\label{lem:f-bound}
Let $g : [0,1]\to[1,\infty)$ be nonincreasing and concave.  Let $\mathcal F$ be
the set of all nondecreasing functions $f:[0,1)\to[0,\infty)$ with
\(
    0 < D_f:=\int_0^1\int_x^1 f(x)\,\mathrm{d}y\,\mathrm{d}x < \infty .
\)
Then
\[
    \sup_{f\in\mathcal F}
    \frac{
        \int_0^1\int_x^1
        f(x) g\!\left(\frac{f(x)}{f(y)}\right)
        \,\mathrm{d}y\,\mathrm{d}x
    }{
        \int_0^1\int_x^1 f(x)\,\mathrm{d}y\,\mathrm{d}x
    }
    =
    \int_0^1 g(t^2)\,\mathrm{d}t.
\]
\end{theorem}
In the theorem above, all integrands are measurable and nonnegative because monotone real-valued functions are measurable and measurability is stable under composition.

An equivalent probabilistic version of this theorem can be written as 
$\sup_{f\in \Fcal}\Ebb\left[g\left(f(X) / f(Y)\right)\right] = \int_{0}^1g(t^2)dt$,
where the pair $(X,Y)\in[0,1]^2$ is a random variable sampled with density $p(x,y)\propto \oneb(x\leq y)f(x)$. It is worth noting that applying Jensen's inequality directly to the concave function $g$ and to the random variable ${f(X)}/{f(Y)}$ does not provide the desired bound. Instead, the proof below applies Jensen's inequality after a careful change of variables, which was motivated in \Cref{sec:overview}.

\begin{proof}[Proof of \Cref{lem:f-bound}]
We first prove the upper bound.  
Fix $f\in\mathcal F$.
Changing variables according to $y=x+\tau(1-x)$ and using Fubini gives
\begin{align}
\int_0^1\int_x^1
    f(x)g\!\left(\frac{f(x)}{f(y)}\right)
    \,\mathrm{d}y\,\mathrm{d}x \notag
&=
\int_0^1\int_0^1
    (1-x)f(x)
    g\!\left(\frac{f(x)}{f(x+\tau(1-x))}\right)
    \,\mathrm{d}x\,\mathrm{d}\tau \notag\\
&\le
D_f\int_0^1
    g\!\left(
        \frac1{D_f}\int_0^1
        (1-x)
        \frac{f(x)^2}{f(x+\tau(1-x))}
        \,\mathrm{d}x
    \right)
    \,\mathrm{d}\tau, \label{eq:f-bound-jensen}
\end{align}
where Equation \eqref{eq:f-bound-jensen} follows from Jensen's inequality applied to the concave function $g$ with
respect to the probability measure with density $(1-x)f(x)/D_f$.  

We next lower-bound the average inside $g$.  Fix $\tau\in[0,1)$; the endpoint
$\tau=1$ is irrelevant for the integral. For every $\lambda\ge0$ and every
$x\in[0,1]$, 
\(
    0\le \bigl(f(x)-\lambda f(x+\tau(1-x))\bigr)^2 ,
\)
and therefore
\[
    \frac{f(x)^2}{f(x+\tau(1-x))}
    \ge
    2\lambda f(x)-\lambda^2 f(x+\tau(1-x)).
\]
Thus
\begin{align*}
\int_0^1
    (1-x)
    \frac{f(x)^2}{f(x+\tau(1-x))}
    \,\mathrm{d}x
 &\ge
    2\lambda D_f
    -\lambda^2\int_0^1(1-x)f(x+\tau(1-x))\,\mathrm{d}x
\\
& =
    2\lambda D_f
    -\frac{\lambda^2}{(1-\tau)^2}
        \int_\tau^1(1-u)f(u)\,\mathrm{d}u
\\
& \ge
D_f\left(2\lambda-\frac{\lambda^2}{(1-\tau)^2}\right).
\end{align*}
Choosing $\lambda=(1-\tau)^2$ yields
\begin{equation}\label{eq:cauchy-s}
    \frac1{D_f}\int_0^1
    (1-x)f(x)
    \frac{f(x)}{f(x+\tau(1-x))}
    \,\mathrm{d}x
    \ge (1-\tau)^2.
\end{equation}
Since $g$ is nonincreasing, plugging \eqref{eq:cauchy-s} into
\eqref{eq:f-bound-jensen} gives
\[
\int_0^1\int_x^1
    f(x)g\!\left(\frac{f(x)}{f(y)}\right)
    \,\mathrm{d}y\,\mathrm{d}x
\le
D_f \int_0^1g\bigl((1-\tau)^2\bigr)\,\mathrm{d}\tau
=
D_f \int_0^1g(t^2)\,\mathrm{d}t .
\]
Dividing by $D_f$ proves the upper bound.

It remains to prove that the upper bound is tight.  Fix
$\gamma\in [1,2)$ and define 
\[
    f_{\gamma}(x)
    :=
    \left(\frac{1}{1-x}\right)^\gamma.
\]
Clearly, $f_\gamma:[0,1)\rightarrow(0,\infty)$ is nondecreasing and $D_{f_\gamma} = \int_0^1 (1-x)f_\gamma(x) \mathrm{d}x <\infty$. 
Observe that 
$\frac{f_\gamma(x)}{f_\gamma(x+\tau(1-x))} = \left(1-\tau\right)^\gamma$ is constant in $x$, so Jensen's inequality in Equation \eqref{eq:f-bound-jensen} is tight, and we have 
\begin{align*}
\int_0^1\int_x^1
    f_\gamma(x)g\!\left(\frac{f_\gamma(x)}{f_\gamma(y)}\right)
    \,\mathrm{d}y\,\mathrm{d}x \notag
&=
D_{f_\gamma}\int_0^1
    g\!\left(
        \frac1{D_{f_\gamma}}\int_0^1
        (1-x)
        \frac{f_\gamma(x)^2}{f_\gamma(x+\tau(1-x))}
        \,\mathrm{d}x
    \right)
    \,\mathrm{d}\tau\\
   &=  D_{f_\gamma} \int_0^1 g((1-\tau)^\gamma)\mathrm{d}\tau
\end{align*}
where the second equality follows from $(1-x)
\frac{f_\gamma(x)^2}{f_\gamma(x+\tau(1-x))} = (1-x)^{1-\gamma}(1-\tau)^\gamma$ and from $D_{f_\gamma} = \int_{0}^1(1-x)^{1-\gamma}\drm x$. The lower bound then follows by letting $\gamma\uparrow2$ and applying the monotone
convergence theorem. 
\end{proof}

We now use the Integral Equality (\Cref{lem:f-bound}) to prove the Global Delay Theorem.
\thmStrongDelay*

\begin{proof}
Relabel the jobs so that $p_1\le \cdots\le p_n$.  Since $\alg$ is
nonidling,
\[
    \sum_{j=1}^n C_j^\alg
    =
    \sum_{j=1}^n d^\alg(j,j)
    +
    \sum_{1\le i<j\le n}
    \bigl(d^\alg(i,j)+d^\alg(j,i)\bigr),
\]
and $d^\alg(j,j)=p_j$. Let $P:=\sum_{j=1}^n p_j$. By the classical delay decomposition recalled in \Cref{sec:overview},
\[
\ex[\ALG]
\le
P
+
\sum_{1\le i<j\le n}
p_i\,g\!\left(\frac{p_i}{p_j}\right) \quad \text{and}
\quad \OPT = P+\sum_{1\le i<j\le n}p_i.
\]

For $j\in[n]$, let $I_j:=((j-1)/n,j/n]$ and define the nondecreasing step
function $f:[0,1]\to[0,\infty)$ by $f(x)=p_j$ for $x\in I_j$.  Set
\[
    A:=\int_0^1\int_x^1
    f(x)g\!\left(\frac{f(x)}{f(y)}\right)
    \,\mathrm{d}y\,\mathrm{d}x \quad \text{and} \quad O:=\int_0^1\int_x^1 f(x)\,\mathrm{d}y\,\mathrm{d}x.
\]
The integrands are constant on each rectangle $I_i\times I_j$ with $i<j$,
and on the triangular half of each square $I_j\times I_j$.  Therefore
\[
    n^2A 
    =
    \sum_{1\le i<j\le n}
    p_i\,g\!\left(\frac{p_i}{p_j}\right)
    +
    \frac{g(1)}2 P  \quad \text{and}\quad 
    n^2O
    =
    \sum_{1\le i<j\le n}p_i+\frac12P
    =
    \OPT-\frac12P .
\]
By \Cref{lem:f-bound}, $A\le \rho O$, where
$\rho=\int_0^1g(t^2)\,\mathrm{d}t$.  Hence
\begin{align}
\ex[\ALG]
\leq
n^2 A+\left(1-\frac{g(1)}2\right)P
&\le
\rho\left(\OPT-\frac12P\right)
+\left(1-\frac{g(1)}2\right)P \label{eq:global:delay:slack}\\
&=
\rho\,\OPT
+\frac{P}{2}\bigl(2-g(1)-\rho\bigr). \nonumber
\end{align}
Finally, $g(1)\ge1$ and, because $g$ is nonincreasing,
$\rho=\int_0^1g(t^2)\,\mathrm{d}t\ge g(1)$.  Thus
$2-g(1)-\rho\le0$, and $\ex[\ALG]\le\rho\,\OPT$.

We complete the proof by arguing that the theorem is tight when the delay bounds are achieved with equality. Let $f:[0,1)\rightarrow[0,\infty)$ be a nondecreasing function with $D_f = \int_{0}^1\int_x^1 f(x)\,\mathrm{d}y\,\mathrm{d}x = \int_{0}^1(1-x)f(x)\,\mathrm{d}x<\infty$ and consider the setting where there are $n$ jobs with processing times $p_j = f(j/n)$ for $j\in\{0,\dots,n-1\}$. We will show that $\frac{1}{n^2}P = \frac{1}{n^2}\sum_{j=0}^{n-1}p_j \xrightarrow{n\rightarrow\infty} 0$, implying that, if the delay bounds are achieved with equality,
    \[
    \frac{\Ebb[\ALG]}{\OPT}\xrightarrow{n\rightarrow\infty} \frac{\int_0^1\int_x^1
    f(x)g\!\left(\frac{f(x)}{f(y)}\right)\,\mathrm{d}y\,\mathrm{d}x}{\int_0^1\int_x^1 f(x)\,\mathrm{d}y\,\mathrm{d}x},
    \]
    and the lower bound then follows from the lower bound in \Cref{lem:f-bound}. Let us now argue that $\frac{1}{n^2}P$ converges to $0$. Fix $\delta \in(0,1)$ and split the sum $\frac{1}{n^2}P$ into the terms with $j/n\leq 1-\delta$ and those with $j/n>1-\delta$. The first term satisfies $\frac{1}{n^2}\sum_{j\leq (1-\delta)n}f(j/n) \le \frac{1}{n}f(1-\delta) \xrightarrow{n\rightarrow\infty} 0$. We bound the other term by $2\int_{1-\delta}^1(1-x)f(x)dx$ because for any $j\in\{0,\dots, n-1\} : \frac{1}{n^2}f(j/n)\leq 2\int_{j/n}^{(j+1)/n} (1-x)f(x)dx$. Thus, for any $\delta \in(0,1)$ we have $\limsup_{n\rightarrow\infty} \frac{1}{n^2}P \leq 2\int_{1-\delta}^1(1-x)f(x)dx$. We conclude that the limit must be $0$ by letting $\delta \downarrow 0$. 
\end{proof}

Finally, we note that the Global Delay Theorem naturally extends 
to the objective of the sum of weighted completion times 
and weighted pairwise delays. We give details in \Cref{app:delay-theorem}.

\section{Scheduling on a Single Machine}\label{sec:single-machine}

\subsection{The $\GREEDY$ Algorithm}\label{sec:greedy-alg}
The goal of this section is to prove the following competitiveness guarantee for \Cref{alg:greedy}. 
\begin{algorithm}[H]
\DontPrintSemicolon
\caption{$\GREEDY$}
\label{alg:greedy}
    Run Round-Robin (RR) \;
    Whenever a job $j$ emits its signal, stop RR, run $j$ to completion, and then resume RR. \;
\end{algorithm}

\begin{restatable}{theorem}{thmGreedy}\label{thm:greedy}
	The competitive ratio of $\GREEDY$ (\Cref{alg:greedy}) under uniform stochastic clairvoyance for minimizing the total completion time is equal to $\nicefrac{16}{9} \approx 1.778$.
\end{restatable}

We start with the following general statement on pairwise delays.

\begin{restatable}[Pairwise delay of $\GREEDY$]{lemma}{pairwisedelaygreedy}
\label{lem:pairwise-delay-greedy}
Consider single-machine scheduling in which job $j\in [n]$ emits a signal at time $S_j = p_jX_j$ where $(X_j)_{j\in [n]}$ are i.i.d.\ random variables on $[0,1]$ with density $f:[0,1]\rightarrow \Rbb$ and cumulative distribution function $F:[0,1]\rightarrow[0,1]$.
Then, the expected pairwise delay of $\GREEDY$ satisfies, for $i,j$ such that $p_i \leq p_j$,
    $\ex[d^{\GREEDY}(i,j)+d^{\GREEDY}(j,i)] = p_i \cdot g(p_i/p_j)$
    where, for $r\in (0,1]$,
    \[
    g(r) = 1 + \int_0^1 \bigl(1-F(t)\bigr)\bigl(1-F(r t)\bigr)\,\mathrm dt
    + \left(\frac{1}{r} - 1\right)\int_0^1 F(r t)\,f(t)\,\mathrm dt.
    \]
\end{restatable}

\begin{proof}
Fix two jobs $i$ and $j$ with $p_i\le p_j$, and write
$x:=\frac{p_i}{p_j}\in (0,1]$.
Let
$D_{ij}:=d^{\textsc{Greedy}}(i,j)+d^{\textsc{Greedy}}(j,i)$.

Until one of the two jobs emits its signal, \textsc{Greedy} processes both jobs
only through the Round-Robin phase. Other jobs may interrupt the execution by
emitting signals and being run to completion, but such interruptions only pause
both $i$ and $j$; they do not change the relative elapsed processing of
$i$ and $j$. Hence the first among $i$ and $j$ to emit a signal is
exactly the one with the smaller signal threshold $S_i$ or $S_j$.

If $S_i\le S_j$, then job $i$ emits a signal first and \textsc{Greedy}
immediately runs $i$ to completion. At that time, job $j$ has received
exactly $S_i$ units of processing. Therefore 
$d^{\textsc{Greedy}}(i,j)=p_i$ and
$d^{\textsc{Greedy}}(j,i)=S_i$,
and hence
\[
D_{ij}=p_i+S_i.
\]
Similarly, if $S_j<S_i$, then job $j$ emits a signal first and is completed
immediately. At the completion time of $j$, job $i$ has received exactly
$S_j$ units of processing. Therefore
$d^{\textsc{Greedy}}(i,j)=S_j$ and
$d^{\textsc{Greedy}}(j,i)=p_j$,
and hence
\[
D_{ij}=p_j+S_j.
\]
Since the signal distribution has a density, ties occur with probability zero.

The two cases can be combined into the single identity
\[
D_{ij} = p_i + \min\{S_i, S_j\} + (p_j - p_i)\,\mathbf{1}_{\{S_j < S_i\}},
\]
which holds pointwise: when $S_i \le S_j$ it evaluates to $p_i + S_i$, and when $S_j < S_i$ it evaluates to $p_i + S_j + (p_j - p_i) = p_j + S_j$. Taking expectations,
\[
\mathbb{E}[D_{ij}]
= p_i + \mathbb{E}\!\left[\min\{S_i, S_j\}\right] + (p_j - p_i)\,\Pr[S_j < S_i].
\]
Dividing by $p_i$ and substituting $S_k = p_k X_k$ for $k \in \{i,j\}$, so that $S_i/p_i = X_i$, $S_j/p_i = X_j/x$, and $\{S_j < S_i\} = \{X_j < x X_i\}$, we obtain
\[
\frac{\mathbb{E}[D_{ij}]}{p_i}
= 1 + \mathbb{E}\!\left[\min\!\left\{X_i, \tfrac{1}{x} X_j\right\}\right]
+ \left(\frac{1}{x} - 1\right)\Pr[X_j < x X_i].
\]
Since $\min\{X_i, \tfrac{1}{x}X_j\} \in [0,1]$ is nonnegative, the tail-integral formula together with the independence of $X_i$ and $X_j$ gives
\[
\mathbb{E}\!\left[\min\!\left\{X_i, \tfrac{1}{x} X_j\right\}\right]
= \int_0^1 \Pr\!\left(X_i > t \text{ and } X_j > x t\right)\,\mathrm dt
= \int_0^1 \bigl(1 - F(t)\bigr)\bigl(1 - F(x t)\bigr)\,\mathrm dt .
\]
Conditioning on $X_i$ and using independence,
\[
\Pr[X_j < x X_i]
= \int_0^1 \Pr(X_j < x t \mid X_i = t)\, f(t)\,\mathrm dt
= \int_0^1 F(x t)\, f(t)\,\mathrm dt .
\]
Combining the last three displays yields
\[
\frac{\mathbb{E}[D_{ij}]}{p_i}
= 1 + \int_0^1 \bigl(1 - F(t)\bigr)\bigl(1 - F(x t)\bigr)\,\mathrm dt
+ \left(\frac{1}{x} - 1\right)\int_0^1 F(x t)\, f(t)\,\mathrm dt ,
\]
which is exactly $g(x)$. Hence $\mathbb{E}[D_{ij}] = p_i g(x)$, as claimed.
\end{proof}

We now apply this result to analyze $\GREEDY$ under uniform stochastic clairvoyance, i.e., $S_j \sim \mathcal U([0,p_j])$ for each job $j \in [n]$.

\begin{restatable}{lemma}{greedydelayuniform}
\label{lem:greedy-delay}
Under uniform stochastic clairvoyance, for any two jobs $i,j$ with $p_i \leq p_j$, it holds that
\[
\ex\!\left[d^{\GREEDY}(i,j)+d^{\GREEDY}(j,i)\right]
=
p_i\left(2 - \frac{2}{3}\frac{p_i}{p_j}\right).
\]
\end{restatable}

\begin{proof}
Fix two jobs $i,j$ with $p_i \le p_j$ and write $x := p_i/p_j \in (0,1]$. By \Cref{lem:pairwise-delay-greedy},
\[
\ex\!\left[d^{\GREEDY}(i,j)+d^{\GREEDY}(j,i)\right] = p_i\, g(x),\]
\[
g(x) = 1 + \int_0^1 (1-F(t))(1-F(xt))\,\mathrm dt + \left(\frac{1}{x}-1\right)\int_0^1 F(xt)\,f(t)\,\mathrm dt,
\]
where $F(t)=t$ and $f(t)=1$ for $t \in [0,1]$ are the CDF and PDF of $X_j = S_j/p_j$. The two integrals evaluate to
\[
\int_0^1 (1-t)(1-xt)\,\mathrm dt = \int_0^1 \bigl(1 - (1+x)t + x t^2\bigr)\,\mathrm dt = \frac{1}{2} - \frac{x}{6}
\]
and
\[
\int_0^1 F(xt)\,f(t)\,\mathrm dt = \int_0^1 x t\,\mathrm dt = \frac{x}{2}.
\]
Substituting these into $g$,
\[
g(x) = 1 + \left(\frac{1}{2} - \frac{x}{6}\right) + \left(\frac{1}{x}-1\right)\frac{x}{2}
= 1 + \frac{1}{2} - \frac{x}{6} + \frac{1}{2} - \frac{x}{2}
= 2 - \frac{2}{3}x.
\]
Therefore
\[
\ex\!\left[d^{\GREEDY}(i,j)+d^{\GREEDY}(j,i)\right] = p_i\left(2 - \frac{2}{3}\,\frac{p_i}{p_j}\right),
\]
which concludes the proof of the lemma.
\end{proof}

\begin{proof}[Proof of \Cref{thm:greedy}.]
Specifically, we use \Cref{lem:greedy-delay} and the Global Delay Theorem (\Cref{lem:strong-delay}) with $g(x) = 2-\frac{2}{3}x$ to obtain that the competitive ratio of $\GREEDY$ is exactly equal to $\rho$, where
\[
\rho = \int_0^1 g(t^2)\,\mathrm{d}t
= \int_0^1 \left(2-\frac{2}{3}t^2\right)\,\mathrm{d}t
= 2 - \frac{2}{3} \int_0^1 t^2 \,\mathrm{d}t
= 2 - \frac{2}{3} \cdot \frac{1}{3}
= \frac{16}{9}.\qedhere
\]
\end{proof}

Finally, we note that $\GREEDY$ has the same competitive ratio for the objective of minimizing the sum of weighted completion times, where we replace RR in \Cref{alg:greedy} with its weighted analogue.

\subsection{The $\LGREEDY$ Algorithm}\label{sec:lgreedy}

We next study a variant of $\GREEDY$,
presented in \Cref{alg:lambda:greedy}, called $\LGREEDY$. 
A bad case for $\GREEDY$ is the event when a very long job $j$ has a very early signal and thus delays much shorter jobs $i$. This regime, where $x = p_i/p_j$ is small, contributes the highest area to the competitive ratio.
To better handle it, we add the following safety: if a job that emits a signal $S_j$ is not finished after receiving an additional processing time of $\frac{1-\lambda}{\lambda} \cdot S_j$, we resume Round Robin (RR).

\begin{algorithm}[H]
\DontPrintSemicolon
\caption{$\LGREEDY$ for $\lambda \in (0,1]$}
\label{alg:lambda:greedy}
Run Round-Robin (RR) \;
    Whenever a job $j$ emits its signal at time $t$ with $e_j(t) = S_j$, stop RR, run $j$ for $\frac{1-\lambda}{\lambda} S_j$ time units (or until it completes), and then resume RR. \;
\end{algorithm}

When optimizing $\lambda$, this variant of $\GREEDY$ provides a slightly improved competitive ratio under uniform stochastic clairvoyance. 
Specifically, we will show the following result.
The result relies heavily on the Global Delay Theorem: in contrast, numerical computations indicate that the strongest bound obtainable via the Local Delay Theorem is approximately $1.83$ (cf.~\Cref{app:lambda-greedy}).

\begin{theorem}\label{thm:lambda-greedy}
    $\LGREEDY$ (\Cref{alg:lambda:greedy}) with $\lambda = 0.038$ is $1.752$-competitive for minimizing the total completion time on a single machine with uniform stochastic clairvoyance.
\end{theorem}

We first analyze the expected pairwise delay of $\LGREEDY$, whose proof is deferred to \Cref{app:lambda-greedy}.

\begin{restatable}{lemma}{lemRobustDelays}\label{lem:robust-delays}
    Consider two jobs $i$ and $j$ with $p_i \le p_j$. Then, for any $\lambda>0$, $\ex[d^{\LGREEDY}(i,j)+d^{\LGREEDY}(j,i)] \leq p_i g_\lambda\left(\frac{p_i}{p_j}\right)$, where $g_\lambda  : [0,1]\rightarrow \Rbb_+$ satisfies
\begin{itemize}
        \item[] if $x \le \lambda$, then  $g_\lambda(x) = \frac{\lambda+3}{2}  + \frac{(1-\lambda)(1+\lambda+\lambda^2)}{6\lambda} x$,
        \item[] if $\lambda \le x$, then 
        $g_\lambda(x) = \frac{-\lambda^4+2\lambda^3-\lambda^2+2\lambda-4}{6(\lambda-1)^2} x
        + \frac{3(\lambda^3-3\lambda+4)}{6(\lambda-1)^2}
        + \frac{3\lambda(\lambda^2-2\lambda-1)}{6(\lambda-1)^2} x^{-1}
        + \frac{(-\lambda^4+2\lambda^3+\lambda^2)}{6(\lambda-1)^2} x^{-2}$.
\end{itemize}
\end{restatable}

\begin{proof}[We are now ready to prove~\Cref{thm:lambda-greedy}]
The normalized expected pairwise delay is bounded by the function $g_\lambda:[0,1]\rightarrow \Rbb_+$ defined in \Cref{lem:robust-delays} and represented in \Cref{fig:lambda-greedy}. The function $g_\lambda$ is not monotone, and we denote its maximizer in $[0,1]$ by $
x^*_\lambda = \argmax_{x\in[0,1]} g_\lambda(x),$
which satisfies $x^*_\lambda \geq \lambda$ (cf.~\Cref{lem:lambda:greedy:maximizer:existence}). Since $g_\lambda$ is nonincreasing and concave for $x\ge x^*_\lambda$ (cf.~\Cref{lem:lambda:greedy:gconcave}), we define 
\begin{equation}\label{eq:ghat}
    \hat g_\lambda(x) = \begin{cases}
    g_\lambda(x^*_\lambda)\qquad \text{ if } x \le x^*_\lambda,\\
    g_\lambda(x) \qquad \text{ otherwise}, 
\end{cases}
\end{equation}
which is nonincreasing, concave and satisfies $\forall x\in [0,1] : g_\lambda(x) \leq \hat g_\lambda(x)$ (cf.~\Cref{lem:lambda:greedy:ghat}). 
Hence, by~\Cref{lem:strong-delay}, the competitive ratio $\rho_\lambda$ of $\LGREEDY$ satisfies $\rho_\lambda \le \int_0^1 \hat g_{\lambda}(t^2)\,dt.$
Integral computations then imply that for $\lambda = 0.038$ it holds that $\rho_\lambda < 1.752$. 
\end{proof}

\begin{figure}
	\centering
	\includegraphics[width=0.8\textwidth]{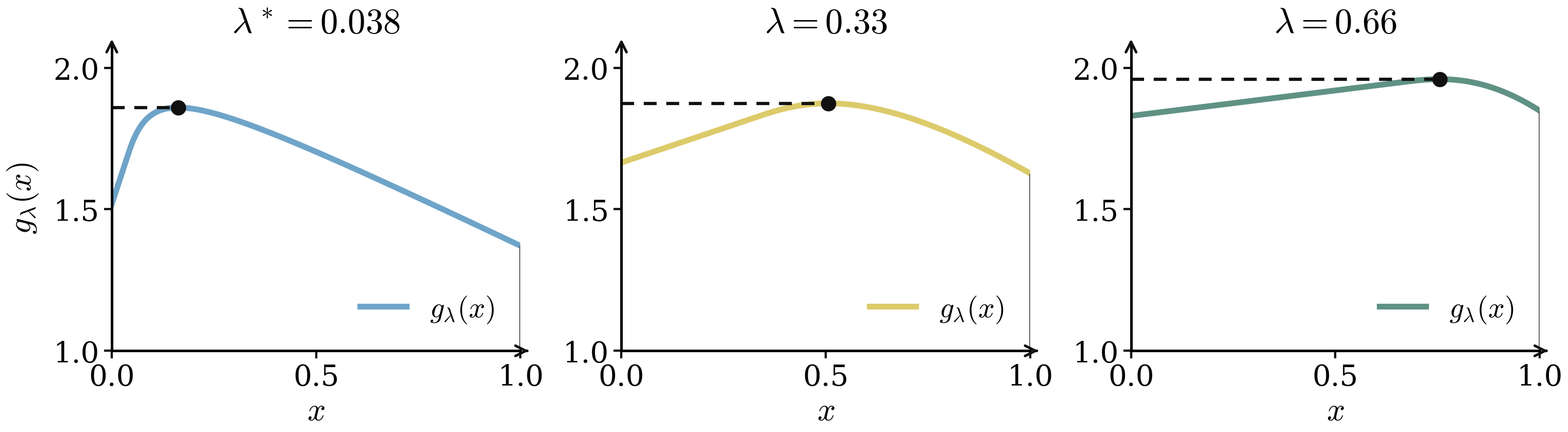}
	\caption{Visualization of the normalized pairwise expected delay $g_\lambda(x)$ of $\LGREEDY$. The black dashed line indicates the first part of the constructed majorant $\hat g_\lambda(x)$.}
	\label{fig:lambda-greedy}
\end{figure}

\subsection{The $\LBOOSTING$ Algorithm for Strong Stochastic Clairvoyance}\label{sec:strong-sto}
In this section, we consider the \textit{strong uniform stochastic clairvoyance} model in which $p_j$ is revealed to the scheduler when job $j$ reaches an elapsed time of $S_j$. Let $K(t) = \{j\in [n] : p_j > e_j(t) \ge S_j\}$ and $U(t)  = \{j\in [n] : e_j(t) < S_j\}$ denote the sets of \textit{known} and \textit{unknown} jobs at time $t$.
We consider the following algorithm, again parametrized by some $\lambda \in (0,1]$. The algorithm runs Round-Robin on the unknown jobs until the elapsed time of the unknown jobs exceeds $\lambda$ times the processing time of the shortest remaining known job; at that point, the shortest known job is run to completion. The pseudocode is given below in \Cref{alg:strong}, where we use the convention that $\min (\emptyset) = \infty$.

\begin{algorithm}[H]
\DontPrintSemicolon
\caption{$\LBOOSTING$ for $\lambda \in (0,1]$}
\label{alg:strong}
    \textbf{If} $\lambda \cdot  \min_{j \in K(t)} p_j \ge \min_{j \in U(t)} e_j(t)$, \textbf{then} run Round-Robin (RR) over the jobs in $U(t)$ \;
    \textbf{Else} run an uncompleted job $j\in K(t)$ with minimum $p_j$ to completion (job $j$ is \emph{boosted}) 
\end{algorithm}

Using the same proof framework as in the previous section and building a nonincreasing
concave majorant of the normalized expected pairwise delay of $\LBOOSTING$,
we show in \Cref{app:strong-signals} the following result. This result relies on the Global Delay Theorem. Numerical computations indicate that the strongest bound achievable via the Local Delay Theorem is approximately $1.61$ (cf.~\Cref{app:strong-signals}).

\begin{theorem}\label{thm:alg-strong}
    $\LBOOSTING$ (\Cref{alg:strong}) with $\lambda = 0.4699$ is $1.59999$-competitive for minimizing the total completion time on a single machine with strong uniform stochastic clairvoyance.
\end{theorem}

\section{Lower Bounds}\label{sec:lower-bounds}
Our lower bounds rely on Yao's minimax principle applied to a randomized problem instance where processing times are sampled following a Gamma distribution, i.e., $P_j\sim \Gamma(2,\lambda)$, for $\lambda>0$. A key technical observation is that, for a uniform signal $S_j\sim \Ucal([0,P_j])$, the variables $S_j$ and $P_j-S_j$ are independent, and both follow the (memoryless) exponential distribution $\Ecal(\lambda)$. This observation allows us to explicitly determine optimal online algorithms for that randomized instance. We summarize our proof strategy in the proposition below, and defer some formal definitions and intermediate lemmas to \Cref{sec:lb-support}.
\begin{theorem}
For the problem of minimizing the total completion time with one uniform random signal, assuming the processing times of the problem instance are sampled from independent Gamma random variables $P_j \sim \Gamma(2,\lambda)$ for some $\lambda>0$, we have
\begin{align}
    \mathbb{E}[\OPT] &= \frac{2n}{\lambda}+\frac{5n(n-1)}{8\lambda},\label{eq:opt}\\
    \mathbb{E}[\ALG] &\geq \frac{n(n+1)}{\lambda}\quad \text{for any online algorithm } \ALG \text{ using weak signals}\label{eq:alg}\\
    \mathbb{E}[\ALG'] &\geq \frac{n(n+3)}{4\lambda}+\frac{n(n+1)}{2\lambda}\quad \text{for any online algorithm } \ALG' \text{ using strong signals}.\label{eq:alg-strong}
\end{align}
Thus, by applying Yao's minimax principle, the competitive ratio of any (randomized) scheduling algorithm is at least $8/5 = 1.6$ for weak uniform signals and $6/5 = 1.2$ for strong uniform signals. 
\end{theorem}
\begin{proof}
To compute the expected value of the optimal offline algorithm $\OPT$, we use 
$\OPT = \sum_{i=1}^n P_i + \sum_{1\leq i<j\leq n}\min\{P_i,P_j\}$.
Equation \eqref{eq:opt} then directly follows from the fact that for a $\Gamma(2,\lambda)$ random variable $\mathbb{E}(P_i) = \frac{2}{\lambda}$ and for a pair $(P_i,P_j)$ of independent $\Gamma(2,\lambda)$ random variables, $\mathbb{E}[\min \{P_i,P_j\}]= \frac{5}{4\lambda}$ (see \Cref{sec:basic-gamma}). 

We then study online algorithms $\ALG$ that use weak signals and aim to prove Equation \eqref{eq:alg}. For this, we observe that any algorithm prioritizing jobs that have emitted a signal over those that haven't minimizes the expected total completion time (see \Cref{lem:opt-online}). In particular, $\GREEDY$ is optimal, as well as the algorithm that completes jobs in increasing order of $j\in[n]$. Because of the optimality of the latter, we obtain:
\(
    \mathbb{E}(\ALG) \geq \mathbb{E}\left[\sum_{i=1}^n \sum_{j=1}^i P_j\right] = \frac{n(n+1)}{\lambda}.
\)

We then study online algorithms $\ALG'$ that use strong signals and aim to prove Equation \eqref{eq:alg-strong}. We consider an even stronger setting in which all values $S_1,\dots, S_n$ are initially known to the algorithm. In this \textit{revealed signals} setting, the algorithm that minimizes the expected total completion time schedules jobs $j\in [n]$ in increasing order of $S_j$ until termination.  
Specifically, letting $(j)\in[n]$ denote the index of the $j$-th signal, so that $S_{(1)}\leq \dots\leq S_{(n)}$, we have
\begin{align*}
    \ALG' &\geq \sum_{i=1}^n \sum_{j=1}^iP_{(j)} = \sum_{i=1}^n \sum_{j=1}^i S_{(j)}  + \sum_{i=1}^n \sum_{j=1}^i \bigl(P_{(j)}-S_{(j)}\bigr).
\end{align*}
Taking expectations and using $\Ebb(S_{(j)}) = \frac{1}{\lambda}(H_n-H_{n-j})$ and $\Ebb(P_{(j)}-S_{(j)}) = \frac{1}{\lambda}$, where $H_j$ denotes the $j$-th harmonic number (see \Cref{sec:basic-gamma} for details on order statistics), we get
\begin{align*}
    \Ebb[\ALG'] \geq \frac{n(n+3)}{4\lambda}+\frac{n(n+1)}{2\lambda} \sim \frac{3}{4}\frac{n^2}{\lambda}.
\end{align*}
We finally argue that Yao's minimax principle applies in this context. Any online algorithm $\ALG$ using weak signals satisfies
\[
\mathbb{E}_{(P_j)}[\mathbb{E}_{\ALG} [\ALG]] \geq \left(\frac{8}{5}-\frac{48}{25n+55}\right)\mathbb{E}_{(P_j)}[\OPT],
\]
where we distinguish the randomness of the instance $\mathbb{E}_{(P_j)}(\cdot)$ from that of the random seed $\mathbb{E}_{\ALG}(\cdot)$. Thus there must exist some deterministic choice of instance $(p_j)_{j\in [n]}$ satisfying 
\[
\mathbb{E}_{\ALG} [\ALG] \geq \left(\frac{8}{5}-\frac{48}{25n+55}\right)\OPT,
\]
    which completes the lower bound of $\nicefrac{8}{5}$ for weak signals, by letting $n\rightarrow \infty$. The same application of Yao's principle provides the $\nicefrac{6}{5}$ lower bound for strong signals. 
\end{proof}

\section{Scheduling on Parallel Identical Machines}\label{sec:parallel-machines}

We consider preemptive total completion time minimization on $m$ parallel identical machines with weak uniform stochastic signals. 
The goal of this section is to analyze the natural generalization of \GREEDY~from~\Cref{sec:greedy-alg} as formalized by~\Cref{alg:parallel:greedy}. We remark that Round-Robin on $k < m$ jobs can use at most $k$ machines.

\begin{algorithm}
\DontPrintSemicolon
\caption{$\PGREEDY$}
\label{alg:parallel:greedy}
    Whenever a job $j$ emits its signal, run only $j$ to completion on one of the machines. \;
    Run Round-Robin over all other remaining jobs on all machines not busy with step 1. 
\end{algorithm}

\begin{restatable}{theorem}{thmGreedyMMachines}
    \label{thm:greedy:m:machines}
    $\PGREEDY$ (\Cref{alg:parallel:greedy}) is $\nicefrac{16}{9}$-competitive for total completion time minimization on parallel identical machines under uniform stochastic clairvoyance.
\end{restatable}

Fix processing times $p_1 \le \ldots \le p_n$, and let $\opt_m$ denote the optimal objective value on $m$ machines. We assume without loss of generality that $n=h \cdot m$ for some integer $h$; otherwise, we add dummy jobs of size zero, which change neither the algorithm's objective value nor the optimum. Let $P:=\sum_{i=1}^n p_i$.

\paragraph{Relating the one-machine and $m$-machine optima.}
The basic lower bound $\opt_m \ge \opt_1/m$ is not by itself sufficient for the analysis, since a job's expected completion time under $\PGREEDY$ can exceed its one-machine completion time divided by $m$. The following exact expression quantifies the additional slack in $\opt_m$.

\begin{restatable}{lemma}{lemMOpt}
    \label{lem:m:opt}
    It holds that
    \[
        \opt_m = \frac{\opt_1}{m} + \sum_{q=0}^{h-1} \sum_{r=1}^m \frac{r-1}{m} p_{qm+r}.
    \]
\end{restatable}

\begin{proof}
For a job $i$, define
$c_i:=\frac{n}{m}-\left\lfloor\frac{i-1}{m}\right\rfloor$
and $c_i':=\frac{n-i+1}{m}$.
The shortest-processing-time order is optimal for total completion time, and therefore
$\opt_m = \sum_{i=1}^n c_i p_i$ and $\frac{\opt_1}{m} = \sum_{i=1}^n c_i' p_i$.

Fix an index $i$ with $i \bmod m=1$. Since $c_{i+r}=c_i$ and $c_{i+r}'=c_i-r/m$ for $r\in\{0,\ldots,m-1\}$,
\begin{align*}
\sum_{r=0}^{m-1} (c_{i+r}-c_{i+r}')p_{i+r}
&= \sum_{r=0}^{m-1} \frac{r}{m}p_{i+r}.
\end{align*}
Summing over the $h$ consecutive blocks proves
\[
\opt_m-\frac{\opt_1}{m}
=\sum_{q=0}^{h-1}\sum_{r=1}^m \frac{r-1}{m}p_{qm+r}.
\]
\end{proof}

\paragraph{Bounding completion times under $\PGREEDY$.}
For a fixed realization of the signals, let $Y_i^k$ indicate the event that job $i$ is the $(n-k)$th job to emit its signal. We first account explicitly for the possibility that some machines idle before $i$ emits.

\begin{lemma}
    \label{lem:para:machines:greedy:completion}
    Fix a job $i$, and let
    $A:=\{j\in J\setminus\{i\}:S_j<S_i\}$
    and 
    $B:=J\setminus(\{i\}\cup A)$.
    Then
    \[
    C_i \le \frac{\sum_{j\in A}p_j+\sum_{j\in B}S_i}{m}
    +\frac{S_i}{m}+(p_i-S_i)
    +\sum_{k=0}^{m-2}Y_i^k\frac{m-k-1}{m}S_i.
    \]
\end{lemma}

\begin{proof}
The volume processed before $i$ starts its speed-up is at most
\[
V:=S_i+\sum_{j\in A}p_j+\sum_{j\in B}S_i.
\]
Indeed, all jobs in $B\cup\{i\}$ receive $S_i$ units of processing before $i$ starts its speed-up, while every job in $A$ may already have completed. If up to $m-1$ jobs in $A$ have emitted but not yet completed, the processed volume is only smaller.

Let $T_i$ be the elapsed processing time of job $i$ when the first machine starts to idle. We distinguish two cases.
\begin{enumerate}
    \item If $T_i\ge S_i$, then no machine idles before $i$ emits. Thus, $i$ starts its speed-up by time $V/m$ and subsequently needs $p_i-S_i$ time to complete. Hence
    \begin{align*}
    C_i
    &\le \frac{\sum_{j\in A}p_j+\sum_{j\in B}S_i}{m}
    +\frac{S_i}{m}+(p_i-S_i)\\
    &\le \frac{\sum_{j\in A}p_j+\sum_{j\in B}S_i}{m}
    +\frac{S_i}{m}+(p_i-S_i)
    +\sum_{k=0}^{m-2}Y_i^k\frac{m-k-1}{m}S_i.
    \end{align*}

    \item If $T_i<S_i$, then $i$ is among the last $m-1$ jobs to emit, so $\sum_{k=0}^{m-2}Y_i^k=1$. Let $k'\in\{0,\ldots,m-2\}$ satisfy $Y_i^{k'}=1$. There are exactly $k'$ jobs that emit after $i$, so $|B|=k'$, and at most $m-k'-1$ machines idle before $i$ starts its speed-up. The volume processed before the first machine idles is at most
    \[
        T_i+\sum_{j\in A}p_j+\sum_{j\in B}T_i.
    \]
    After that time, job $i$ needs another $p_i-T_i$ time units. Therefore
    \begin{align*}
    C_i
    &\le \frac{\sum_{j\in A}p_j+\sum_{j\in B}T_i}{m}
        +\frac{T_i}{m}+(p_i-T_i)\\
    &=\frac{\sum_{j\in A}p_j+\sum_{j\in B}S_i}{m}
        +\frac{S_i}{m}+(p_i-S_i)
        +\frac{m-|B|-1}{m}(S_i-T_i)\\
    &=\frac{\sum_{j\in A}p_j+\sum_{j\in B}S_i}{m}
        +\frac{S_i}{m}+(p_i-S_i)
        +\sum_{k=0}^{m-2}Y_i^k\frac{m-k-1}{m}(S_i-T_i)\\
    &\le \frac{\sum_{j\in A}p_j+\sum_{j\in B}S_i}{m}
        +\frac{S_i}{m}+(p_i-S_i)
        +\sum_{k=0}^{m-2}Y_i^k\frac{m-k-1}{m}S_i.
    \end{align*}
\end{enumerate}
\end{proof}

Let $d_1(j,i)$ denote the delay that job $j$ causes to job $i$ when $\GREEDY$ runs on a single machine. Taking expectations in the preceding realization-wise bound gives the following comparison.

\begin{restatable}{lemma}{lemMMachinesExpCompletion}
    \label{lem:m:machines:exp:completion}
    For every job $i$,
    \[
    \ex[C_i]\le \frac{m-1}{2m}p_i
    +\frac{1}{m}\sum_{j=1}^n\ex[d_1(j,i)]
    +\sum_{k=0}^{m-2}\frac{m-k-1}{m}\ex[Y_i^kS_i].
    \]
\end{restatable}

\begin{proof}
Fix a job $i$ and define, for every $j\ne i$,
\[
X_j:=
\begin{cases}
    p_j, & S_j<S_i,\\
    S_i, & S_j\ge S_i.
\end{cases}
\]
By~\Cref{lem:para:machines:greedy:completion},
\[
C_i\le \frac{1}{m}\sum_{j\ne i}X_j+\frac{S_i}{m}+(p_i-S_i)
+\sum_{k=0}^{m-2}\frac{m-k-1}{m}Y_i^kS_i.
\]
Moreover,
\[
\frac{\ex[S_i]}{m}=\frac{\ex[d_1(i,i)]}{m}-\frac{p_i}{2m},
\qquad
\ex[p_i-S_i]=\frac{p_i}{2},
\qquad
\ex[X_j]=\ex[d_1(j,i)].
\]
Taking expectations and combining these identities proves the claim.
\end{proof}

\paragraph{Bounding the idle-machine correction.}
The remaining correction term is controlled by the same slack that separates $\opt_m$ from $\opt_1/m$.

\begin{lemma}\label{lem:parallel-final-sum}
We have
\[
\sum_{i=1}^n\sum_{k=0}^{m-2}\frac{m-k-1}{m}\ex[Y_i^kS_i]
\le \frac{2}{3}\sum_{q=0}^{h-1}\sum_{r=1}^m\frac{r-1}{m}p_{qm+r}.
\]
\end{lemma}

\begin{proof}
For every realization, the indicators $Y_i^k$ select the last $m-1$ jobs to emit. Matching the weights $(m-1)/m,\ldots,1/m$ to the corresponding signal order statistics $S_{(1)}\le\ldots\le S_{(n)}$ gives
\[
\sum_{i=1}^n\sum_{k=0}^{m-2}\frac{m-k-1}{m}\ex[Y_i^kS_i]
\le \frac{1}{m}\ex\left[\sum_{r=1}^{m-1}rS_{(n-m+1+r)}\right].
\]
Partition $S=(S_1,\ldots,S_n)$ into $h$ vectors $S^0,\ldots,S^{h-1}$, where $S^q=(S_{qm+1},\ldots,S_{qm+m})$. Then
\begin{align*}
\sum_{r=1}^{m-1}rS_{(n-m+1+r)}
&\le \sum_{q=0}^{h-1}\sum_{r=1}^{m-1}rS^q_{(r+1)}\\
&=\sum_{q=0}^{h-1}\sum_{r=1}^m\sum_{\ell=1}^{r-1}\max\{S_r^q,S_\ell^q\}.
\end{align*}
For $\ell<r$, the ordering of the processing times gives $p_{qm+\ell}\le p_{qm+r}$. Since the signals are uniform,
\[
\ex[\max\{S_r^q,S_\ell^q\}]
=\frac{p_{qm+r}}{2}+\frac{p_{qm+\ell}^2}{6p_{qm+r}}
\le \frac{2}{3}p_{qm+r}.
\]
Consequently,
\begin{align*}
\frac{1}{m}\ex\left[\sum_{r=1}^{m-1}rS_{(n-m+1+r)}\right]
&\le \frac{1}{m}\sum_{q=0}^{h-1}\sum_{r=1}^m\sum_{\ell=1}^{r-1}\frac{2}{3}p_{qm+r}\\
&=\frac{2}{3}\sum_{q=0}^{h-1}\sum_{r=1}^m\frac{r-1}{m}p_{qm+r}.
\end{align*}
\end{proof}

\begin{proof}[Proof of~\Cref{thm:greedy:m:machines}]
Set
\[
\Delta:=\opt_m-\frac{\opt_1}{m}
=\sum_{q=0}^{h-1}\sum_{r=1}^m\frac{r-1}{m}p_{qm+r},
\]
where the equality follows from~\Cref{lem:m:opt}. Summing~\Cref{lem:m:machines:exp:completion} over all jobs gives
\begin{align*}
\sum_{i=1}^n\ex[C_i]
&\le \frac{m-1}{2m}P
+\frac{1}{m}\left(
\sum_{i=1}^n\ex[d_1(i,i)]
+\sum_{i=1}^n\sum_{j=1}^{i-1}\ex[d_1(i,j)+d_1(j,i)]
\right)\\
&\qquad+\sum_{i=1}^n\sum_{k=0}^{m-2}\frac{m-k-1}{m}\ex[Y_i^kS_i].
\end{align*}
For single-machine $\GREEDY$, we have $g(x)=2-\frac{2}{3}x$, so $g(1)=4/3$ and $\rho=16/9$. The slack form~\eqref{eq:global:delay:slack} of the Global Delay Theorem therefore yields
\begin{align*}
&\frac{1}{m}\left(
\sum_{i=1}^n\ex[d_1(i,i)]
+\sum_{i=1}^n\sum_{j=1}^{i-1}\ex[d_1(i,j)+d_1(j,i)]
\right)\\
&\qquad\le \frac{1}{m}\left(\frac{16}{9}\opt_1-\frac{5}{9}P\right)
=\frac{16}{9}\opt_m-\frac{16}{9}\Delta-\frac{5}{9m}P.
\end{align*}
By~\Cref{lem:parallel-final-sum}, the final correction term is at most $2\Delta/3$. It remains to compare the other additive term with $\Delta$. For each $q\in\{0,\ldots,h-1\}$, Chebyshev's sum inequality gives
\begin{align*}
\frac{1}{m}\sum_{r=1}^m(r-1)p_{qm+r}
&\ge \left(\frac{1}{m}\sum_{r=1}^m(r-1)\right)
\left(\frac{1}{m}\sum_{r=1}^m p_{qm+r}\right)\\
&=\frac{m-1}{2m}\sum_{r=1}^m p_{qm+r}.
\end{align*}
Summing over $q$ shows that $\Delta\ge (m-1)P/(2m)$. Combining all bounds,
\begin{align*}
\sum_{i=1}^n\ex[C_i]
&\le \frac{16}{9}\opt_m
+\frac{m-1}{2m}P-\frac{10}{9}\Delta-\frac{5}{9m}P\\
&\le \frac{16}{9}\opt_m,
\end{align*}
which proves the theorem.
\end{proof}

\section{Makespan Minimization on Parallel Identical Machines}
\label{sec:makespan}

We turn to makespan minimization in the \emph{strong} stochastic clairvoyance model, where the value of $p_j$ is revealed at the signal time $S_j$. 
The optimal offline solution can be computed in polynomial time using McNaughton's Wrap-Around-Rule~\cite{McN59}. Hence $\OPT = \max\{P/m, p_{\max}\}$, where $P = \sum_j p_j$ and $p_{\max} = \max_j p_j$.
Intuitively, a good solution needs to identify long jobs early to ensure that they are scheduled together with other jobs and do not stick out at the very end.
\Cref{alg:lfreeze} exactly follows this idea by first testing jobs with RR, freezing potentially short jobs, and deferring those to the very end of the schedule. 

\begin{algorithm}
\DontPrintSemicolon
\caption{$\LFREEZE$ for $\lambda > 0$}
\label{alg:lfreeze}
Run Round-Robin (RR). \;
    If $j$'s elapsed time reaches $f_j := \max\{S_j,\; p_j/(1+\lambda)\}$,
    \emph{freeze} $j$ and remove it from RR. \;
    At the first time when exactly $m$ jobs remain unfrozen, complete those on their respective machines and finish the frozen jobs by List Scheduling. \;
\end{algorithm}

Note that $f_j$ can be computed because after $j$'s elapsed time reaches $S_j$, the algorithm learns about $p_j$ and thus can evaluate the maximum. We show the following theorem.

\begin{restatable}{theorem}{thmMakespan}\label{thm:makespan}
    $\LFREEZE$ with $\lambda=1$ is $\nicefrac{13}{8}$-competitive for makespan minimization on parallel identical machines under strong uniform stochastic clairvoyance.
\end{restatable}

Assume $n > m$, as otherwise the problem is trivial.
Note that $\OPT = \max\{P/m, p_{\max}\}$, where $P = \sum_j p_j$ and $p_{\max} = \max_j p_j$.
Let $J$ denote the set of all jobs, $F$ the set of frozen jobs at time $\tau$, and $U$ the set of $m$ unfrozen jobs at $\tau$. Each frozen job $j \in F$ has received exactly $f_j$ units of work by time $\tau$. Let $L$ denote the common elapsed time of each unfrozen job at $\tau$; note that $f_j \leq L$ for each $j \in F$. For a subset $Q \subseteq J$, write $f(Q) := \sum_{j \in Q} f_j$. Since all machines are busy throughout $[0, \tau]$,
\begin{equation}\label{eq:makespan-volume-tau}
    m\tau = f(F) + mL.
\end{equation}
Since each unfrozen job $j$ has $f_j \geq L$,
\begin{equation}\label{eq:makespan-unfrozen-bound}
    f(J) \geq f(U) \geq mL.
\end{equation}

The following lemma bounds the residual size of frozen jobs.

\begin{lemma}\label{lem:frozen-job-remaining}
    Every frozen job $j \in F$ has residual size $p_j(\tau) \leq \lambda L$ at time $\tau$.
\end{lemma}
\begin{proof}
    Since $f_j \geq p_j/(1+\lambda)$, we have $p_j \leq (1+\lambda) f_j$, so
    \[
        p_j(\tau) = p_j - f_j \leq (1+\lambda) f_j - f_j = \lambda f_j \leq \lambda L. \qedhere
    \]
\end{proof}

For each $j \in U$, the residual is $p_j(\tau) = p_j - L$, so $\max_{j \in U} p_j(\tau) \leq p_{\max} - L$.

\begin{lemma}\label{lem:list-scheduling-bound}
    Suppose $m$ machines have fixed initial loads $\ell_1, \ldots, \ell_m$, and a set $Q$ of additional jobs with sizes $q_1, \ldots, q_k$ is to be scheduled, with $q_{\max} = \max_{j \in Q} q_j$. Then applying List Scheduling to $Q$ on these machines yields a makespan of at most
    \[
        \max \left\{ \max_{i \in [m]} \ell_i,\;\; \frac{1}{m} \left( \sum_{i=1}^m \ell_i + \sum_{j=1}^k q_j \right) + q_{\max} \right\}.
    \]
\end{lemma}
\begin{proof}
    Consider the job that completes last. If it starts at time $0$, the makespan is at most $\max\{\max_i \ell_i, q_{\max}\}$.

    Otherwise, let $q$ be the size of the job that completes last, and let $t$ be its start time. Since List Scheduling assigns jobs as soon as a machine becomes free, all machines are busy during $[0, t]$. Therefore
    \[
        mt \;\leq\; \left( \sum_{i=1}^m \ell_i + \sum_{j=1}^k q_j \right) - q,
    \]
    and the makespan is bounded by
    \[
        t + q \;\leq\; \frac{1}{m}\!\left( \sum_{i=1}^m \ell_i + \sum_{j=1}^k q_j \right) + \left( 1 - \frac{1}{m} \right) q \;\leq\; \frac{1}{m}\!\left( \sum_{i=1}^m \ell_i + \sum_{j=1}^k q_j \right) + q_{\max}. \qedhere
    \]
\end{proof}

Apply \Cref{lem:list-scheduling-bound} at time $\tau$, with initial loads $\{p_i(\tau) : i \in U\}$ on the $m$ unfrozen machines and additional jobs of sizes $\{p_j(\tau) : j \in F\}$. Using \Cref{lem:frozen-job-remaining} and $\max_{i \in U} p_i(\tau) \leq p_{\max} - L$, the remaining schedule length after $\tau$ is at most
\begin{align*}
    &\max\!\left\{ p_{\max} - L,\;\; \frac{1}{m}\!\left( \textstyle\sum_{i \in U} p_i(\tau) + \sum_{j \in F} p_j(\tau) \right) + \lambda L \right\} \\
    &\;\leq\; \max\!\left\{ p_{\max} - L,\;\; \frac{1}{m}\!\left( P - f(F) - mL \right) + \lambda L \right\}.
\end{align*}
Combining with $\tau = f(F)/m + L$ from~\eqref{eq:makespan-volume-tau},
\begin{align*}
    C_{\max}
    &\leq \tau + \max\!\left\{ p_{\max} - L,\;\; \tfrac{1}{m}\!\left( P - f(F) - mL \right) + \lambda L \right\} \\
    &= \tfrac{1}{m} f(F) + L + \max\!\left\{ p_{\max} - L,\;\; \tfrac{1}{m}\!\left( P - f(F) - mL \right) + \lambda L \right\} \\
    &= \max\!\left\{ p_{\max} + \tfrac{1}{m} f(F),\;\; \tfrac{1}{m} P + \lambda L \right\} \\
    &\leq \max\!\left\{ p_{\max},\; \tfrac{P}{m} \right\} + \max\{1, \lambda\} \cdot \tfrac{f(J)}{m} \\
    &= \OPT + \max\{1, \lambda\} \cdot \tfrac{f(J)}{m},
\end{align*}
where the second-to-last inequality uses $\lambda L \leq \lambda \tfrac{1}{m} f(U) \leq \lambda \tfrac{1}{m} f(J)$ from~\eqref{eq:makespan-unfrozen-bound} and $f(F) \leq f(J)$. Taking expectations,
\begin{equation}\label{eq:makespan-master}
    \ex[C_{\max}] \;\leq\; \OPT + \max\{1, \lambda\} \cdot \frac{\ex[f(J)]}{m}.
\end{equation}

\begin{lemma}\label{lem:E-fJ}
    $\ex[f(J)] = \left(\tfrac{1}{2} + \tfrac{1}{2(1+\lambda)^2}\right) P$.
\end{lemma}
\begin{proof}
    Since $f_j = \max\{S_j,\, p_j/(1+\lambda)\}$ and $S_j \sim \mathcal{U}[0, p_j]$,
    \begin{align*}
        \ex[f_j]
        &= \frac{1}{p_j} \left( \int_0^{p_j/(1+\lambda)} \frac{p_j}{1+\lambda}\, ds \;+\; \int_{p_j/(1+\lambda)}^{p_j} s\, ds \right) \\
        &= \frac{p_j}{(1+\lambda)^2} \;+\; \frac{p_j}{2}\!\left(1 - \frac{1}{(1+\lambda)^2}\right) \\
        &= \frac{p_j}{2} + \frac{p_j}{2(1+\lambda)^2}.
    \end{align*}
    Summing over $j$ gives the claim.
\end{proof}

\begin{proof}[Proof of \Cref{thm:makespan}]
Substituting \Cref{lem:E-fJ} into~\eqref{eq:makespan-master} and using $P/m \leq \OPT$,
\[
    \ex[C_{\max}] \;\leq\; \OPT + \max\{1, \lambda\} \!\left(\tfrac{1}{2} + \tfrac{1}{2(1+\lambda)^2}\right)\!\OPT.
\]
The factor $\max\{1, \lambda\} \!\left(\tfrac{1}{2} + \tfrac{1}{2(1+\lambda)^2}\right)$ equals $\tfrac{1}{2} + \tfrac{1}{2(1+\lambda)^2}$ for $\lambda \leq 1$ (decreasing in $\lambda$) and $\tfrac{\lambda}{2} + \tfrac{\lambda}{2(1+\lambda)^2}$ for $\lambda \geq 1$ (increasing in $\lambda$), so the minimum is attained at $\lambda = 1$, giving $\tfrac{1}{2} + \tfrac{1}{8} = \tfrac{5}{8}$ and yielding $\ex[C_{\max}] \leq \tfrac{13}{8}\,\OPT$. This proves \Cref{thm:makespan}.
\end{proof}

\section{Beta Distributions and Multiple Uniform Signals}\label{sec:beta}

In this section, we again study total completion time minimization on a single machine.

\subsection{\textsc{Greedy} on Signals with Beta Distributions}\label{sec:greedy-beta}

We extend the analysis of $\GREEDY$ to signals $S_j \sim p_j \cdot \mathrm{Beta}(\alpha, \beta)$ with $\alpha, \beta \geq 1$. The Beta family generalizes the uniform case ($\mathrm{Beta}(1,1)$) and captures order statistics of multiple uniform signals: $\mathrm{Beta}(\alpha, 1)$ and $\mathrm{Beta}(1,\beta)$ are the distributions of the maximum of $\alpha$ and the minimum of $\beta$ i.i.d.\ uniform signals in $[0,1]$.

Applying \Cref{lem:pairwise-delay-greedy} to Beta signals, with $X_k := S_k/p_k \sim \mathrm{Beta}(\alpha,\beta)$ and $F$ and $f$ denoting its CDF and PDF, we find that the normalized expected pairwise delay of $\GREEDY$ is $g(x) = 1 + \int_0^1 (1-F(t))(1-F(xt))\, dt + \left(\frac{1}{x} - 1\right) \int_0^1 F(xt)\, f(t)\, dt$,
where $x = p_i/p_j$. Results on competitive ratios can be obtained from the delay bounds applied to the above $g$. We give three bounds of increasing strength, compared in \Cref{fig:beta-comp-heatmaps}.

\paragraph{Baseline from Local Delay Theorem.}
Bounding $g$ by a single $x$-independent constant already yields a competitive ratio, through the Local Delay Theorem (\Cref{lem:weak-delay}).
\begin{lemma}[Local Delay Theorem for $\mathrm{Beta}(\alpha,\beta)$ distribution]\label{prop:beta-weak}
    For $\mathrm{Beta}(\alpha,\beta)$-stochastic clairvoyance with $\alpha >2$ and $\beta>0$, the competitive ratio of the $\GREEDY$ algorithm is at most
    \[
    1 + \frac{\alpha}{\alpha+\beta}
 + \sqrt{\frac{\alpha\beta(\alpha+\beta-1)(\alpha+\beta-2)}{(\alpha+\beta)^2(\alpha+\beta+1)(\alpha-1)(\alpha-2)}}. 
    \]
\end{lemma}
This is the baseline that the next two bounds improve upon; it is plotted in the left panel of \Cref{fig:beta-comp-heatmaps}, and we prove it in \Cref{ap:beta-weak}. Next we strengthen the result via the Global Delay Theorem (\Cref{lem:strong-delay}).

\paragraph{A Global Delay Theorem via the envelope.} As $g$ is in general neither monotone nor concave, we first bound it by the analytically tractable \emph{envelope} (\Cref{lem:envelope})
\[
    g(x) \leq U(x) := 1 + \mu + C\, x^{\alpha-1} - (C+\Delta)\, x^\alpha,
\]
with the constants $\mu, C, \Delta$ as defined in \Cref{lem:envelope}, and then find a majorant for $U$. 
The envelope is unimodal on $(0,1]$ with maximizer $x_U^* = \tfrac{\alpha-1}{\alpha}\cdot \tfrac{C}{C+\Delta}$ (\Cref{lem:U-shape}); capping it at $U(x_U^*)$ on $[0,x_U^*]$ yields a valid nonincreasing concave majorant (for $\alpha=1$, $x_U^*=0$ and the majorant is $U$ itself).
\begin{restatable}{theorem}{thmBetaCaseTwoU}\label{thm:beta-case2-U}
    For $\mathrm{Beta}(\alpha,\beta)$-stochastic clairvoyance with $\alpha \geq 1$ and $\beta \geq 1$, the $\GREEDY$ algorithm has competitive ratio at most
    \[
        \rho_U \;=\; \sqrt{x_U^*}\cdot U(x_U^*) + \int_{\sqrt{x_U^*}}^{1} U(t^2)\, dt,
    \]
    with $U(x_U^*) = 1 + \tfrac{\alpha}{\alpha+\beta} + \tfrac{C}{\alpha}\,(x_U^*)^{\alpha-1}$.
\end{restatable}
The integral admits a closed form (\Cref{ap:beta-envelope}). For $\alpha=1$ the envelope is linear, $U(x)=2-\tfrac{2\beta}{2\beta+1}x$, so $x_U^*=0$ (no flat piece) and $\rho_U = 2 - \tfrac{2\beta}{3(2\beta+1)}$, recovering $16/9$ at $\beta=1$ (\Cref{thm:greedy}); since $g$ is convex there with $U$ its secant, no linear majorant does better (\Cref{ap:beta-tightness}). Relative to applying the local delay bound to the same envelope, which would use the constant $U(x_U^*)$, $\rho_U$ is a strict improvement whenever the envelope is not flat, i.e.\ $U(x_U^*) > U(1)$, with the gain growing in $\alpha$; e.g.\ for $\mathrm{Beta}(2,2)$, $\rho_U \approx 1.664$ versus $U(x_U^*) \approx 1.697$.

\paragraph{A tighter bound via direct analysis of $g$.} Majorizing $g$ directly, rather than through $U$, closes a relaxation gap that grows with $\alpha,\beta$. The resulting majorant---constant up to the maximizer of $g$, then equal to $g$ (with a tangent extension near $x=1$ where $g$ may be convex)---is pointwise no larger than the majorant from \Cref{thm:beta-case2-U}, so the bound $\rho_g$ of \Cref{thm:beta-g-direct} is never worse than $\rho_U$ and is typically much tighter (e.g.\ $\mathrm{Beta}(5,5)$: $\rho_g \approx 1.559$ versus $\rho_U \approx 2.100$). The construction and the full numerical comparison are given in \Cref{ap:beta-direct-g} and \Cref{tab:CR-comparison}; it relies on a structural property of $g$ stated in an assumption of \Cref{thm:beta-g-direct}, currently verified computationally for $\alpha,\beta \in [1,20]$.

\begin{figure}[t]
    \centering
    \includegraphics[width=\linewidth]{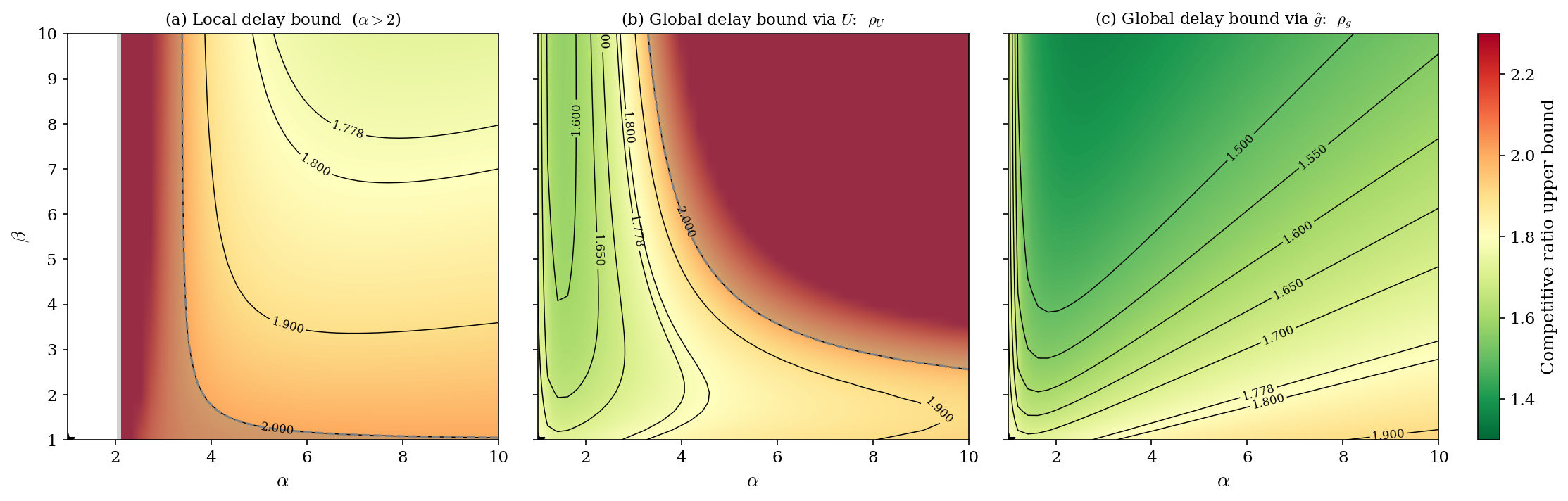}
    \caption{Competitive-ratio upper bounds for $\GREEDY$ under $\mathrm{Beta}(\alpha,\beta)$ signals.
    \textbf{(a)} Local Delay Theorem (\Cref{prop:beta-weak}, via \Cref{lem:weak-delay}; $\alpha>2$ only).
    \textbf{(b)} $\rho_U$ via the envelope (\Cref{thm:beta-case2-U}).
    \textbf{(c)} $\rho_g$ via the direct majorant (\Cref{thm:beta-g-direct}); 
    both \textbf{(b)} and \textbf{(c)} apply the Global Delay Theorem (\Cref{lem:strong-delay}).}
    \label{fig:beta-comp-heatmaps}
\end{figure}

\subsection{Scheduling with Multiple Uniform Signals}\label{sec:asymptotics}
In this section, we consider the setting where each job $j\in[n]$ is associated with a list of $k\in \Nbb$ random signals $S_j^{1},\dots, S_j^{k}$. This setting is also known as \textit{scheduling with a progress bar} and $k$ is referred to as the \textit{granularity} of the progress bar (see \cite{BenomarCLS25} for detailed definitions and motivations). The following algorithm (\Cref{alg:etc}) was proposed in \cite{BenomarCLS25}, which analyzed it under the assumption that the signals are generated by a Poisson process. In this section, we study the natural variant in which the number of signals $k$ is fixed and each signal is uniformly and independently distributed over $[0,1]$.  

\begin{algorithm}
\DontPrintSemicolon
\caption{Repeated Explore-Then-Commit \cite{BenomarCLS25}}
\label{alg:etc}
Run Round-Robin (RR) \;
    Whenever a job $j$ emits its $h = \ceil{k^{2/3}}$-th signal, stop RR and run $j$ to completion.
\end{algorithm}

Classically, considering $k$ i.i.d.\ random variables $S_j^{1},\dots, S_j^{k} \sim \Ucal([0,1])$, and denoting the index of the $h$-th value by $(h)\in[k]$ so that $S_j^{(1)}\leq \dots \leq S_j^{(k)}$, the random variable $S_j^{(h)}$ follows a $\mathrm{Beta}(h,k-h+1)$ distribution. Thus, in the model described above, the competitive ratio of \Cref{alg:etc} under $k$ uniform signals is equal to the competitive ratio of $\GREEDY$ under a single $\mathrm{Beta}(h,k+1-h)$ signal. This motivates the following result, which captures the asymptotic behavior of the competitive ratio as the granularity $k$ increases.

\begin{restatable}{theorem}{BetaAsymptoticUB} \label{thm:beta-asymptotic-ub} The competitive ratio of \Cref{alg:etc} for $k$ uniform weak signals and threshold $h = \ceil{k^{2/3}}$ is $1 + O(k^{-1/3})$. Conversely, the competitive ratio of any scheduling algorithm using $k$ uniform weak signals is at least $1 + \Omega(k^{-1/3})$. \end{restatable}

\begin{proof}
For the upper bound, by the above discussion, the competitive ratio of \Cref{alg:etc} under $k$ uniform signals is equal to the competitive ratio of $\GREEDY$ under a single $\mathrm{Beta}(\ceil{k^{2/3}},k - \ceil{k^{2/3}} + 1)$-distributed signal. 
Applying \Cref{prop:beta-weak} with $\alpha = \ceil{k^{2/3}} \sim_{k\to\infty} k^{2/3}$ and $\beta = k - \ceil{k^{2/3}} + 1 \sim_{k\to\infty} k$ for sufficiently large $k$, we obtain that the competitive ratio is bounded by
\[
1 + O(k^{-1/3}) + \sqrt{\frac{\Theta(k^{11/3})}{\Theta(k^{13/3})}}
= 1 + O(k^{-1/3}) + \sqrt{O(k^{-2/3})}
= 1 + O(k^{-1/3}).
\]

The proof of the lower bound shares a common thread with \cite{BenomarCLS25}. Let $\tau = k^{-1/3}/2>0$. We consider a short job with $p_S = 1$, a long job with $p_L = 1+\tau$, and a scheduling instance with those two jobs where the identity of the short (resp.\ long) job is shuffled at random. Clearly, the optimal schedule in this instance satisfies $\OPT = 1+2+\tau = 3+\tau$. We now study the expected total completion time of an online algorithm $\ALG$.

We say that an online algorithm $\ALG$ is ``misguided'' if the first job to reach an elapsed time of $\tau$ is the long job. The competitive ratio of a scheduling algorithm on this instance is lower bounded by 
\[
1+\frac{q}{3}\tau - \frac{2}{9}\tau^2,
\]
where $q$ is the probability that the algorithm is misguided
\cite[Lemma C.7]{BenomarCLS25}. Also, by \cite[Lemma C.8]{BenomarCLS25}, the probability that an online algorithm is misguided satisfies $q \geq 1/2-d_{\mathrm{TV}}(X_S(\leq \tau),X_L(\leq \tau))$ where $d_{\mathrm{TV}}(\cdot,\cdot)$ is the total variation distance and where $X_S(\leq \tau)$ (resp.\ $X_L(\leq \tau)$) is defined as the list of signals for job $S$ (resp.\ $L$) that appeared before time $\tau$. Therefore, to conclude the proof it suffices to show that $d_{\mathrm{TV}}(X_S(\leq \tau),X_L(\leq \tau))$ is bounded away from $1/2$. We do this via the following sequence of equations
\begin{align}
    d_{\mathrm{TV}}(X_S(\leq \tau),X_L(\leq \tau)) &= d_{\mathrm{TV}}\!\left(
\mathrm{Bin}(k,\tau),
\mathrm{Bin}\!\left(k,\frac{\tau}{1+\tau}\right)
\right)\label{eq:cond}\\
&\leq \sqrt{\frac{1}{2}D_{\mathrm{KL}}\!\bigl(\mathrm{Bin}(k,\tau)\,\|\,\mathrm{Bin}(k,\tfrac{\tau}{1+\tau})\bigr)}\label{eq:pink}\\
&= \sqrt{\frac{k}{2} D_{\mathrm{KL}}\!\bigl(\mathrm{Ber}(\tau)\,\|\,\mathrm{Ber}(\tfrac{\tau}{1+\tau})\bigr)}\label{eq:ber}\\
&\leq \sqrt{\frac{k}{2}\tau^3}\label{eq:bernouilli-lemma}\\
&\leq \frac{1}{4}\label{eq:tauk}
\end{align}
where \eqref{eq:cond} uses the balls-in-urns conditioning trick \cite{mitzenmacher2017probability}, \eqref{eq:pink} is Pinsker's inequality,  \eqref{eq:ber} uses $D_{\mathrm{KL}}\!\bigl(\mathrm{Bin}(k,p)\,\|\,\mathrm{Bin}(k,r)\bigr)
=
k\,D_{\mathrm{KL}}\!\bigl(\mathrm{Ber}(p)\,\|\,\mathrm{Ber}(r)\bigr)$, \eqref{eq:bernouilli-lemma} is an application of \Cref{lem:kl-ber} below and \eqref{eq:tauk} follows from the definition of $\tau$. 
\end{proof}

\begin{restatable}[Bernoulli KL upper bound]{lemma}{KLBerUpperBound}
\label{lem:kl-ber}
For $p\in [0,1]$ and $r\in (0,1)$,
\[
D_{\mathrm{KL}}\!\bigl(\mathrm{Ber}(p)\,\|\,\mathrm{Ber}(r)\bigr)
=
p\log\frac{p}{r} + (1-p)\log\frac{1-p}{1-r}
\]
satisfies
\[
D_{\mathrm{KL}}\!\bigl(\mathrm{Ber}(p)\,\|\,\mathrm{Ber}(r)\bigr)
\le
\frac{(p-r)^2}{r(1-r)}.
\]
Consequently, for $\tau\in[0,1/4]$,
\[
D_{\mathrm{KL}}\!\bigl(\mathrm{Ber}(\tau)\,\|\,\mathrm{Ber}(\tfrac{\tau}{1+\tau})\bigr)
\le
\tau^3.
\]
\end{restatable}

\begin{proof}
By definition,
\[
D_{\mathrm{KL}}\!\bigl(\mathrm{Ber}(p)\,\|\,\mathrm{Ber}(r)\bigr)
=
p\log\frac{p}{r}
+
(1-p)\log\frac{1-p}{1-r}.
\]
Using $\log x\le x-1$ for all $x>0$, we get
\[
p\log\frac{p}{r}
\le
p\left(\frac{p}{r}-1\right)
=
\frac{p(p-r)}{r},
\]
and
\[
(1-p)\log\frac{1-p}{1-r}
\le
(1-p)\left(\frac{1-p}{1-r}-1\right)
=
\frac{(1-p)(r-p)}{1-r}.
\]
Therefore,
\begin{align*}
D_{\mathrm{KL}}\!\bigl(\mathrm{Ber}(p)\,\|\,\mathrm{Ber}(r)\bigr)
&\le
\frac{p(p-r)}{r}
+
\frac{(1-p)(r-p)}{1-r} \\
&=
(p-r)\left(\frac{p}{r}-\frac{1-p}{1-r}\right) \\
&=
\frac{(p-r)^2}{r(1-r)}.
\end{align*}
Applying the above inequality to $p= \tau$ and $r = \frac{\tau}{1+\tau}$, we get, 
\begin{align*}
D_{\mathrm{KL}}\!\bigl(\mathrm{Ber}(\tau)\,\|\,\mathrm{Ber}(\tfrac{\tau}{1+\tau})\bigr)\leq \frac{\left(\frac{\tau^2}{1+\tau}\right)^2}{\frac{\tau}{({1+\tau})^2}}  = \tau^3
\end{align*}
\end{proof}

\section{Conclusion}

We introduced stochastic clairvoyance, a new beyond-worst-case model for 
online scheduling.
We showed that for various distributions, algorithms can break 
the nonclairvoyant lower bound of $2$ for various scheduling 
environments and objective functions.
This suggests that stochastic clairvoyance is not a one-time trick, but rather points to a new family of beyond-worst-case models.
We believe that our work is a starting point for finding similar practically relevant models for other fundamental online problems, such as optimal stopping, allocation, or caching problems.

\newpage 
\appendix
\crefalias{section}{appendix}
\crefalias{subsection}{appendix}
\crefalias{subsubsection}{appendix}
\section*{Appendices}

\section{Appendix for \Cref{sec:global-delay} (Global Delay Theorem)}
\label{app:delay-theorem}

We also show a weighted version of the Global Delay Theorem. 

\begin{theorem}[Weighted Global Delay Theorem]
\label{thm:weighted-strong-delay}
Let $g:[0,1]\to[1,\infty)$ be nonincreasing and concave, and let
\[
    \rho:=\int_0^1g(t^2)\,\mathrm{d}t .
\]
Consider a single-machine nonidling scheduling algorithm $\alg$.  Suppose
that for every instance with positive processing times $p_j$ and positive
weights $w_j$, and for every two distinct jobs $i,j$ with
\[
    q_i:=\frac{p_i}{w_i}\le\frac{p_j}{w_j} =: q_j,
\]
it holds that
\[
    \ex\!\left[w_j d^\alg(i,j)+w_i d^\alg(j,i)\right]
    \le
    w_jp_i\,g\!\left(\frac{q_i}{q_j}\right).
\]
Then $\alg$ is $\rho$-competitive for minimizing
$\ex[\sum_j w_jC_j]$.
\end{theorem}

\begin{proof}
Relabel the jobs so that $q_1\le\cdots\le q_n$, where
$q_j=p_j/w_j$.  The weighted delay decomposition is
\[
    \ex[\ALG]
    =
    \sum_{j=1}^n w_jp_j
    +
    \sum_{1\le i<j\le n}
    \ex\!\left[w_j d^\alg(i,j)+w_i d^\alg(j,i)\right].
\]
Let $B:=\sum_{j=1}^n w_jp_j$.
By our hypothesis, we have
\[
    \ex[\ALG]
    \le
    B+
    \sum_{1\le i<j\le n}
    w_jp_i\,g\!\left(\frac{q_i}{q_j}\right).
\]
By Smith's rule, the optimum weighted completion time is
\[
    \OPT
    =
    B+\sum_{1\le i<j\le n}w_jp_i .
\]

Let $W:=\sum_{j=1}^n w_j$ and define
$a_0:=0$ and $a_j:=W^{-1}\sum_{k=1}^j w_k$ for $j\in[n]$.  Set
$I_j:=(a_{j-1},a_j]$, so that $|I_j|=w_j/W$, and define
$f:[0,1]\to[0,\infty)$ by $f(x)=q_j$ for $x\in I_j$.  Then $f$ is
nondecreasing. Set
\[
    N:=\int_0^1\int_x^1
    f(x)g\!\left(\frac{f(x)}{f(y)}\right)
    \,\mathrm{d}y\,\mathrm{d}x
\]
and
\[
    D:=\int_0^1\int_x^1 f(x)\,\mathrm{d}y\,\mathrm{d}x .
\]
The integrands are constant on each rectangle $I_i \times I_j$ with $i < j$, and on the triangular half of each square $I_j \times I_j$. Therefore
\[
    W^2N
    =
    \sum_{1\le i<j\le n}
    w_jp_i\,g\!\left(\frac{q_i}{q_j}\right)
    +
    \frac{g(1)}2B
\]
and
\[
    W^2D
    =
    \sum_{1\le i<j\le n}w_jp_i+\frac12B
    =
    \OPT-\frac12B .
\]
Applying \Cref{lem:f-bound} gives $N\le\rho D$, and therefore
\begin{align*}
\ex[\ALG]
&\le
W^2N+\left(1-\frac{g(1)}2\right)B\\
&\le
\rho\left(\OPT-\frac12B\right)
+\left(1-\frac{g(1)}2\right)B\\
&=
\rho\,\OPT+\frac{B}{2}\bigl(2-g(1)-\rho\bigr)
\le
\rho\,\OPT .
\end{align*}
The last inequality follows from $g(1)\ge1$ and
$\rho\ge g(1)$.
\end{proof}

We can use \Cref{thm:weighted-strong-delay} to show that $\GREEDY$ is $16/9$-competitive on 
a single machine
for minimizing the sum of weighted completion times. Note that 
in the weighted setting, $\GREEDY$ runs Weighted Round Robin for exploration 
where each job $j$ gets a rate $y_j(t) = \frac{w_j}{\sum_{i \in U(t)} w_i}$.

\section{Appendix for \Cref{sec:single-machine} (Single Machine Algorithms)}\label{ap:single-machine}

\subsection{Omitted proofs for \Cref{sec:lgreedy} ($\LGREEDY$)}\label{app:lambda-greedy}

\lemRobustDelays*

\begin{proof}

We distinguish the two cases (1) $p_j \ge \frac{p_i}{\lambda}$ and (2) $p_j < \frac{p_i}{\lambda}$.

\textbf{Case (1):} Assume $p_j \ge \frac{p_i}{\lambda}$. In this case, we can observe that $i$ always completes before $j$: By definition of the algorithm, the maximum amount of processing that $j$ can receive before $i$ completes is strictly less than $p_i + \frac{1-\lambda}{\lambda} p_i = \frac{p_i}{\lambda}$, where the first summand is caused by running both jobs in parallel during Round-Robin and the second summand is caused by the speed-up of $j$ after it emits its signal. Hence, the amount of processing that $j$ receives before $i$ completes does not suffice to complete $j$.

We analyze the pairwise delay between $i$ and $j$ by distinguishing between different configurations of $S_i$ and $S_j$.

\begin{enumerate}
    \item $S_i \ge \lambda p_i$: In this case, $i$ will complete during the speed-up phase after emitting its signal. The pairwise delay depends on whether $j$ emits earlier or later than $i$:
    \begin{enumerate}
         \item $S_j \ge S_i$: Since $j$ emits later than $i$, both jobs receive $S_i$ units of processing before $i$ emits. Hence,
        $$
            d(i,j) + d(j,i) = p_i + S_i.
        $$
        \item $S_j \le S_i$: Since $j$ emits before $i$, the pairwise delay increases by $\frac{1-\lambda}{\lambda} \cdot S_j$, which is the amount of processing that $j$ receives during its speedup, compared to the previous case. Hence,
        $$
            d(i,j) + d(j,i) = p_i + S_i + \frac{1-\lambda}{\lambda} \cdot S_j.
        $$
    \end{enumerate}
    \item $S_i < \lambda p_i$: In this case, $i$ does not complete during the speed-up phase after emitting its signal and, thus, will receive more processing while the algorithm executes Round-Robin after the speed-up phase. The pairwise delay again depends on whether $j$ emits its signal before $i$ completes:
    \begin{enumerate}
        \item $S_j > p_i - \frac{1-\lambda}{\lambda} S_i$: The right-hand side of this inequality is the amount of processing that $j$ receives while the algorithm works on $i$ \emph{and} $j$ while running Round-Robin, which is equal to the processing time of $i$ minus the amount of processing that $i$ receives during the speed-up. Since $S_j$ is strictly larger than this amount, $j$ does not emit its signal before $i$ completes. Hence,
        $$
        d(i,j) + d(j,i) = 2p_i - \frac{1-\lambda}{\lambda} \cdot S_i.
        $$        
        \item $S_j \le p_i - \frac{1-\lambda}{\lambda} S_i$: In contrast to the previous case, this inequality implies that $j$ emits its signal before $i$ completes. This increases the pairwise delay by $\frac{1-\lambda}{\lambda} \cdot S_j$, which is the amount of processing that $j$ receives during its speed-up. Hence,
         $$
        d(i,j) + d(j,i) = 2p_i - \frac{1-\lambda}{\lambda} \cdot S_i + \frac{1-\lambda}{\lambda} \cdot S_j.
        $$      
    \end{enumerate}
\end{enumerate}

Based on these four cases, we can compute the expected pairwise delay as follows:
\begin{align*}
    &\ex[d(i,j)+d(j,i)]\\
    &= \int_{S_i = \lambda p_i}^{p_i} \frac{1}{p_i} \left(\int_{S_j = 0}^{S_i} \frac{p_i+S_i+\frac{1-\lambda}{\lambda} \cdot S_j}{p_j} dS_j + \int_{S_j = S_i}^{p_j} \frac{p_i+S_i}{p_j} dS_j \right) dS_i\\
        &+ \int_{S_i = 0}^{\lambda p_i} \frac{1}{p_i} \left(\int_{S_j = 0}^{p_i-\frac{1-\lambda}{\lambda}S_i} \frac{2p_i-\frac{1-\lambda}{\lambda} S_i+\frac{1-\lambda}{\lambda}S_j}{p_j} dS_j + \int_{S_j = p_i-\frac{1-\lambda}{\lambda}S_i}^{p_j} \frac{2p_i-\frac{1-\lambda}{\lambda} S_i}{p_j} dS_j \right) dS_i\\
         &= \int_{S_i = \lambda p_i}^{p_i} \frac{1}{p_i} \left(p_i + S_i + \frac{1-\lambda}{2\lambda p_j}  S_i^2\right) dS_i\\
        &+ \int_{S_i = 0}^{\lambda p_i} \frac{1}{p_i} \left(2p_i - \frac{1-\lambda}{\lambda} S_i+\frac{1-\lambda}{2\lambda p_j}\left(p_i - \frac{1-\lambda}{\lambda} S_i\right)^2\right) dS_i\\    
        &= \frac{1-\lambda}{2}(\lambda+3)p_i + \frac{(1-\lambda)^2(1+\lambda+\lambda^2)}{6\lambda p_j}p_i^2 + \frac{\lambda(\lambda+3)}{2}p_i+\frac{(1-\lambda)(1+\lambda+\lambda^2)}{6p_j}p_i^2\\
        &= \frac{\lambda+3}{2} \cdot p_i + \frac{(1-\lambda)(1+\lambda+\lambda^2)}{6\lambda} \frac{p_i^2}{p_j}
\end{align*}

\textbf{Case (2):} Assume $p_j < \frac{p_i}{\lambda}$. In this case, it can happen that $j$ completes before $i$. As before, we start by analyzing the pairwise delay between $i$ and $j$ for different configurations of $S_i$ and $S_j$:

\begin{enumerate}
    \item $S_i \ge \lambda p_i$: In this case, $i$ will complete during the speed-up phase after emitting its signal. The pairwise delay depends on whether $j$ emits earlier or later than $i$:
    \begin{enumerate}
         \item $S_j \ge S_i$: Since $j$ emits later than $i$, both jobs receive $S_i$ units of processing before $i$ emits. Hence,
        $$
            d(i,j) + d(j,i) = p_i + S_i.
        $$
        
        \item $S_j \le \frac{\lambda (p_j-S_i)}{1-\lambda}$ and $S_i \ge \lambda p_j$: Note that the first condition is equivalent to $S_i + \frac{1-\lambda}{\lambda} S_j \le p_j$. Since $i$ completes during the speed-up, the term $S_i + \frac{1-\lambda}{\lambda} S_j$ describes the amount of processing that $j$ receives before $i$ completes. If this term is less than $p_j$, then $i$ completes before $j$. Hence, the pairwise delay is 
        $$
            d(i,j) + d(j,i) = p_i + S_i + \frac{1-\lambda}{\lambda} \cdot S_j.
        $$
        
        \item $S_j \le S_i$ and $S_i \le \lambda p_j$: The two conditions of this case imply that $S_i + \frac{1-\lambda}{\lambda} S_j \le p_j$. This allows us to argue as in the previous case, and conclude that
        $$
            d(i,j) + d(j,i) = p_i + S_i + \frac{1-\lambda}{\lambda} \cdot S_j.
        $$
        
        \item $S_i \ge \lambda p_j$ and $\lambda p_j \ge S_j \ge \frac{\lambda (p_j-S_i)}{1-\lambda}$: In this case, $j$ emits before $i$ but does not complete during its speed-up because of $S_j \le \lambda p_j$. However, we can observe that $j$ still completes before $i$ emits its signal. This is because if $j$ does not complete before $i$ emits, then $j$ receives $S_i + \frac{1-\lambda}{\lambda} \cdot S_j$ units of processing before $i$ emits its signal. However,
        $$
        S_i + \frac{1-\lambda}{\lambda} \cdot S_j \ge S_i +  (p_j - S_i)  \ge p_j,
        $$
        which implies that $j$ completes before $i$ emits its signal. Hence, the pairwise delay is
        $$
            d(i,j) + d(j,i) = 2p_j - \frac{1-\lambda}{\lambda} S_j.
        $$
        \item $S_i \ge S_j \ge \lambda p_j$: In this case, $j$ completes before $i$ during the speed-up phase after emitting its signal. Hence, the amount of processing that $i$ receives before $j$ completes is $S_j$. Thus,
        $$
            d(i,j) + d(j,i) = p_j + S_j.
        $$
    \end{enumerate}

    To summarize the contribution of these cases to the expected pairwise delay $\ex[d(j,i)+d(i,j)]$, we define
    \begin{align*}
        I_1 &=  \int_{S_i = \lambda p_i}^{p_i} \frac{1}{p_i} \left(\int_{S_j = S_i}^{p_j} \frac{p_i+S_i}{p_j} dS_j\right) dS_i + \int_{S_i = \lambda p_j}^{p_i} \frac{1}{p_i} \left(\int_{S_j = 0}^{ \frac{\lambda (p_j-S_i)}{1-\lambda}}\frac{p_i+S_i+\frac{1-\lambda}{\lambda} S_j}{p_j} dS_j\right) dS_i\\
    &+ \int_{S_i = \lambda p_i}^{\lambda p_j} \frac{1}{p_i} \left( \int_{S_j = 0}^{S_i} \frac{p_i+S_i+\frac{1-\lambda}{\lambda}S_j}{p_j} dS_j \right) dS_i + \int_{S_i = \lambda p_j}^{p_i} \frac{1}{p_i} \left( \int_{S_j = \frac{\lambda (p_j-S_i)}{1-\lambda}}^{\lambda p_j} \frac{2p_j-\frac{1-\lambda}{\lambda}S_j}{p_j} dS_j\right) dS_i\\
    &+ \int_{S_i = \lambda p_j}^{p_i} \frac{1}{p_i} \left( \int_{S_j = \lambda p_j}^{S_i} \frac{p_j+S_j}{p_j} dS_j\right) dS_i\\
    &= \frac{(-1 + \lambda) p_i \left((5 + 5\lambda + 2\lambda^2) p_i - 3(3 + \lambda) p_j \right)}{6 p_j}\\
    &-\frac{\lambda \left(-4 p_i^3 + 6 p_i^2 p_j + 3(-1 + \lambda)^2 p_i p_j^2 + \lambda(-3 + \lambda^2) p_j^3\right)}{6(-1 + \lambda) p_i p_j}\\
    &+ \frac{\lambda^2\left(-(4+\lambda)p_i^3 + 3 p_i p_j^2 + (1+\lambda)p_j^3\right)}{6 p_i p_j} -\frac{\lambda \left(p_i - \lambda p_j\right)^2 \left(p_i + (3 + 2\lambda) p_j\right)}{6(-1 + \lambda) p_i p_j}\\
    &+ \frac{\left(p_i - \lambda p_j\right)^2 \left(p_i + (3 + 2\lambda) p_j\right)}{6 p_i p_j}\\
    &= \frac{\left((\lambda^4 -2\lambda^3 +\lambda^2 -\lambda +4)p_i^3-3(\lambda^3+\lambda^2-3\lambda+4)p_i^2 p_j+3\lambda(2\lambda+1)p_i p_j^2-\lambda^2(2\lambda+1)p_j^3\right)}{6(\lambda-1)p_i p_j}
    \end{align*}

    \item $S_i < \lambda p_i$: In this case, $i$ will not complete during the speed-up phase after emitting its signal.
     \begin{enumerate}
        \item $S_j > p_i - \frac{1-\lambda}{\lambda} S_i$: The right-hand side of this inequality is the amount of processing that $j$ receives while the algorithm works on $i$ \emph{and} $j$ while running Round-Robin, which is equal to the processing time of $i$ minus the amount of processing that $i$ receives during the speed-up. Since $S_j$ is strictly larger than this amount, $j$ does not emit its signal before $i$ completes. Hence,
        $$
        d(i,j) + d(j,i) = 2p_i - \frac{1-\lambda}{\lambda} \cdot S_i.
        $$        
        \item $S_j < \frac{\lambda(p_j-p_i)}{1-\lambda} + S_i$ and $S_i \le \frac{\lambda(p_i - \lambda p_j)}{1- \lambda}$: Since 
        $$
        S_j < \frac{\lambda(p_j-p_i)}{1-\lambda} + S_i \le \frac{\lambda(p_j-p_i)}{1-\lambda} + \frac{\lambda(p_i - \lambda p_j)}{1- \lambda} = \lambda p_j,
        $$
        job $j$ will not complete during its speed-up. The maximum amount of processing that $j$ can receive before $i$ completes is $p_i - \frac{1-\lambda}{\lambda} S_i + \frac{1-\lambda}{\lambda} S_j$. Since
        $$
        p_i - \frac{1-\lambda}{\lambda} S_i + \frac{1-\lambda}{\lambda} S_j < p_i - \frac{1-\lambda}{\lambda} S_i + p_j-p_i + \frac{1-\lambda}{\lambda} S_i = p_j,
        $$
        job $i$ completes before job $j$. Hence,
        $$
        d(i,j) + d(j,i) = 2p_i - \frac{1-\lambda}{\lambda} S_i +  \frac{1-\lambda}{\lambda} S_j.
        $$

        \item  $S_j < p_i - \frac{1-\lambda}{\lambda} S_i$ and $S_i \ge \lambda p_i - \frac{\lambda^2(p_j - p_i)}{1 - \lambda} $: Since
        $$
        S_j < p_i - \frac{1-\lambda}{\lambda} S_i \le p_i - (1-\lambda) p_i + \lambda p_j - \lambda p_i = \lambda p_j,
        $$
        job $j$ does not complete during its speed-up. The maximum amount of processing that $j$ can receive before $i$ completes is $p_i - \frac{1-\lambda}{\lambda} S_i + \frac{1-\lambda}{\lambda} S_j$. Since
        $$
        p_i - \frac{1-\lambda}{\lambda} S_i + \frac{1-\lambda}{\lambda} S_j < p_j,
        $$
        job $i$ completes before job $j$. Hence,
        $$
        d(i,j) + d(j,i) = 2p_i - \frac{1-\lambda}{\lambda} S_i +  \frac{1-\lambda}{\lambda} S_j.
        $$
        \item  $p_i-\frac{1-\lambda}{\lambda}S_i \ge S_j \ge \lambda \cdot p_j$ and $S_i \le \frac{\lambda p_i -\lambda^2p_j}{1-\lambda}$: Since $S_j \ge \lambda p_j$, job $j$ finishes during the speed-up phase after emitting its signal. Furthermore, $S_j \le p_i-\frac{1-\lambda}{\lambda} S_i$ implies that $i$ does not complete before $j$ emits. Note that the condition $p_i-\frac{1-\lambda}{\lambda}S_i \ge S_j \ge \lambda \cdot p_j$ can only hold if  $S_i \le \frac{\lambda p_i -\lambda^2p_j}{1-\lambda}$. The pairwise delay is
        $$
            d(i,j)+ d(j,i) = p_j + S_j + \frac{1-\lambda}{\lambda}S_i.
        $$
        \item $\frac{\lambda(p_j-p_i)}{1-\lambda} + S_i \le  S_j < \lambda p_j$ and $S_i \le \frac{\lambda p_i -\lambda^2p_j}{1-\lambda}$: Since $S_j < \lambda p_j$, job $j$ does not complete during the speed-up phase after emitting its signal. The maximum amount of processing that $i$ can receive before $j$ completes is 
        $$p_j - \frac{1-\lambda}{\lambda} S_j + \frac{1-\lambda}{\lambda} S_i < p_j -(1 -\lambda) p_j + p_i - \lambda p_j = p_i.$$
        Therefore, $i$ does not complete before $j$ completes. The pairwise delay is
        $$
        d(i,j)+ d(j,i) = 2p_j - \frac{1-\lambda}{\lambda} S_j+ \frac{1-\lambda}{\lambda}S_i.
        $$
    \end{enumerate}
    To summarize the contribution of these remaining cases to the expected pairwise delay $\ex[d(j,i)+d(i,j)]$, we define
    \begin{align*}
    I_2 &=  \int_{S_i = 0}^{\lambda p_i} \frac{1}{p_i} \left(\int_{S_j = p_i - \frac{1-\lambda}{\lambda}S_i}^{p_j} \frac{2p_i-\frac{1-\lambda}{\lambda}S_i}{p_j} dS_j\right) dS_i\\
    &+ \int_{S_i=0}^{\frac{\lambda(p_i - \lambda p_j)}{1- \lambda}} \frac{1}{p_i} \left(\int_{S_j=0}^{ \frac{\lambda(p_j-p_i)}{1-\lambda} + S_i} \frac{2p_i - \frac{1-\lambda}{\lambda} S_i + \frac{1-\lambda}{\lambda}S_j}{p_j} dS_j\right) dS_i\\
    &+ \int_{S_i= \lambda p_i - \frac{\lambda^2(p_j - p_i)}{1 - \lambda}}^{\lambda p_i} \frac{1}{p_i} \left(\int_{S_j=0}^{p_i - \frac{1-\lambda}{\lambda}S_i} \frac{ 2p_i - \frac{1-\lambda}{\lambda} S_i +  \frac{1-\lambda}{\lambda} S_j}{p_j} dS_j\right) dS_i\\
    &+ \int_{S_i=0}^{ \frac{\lambda p_i -\lambda^2p_j}{1-\lambda}} \frac{1}{p_i} \left(\int_{S_j=\lambda p_j}^{p_i-\frac{1-\lambda}{\lambda}S_i} \frac{p_j + S_j + \frac{1-\lambda}{\lambda}S_i}{p_j} dS_j\right) dS_i\\
    &+ \int_{S_i=0}^{\frac{\lambda p_i -\lambda^2p_j}{1-\lambda}} \frac{1}{p_i} \left(\int_{S_j=\frac{\lambda(p_j-p_i)}{1-\lambda} + S_i}^{\lambda p_j} \frac{2p_j - \frac{1-\lambda}{\lambda} S_j+ \frac{1-\lambda}{\lambda}S_i}{p_j} dS_j\right) dS_i \\
    &= -\frac{\lambda p_i \left((5 + 5\lambda + 2\lambda^2)p_i - 3(3 + \lambda)p_j\right)}{6 p_j}\\
    &+ \frac{\lambda^2 \left(-4 p_i^3 + 6 p_i^2 p_j + 3(-1+\lambda)^2 p_i p_j^2 + \lambda(-3+\lambda^2) p_j^3\right)}{6(-1+\lambda)^2 p_i p_j} + \frac{\lambda^3 ((4 + \lambda) p_i^3 - 3 p_i p_j^2 - (1 + \lambda) p_j^3)}{6(-1 + \lambda) p_i p_j}\\
    &-\frac{\lambda \left(p_i - \lambda p_j\right)^2 \left(2 p_i + (3 + \lambda) p_j\right)}{6(-1 + \lambda) p_i p_j}+ \frac{\lambda^2 \left(p_i - \lambda p_j\right)^2 \left(2 p_i + (3 + \lambda) p_j\right)}{6(-1 + \lambda)^2 p_i p_j}\\
    &= \frac{\lambda\left((-\lambda^4+2\lambda^3-\lambda^2+\lambda-3)p_i^3+3(\lambda^3+\lambda^2-4\lambda+4)p_i^2 p_j-3\lambda(\lambda+1)p_i p_j^2+\lambda^2(\lambda+1)p_j^3\right)}{6(\lambda-1)^2 p_i p_j}
\end{align*}
\end{enumerate}

The expected pairwise delay is
\begin{align*}
&\ex[d(i,j)+d(j,i)]\\ 
&= I_1 + I_2\\
&=\frac{(-\lambda^4+2\lambda^3-\lambda^2+2\lambda-4)\,p_i^3 +3(\lambda^3-3\lambda+4)p_i^2 p_j +3\lambda(\lambda^2-2\lambda-1)p_i p_j^2 +(-\lambda^4+2\lambda^3+\lambda^2)p_j^3}{6(\lambda-1)^2p_ip_j}
\end{align*}
\end{proof}

\begin{restatable}{lemma}{singlemachinebound}
\label{lem:single-machine-bound}
For $\lambda = 0.038$ and $\hat g_\lambda$ as defined in Equation \eqref{eq:ghat}, we have $\int_0^1 \hat g_{\lambda}(t^2)\,dt < 1.752$.
\end{restatable}

\begin{proof}
For $\lambda = 0.0380$, let
\[
\bar g_\lambda(x) = -0.7069\,x + 2.0996 - 0.0220\,x^{-1} + 0.0003\,x^{-2},
\qquad x \geq \lambda.
\]
By construction, $\bar g_\lambda(x) \ge g_\lambda(x)$ on $[\lambda,1]$. Maximizing $\bar g_\lambda$ over $x \in [\lambda,1]$ yields a maximizer
\[
x^* \in [0.1607,\,0.1608].
\]
Fix any $x^*$ in this interval. Then $x^* \ge \lambda$ and
\[
\bar g_\lambda(x^*) \le 1.8608.
\]
Moreover,
\[
\hat g_\lambda(t^2) \le
\begin{cases}
\bar g_\lambda(x^*), & 0 \le t < \sqrt{x^*},\\
\bar g_\lambda(t^2), & \sqrt{x^*} \le t \le 1.
\end{cases}
\]
Therefore,
\[
\rho \le \int_0^{\sqrt{x^*}} \bar g_\lambda(x^*)\,dt
  + \int_{\sqrt{x^*}}^1 \bar g_\lambda(t^2)\,dt.
\]
Choosing $x^* = 0.1608$ and substituting the expression for $\bar g_\lambda$ gives
\[
\rho \le  1.8608 \cdot \sqrt{x^*}
+ \int_{\sqrt{x^*}}^1 \left(
-0.7069\,t^2
+ 2.0996
- 0.0220\,t^{-2}
+ 0.0003\,t^{-4}
\right) dt.
\]
Evaluating the right-hand side gives $\rho < 1.7520$.
\end{proof}

We conclude the proof of~\Cref{thm:lambda-greedy} by showing that $\hat g_\lambda$ indeed has the required form to apply the Global Delay Theorem. As outlined in~\Cref{sec:lgreedy}, it suffices to analyze the piece of $g_{\lambda}(x)$ with $x \in (\lambda,1]$. Hence, we restrict ourselves to this piece and, for the sake of convenience, rewrite it as
$$
g_\lambda(x)=A x + B + Cx^{-1} + D x^{-2}
$$
with 
\begin{itemize}
    \item $A = \frac{-\lambda^4+2\lambda^3-\lambda^2+2\lambda-4}{6(\lambda-1)^2}$,
    \item $B = \frac{3(\lambda^3-3\lambda+4)}{6(\lambda-1)^2}$,
    \item $C = \frac{3\lambda(\lambda^2-2\lambda-1)}{6(\lambda-1)^2}$ and
    \item $D = \frac{-\lambda^4+2\lambda^3+\lambda^2}{6(\lambda-1)^2}$.
\end{itemize}

\begin{lemma}
    \label{lem:lambda:greedy:gconcave}
    $g_{\lambda}$ is strictly concave on $(\lambda,\infty)$.
\end{lemma}

\begin{proof}
    The second derivative of $g_\lambda$ is
    $$
    g''_{\lambda}(x) = 2Cx^{-3} + 6Dx^{-4}.
    $$
    We can rewrite the remaining coefficients as $C =3\lambda k$ and $D= -\lambda^2 k$ for $k=\frac{\lambda^2-2\lambda-1}{6(\lambda-1)^2}$, and simplify the second derivative to
    $$
    g''_{\lambda}(x) = \frac{6\lambda k}{x^4} (x-\lambda).
    $$
    For $x = \lambda$, we have $g''_{\lambda}(\lambda)=0$. For $x > \lambda$, we can observe that $\lambda \in (0,1)$ implies $k < 0$. Hence, $g''_{\lambda}(x) < 0$.
\end{proof}

\begin{lemma}
    \label{lem:lambda:greedy:maximizer:existence}
    The function $g_{\lambda}$ has a unique global maximizer $x^*_{\lambda} \in (\lambda,1)$. 
\end{lemma}

\begin{proof}
    Since we already showed that $g_{\lambda}$ is strictly concave in $(\lambda,1)$, we can prove the existence of a unique global maximizer in $(\lambda,1)$ by showing that the first derivative $g'_{\lambda}$ changes its sign in $(\lambda,1)$. The first derivative is
    $$
    g'_{\lambda}(x) = A- C x^{-2} -2 D x^{-3}.
    $$

    Again using $k=\frac{\lambda^2-2\lambda-1}{6(\lambda-1)^2}$, we can rewrite
    $$
    g'_{\lambda}(x) = A-3\lambda k x^{-2} +2\lambda^2 k x^{-3}.
    $$

    We check the sign of $g'_{\lambda}(x)$ at $x = \lambda$ and $x = 1$:
    \begin{itemize}
        \item $g'_\lambda(\lambda) = A - \frac{k}{\lambda} = \frac{1-\lambda^3}{6\lambda}$, which is $>0$ for $\lambda \in (0,1)$.
        \item $g'_{\lambda}(1) = A + k\lambda(2\lambda-3) = \frac{(\lambda-4)(\lambda+1)}{6}$, which is $<0$ for $\lambda \in (0,1]$.
    \end{itemize}
\end{proof}

\begin{lemma}
    \label{lem:lambda:greedy:ghat}
    It holds that $\hat g_\lambda(x) \ge g_\lambda(x)$ for all $x \in [0,1]$.
\end{lemma}

\begin{proof}
    For $x \in [\lambda,x_\lambda^*]$, we have $\hat g_\lambda(x) =g_\lambda(x^*_{\lambda}) \ge g_\lambda(x)$, where the inequality follows from~\Cref{lem:lambda:greedy:maximizer:existence}.

    For $x \in [x_\lambda^*,1]$, we have $\hat g_\lambda(x) = g_\lambda(x)$ by definition.

    To argue for $x \le \lambda$, we can first observe that the two pieces of~\Cref{lem:robust-delays} are equal at $x = \lambda$.  For $x \le \lambda$, $g_{\lambda}(x)$ is a linear increasing function. Hence, $g_\lambda(x) \le g_\lambda(\lambda) \le g_\lambda(x^*_{\lambda}) = \hat g_\lambda(x)$, where the second inequality again holds by~\Cref{lem:lambda:greedy:maximizer:existence}. 
\end{proof}

\paragraph{Analysis of $\LGREEDY$ via Local Delay Theorem.} For a fixed value of $\lambda$, if we want to analyze the competitiveness of $\LGREEDY$ via the Local Delay Theorem, then the best bound we can obtain is $\alpha_{\lambda} := \max_{x \in [0,1]} g_{\lambda}(x)$. Optimizing over $\lambda$, we can then obtain the bound $\alpha^* := \min_{\lambda \in (0,1]} \alpha_{\lambda}$. Numerical computations show that $\alpha^* \approx 1.833$ is obtained at $\lambda \approx 0.12012$. Hence, the Local Delay Theorem is strong enough to prove that $\LGREEDY$ has a competitive ratio $< 2$. However, stronger methods, such as the Global Delay Theorem, seem to be necessary to prove that $\LGREEDY$ has a smaller competitive ratio than $\GREEDY$.

\subsection{Omitted proofs for \Cref{sec:strong-sto} (Strong Signals)}\label{app:strong-signals}

Recall that a job $j$ is \emph{known} at time $t$ if $e_j(t) \ge S_j$, i.e., if $j$ already emitted its signal and, therefore, revealed its processing time. Otherwise, call $j$ unknown at $t$. Let $K(t)$ and $U(t)$ denote the sets of known and unknown jobs at time $t$. We analyze the expected pairwise delay.

\begin{lemma}
    \label{lem:boosting-case-delays}
    Consider two jobs $i$ and $j$ with $p_i \le p_j$ under $\LBOOSTING$ with parameter $0<\lambda\le 1$.
    \begin{itemize}
        \item If $p_j \ge \frac{p_i}{\lambda}$, then  $\ex[d(i,j) + d(j,i)] \le \frac{3 + \lambda^2}{2} \cdot p_i - \frac{(1 + 2\lambda^3)}{6} \frac{p_i^2}{p_j}$.
        \item If $p_j \le \frac{p_i}{\lambda}$, then  $\ex[d(i,j) + d(j,i)] \le \lambda^2 p_j+\frac{4 + \lambda^2}{2}p_i-\frac{\lambda^2 + \lambda^3}{2}\frac{p_j^2}{p_i}-\frac{\lambda^3 + 2}{3}\frac{p_i^2}{p_j}$.
    \end{itemize}
\end{lemma}

\begin{proof}
Consider two jobs $i$ and $j$ with $p_i \le p_j$. In the following, we analyze the expected pairwise delay $\ex[d(j,i)+d(i,j)]$ between the two jobs. To this end, we distinguish the two cases $p_j \le \frac{p_i}{\lambda}$ and $p_j \ge \frac{p_i}{\lambda}$.

Since $\LBOOSTING$ processes unknown jobs only by Round-Robin over $U(t)$, all unfinished unknown jobs have the same elapsed processing at all times. Thus, whenever $j$ is unknown, $\min_{j'\in U(t)} e_{j'}(t)=e_j(t)$, and analogously for $i$. Jobs other than $i$ and $j$ can only pause both jobs while some known job is boosted, or cause a known job with no larger processing time to be boosted first. Such events do not increase the amount of processing that one of $i$ and $j$ receives before the other completes. We break ties among equal-length known jobs according to the fixed order used to relabel the jobs. Hence the following two-job case analysis gives a valid upper bound on the pairwise delay in the full instance.

\textbf{Case 1:} Assume that $p_j \ge \frac{p_i}{\lambda}$. In this case, $j$ cannot be boosted before $i$ completes. Indeed, if $i$ is known and unfinished, then $p_i\le p_j$, so $\LBOOSTING$ boosts $i$ or another known job no longer than $i$ before $j$ (up to fixed tie-breaking). If $i$ is still unknown, then boosting $j$ would require
$\min_{j' \in U(t)} e_{j'}(t) \ge \lambda p_j \ge p_i$; by the invariant above, this implies $e_i(t)\ge p_i$, so $i$ has already completed. Thus $j$ is not boosted before $i$ completes, and $i$ completes before $j$ except possibly for harmless ties. We continue by analyzing the pairwise delay between $i$ and $j$ for different configurations of $S_i$ and $S_j$ (cf.~\Cref{fig:strong-signals:delay:cases:1} for an illustration of the configurations):

\begin{figure}[t]
    \centering
    \begin{tikzpicture}[scale = 0.8]
        \draw[thick] (0,0) rectangle (15,5);
        \draw[thick] (0,2.5) -- (15,2.5);
        \draw[thick] (0,0) -- (5,5);
        \draw[thick] (2.5,0) -- (2.5,2.5); 
    
\node at (-1,2.5){\large$S_i$};
    
        \node at (-0.25,-0.25){0};
        \node at (-0.25,5){$p_i$};
        \node at (-0.35,2.5){$\lambda p_i$};

\node at (7.5,-1){\large$S_j$};
    
        \node at (5,-0.25){$p_i$};
        \node at (2.5,-0.25){$\lambda p_i$};
        \node at (7.5,-0.25){$\lambda p_j$};
        \node at (15,-0.25){$p_j$};
    
\node at (0.8333,1.6666){1.};
        \node at (1.6666,0.8333){2.};
        \node at (2.5+0.5*12.5,1.25){3.};
        \node at (2.5+0.5*12.5,2.5+1.25){4.};
        \node at (2,2.5+1.25){5.};
    \end{tikzpicture}
    \caption{Illustration of the different subcases for $p_j \ge \frac{p_i}{\lambda}$ in the proof of~\Cref{lem:boosting-case-delays}. The figure assumes $\lambda = 0.5$ and $p_j = 3 \cdot p_i$.}
    \label{fig:strong-signals:delay:cases:1}
\end{figure}

\begin{enumerate}
    \item If $S_i \le \lambda \cdot p_i$ and $S_j \le S_i$, then $j$ is not processed again once it emits its signal until $i$ completes, because we already argued that $j$ can only be boosted once $i$ is completed.
    Hence,
    $$
    d(i,j)+d(j,i) \le p_i + S_j.
    $$
    \item If $S_i \le \lambda \cdot p_i$ and $ \lambda \cdot p_i \ge S_j \ge S_i$,  then $j$ is not processed again once it emits its signal until $i$ completes for the same reason as in the previous case. Hence,
    $$
    d(i,j)+d(j,i) \le p_i + S_j.
    $$
    \item If $S_i \le \lambda \cdot p_i$ and $ \lambda \cdot p_i \le S_j$, then we can observe that $j$ is not processed after its elapsed time reaches $\lambda \cdot p_i$ until $i$ completes: If $j \in U(t)$, $e_j(t) \ge \lambda \cdot p_i$ and $i \in K(t)$, then $i$ (or smaller known jobs than $i$) will be boosted to completion before the algorithm returns to round-robin. 
    Hence,
    $$
    d(i,j)+d(j,i) \le p_i + \lambda p_i.
    $$

    \item If $S_i \ge \lambda \cdot p_i$ and $S_j \ge S_i$, then $j$ will not be processed after $i$ emits its signal until $i$ completes for the same reason as in the previous case. The delay is
    $$
    d(i,j) + d(j,i) = p_i + S_i.
    $$
    
    \item If $S_i \ge \lambda \cdot p_i$ and $S_j \le S_i$, then $j$ will not be processed after it emits its signal until $i$ completes. The delay is
    $$
    d(i,j) + d(j,i) = p_i + S_j.
    $$
\end{enumerate}
Based on these cases, we can upper-bound the expected pairwise delay:
\begin{align*}
    \ex[d(i,j) + d(j,i)] &\le \int_{S_i = 0}^{\lambda p_i} \frac{1}{p_i} \cdot \left( \int_{S_j = 0}^{\lambda p_i } \frac{p_i + S_j}{p_j} dS_j +  \int_{S_j = \lambda p_i}^{p_j} \frac{p_i + \lambda p_i}{p_j} dS_j  \right) dS_i\\
    &+ \int_{S_i = \lambda p_i}^{p_i} \frac{1}{p_i} \left( \int_{S_j = 0}^{S_i} \frac{p_i + S_j}{p_j} dS_j +  \int_{S_j = S_i}^{p_j} \frac{p_i + S_i}{p_j} dS_j  \right) dS_i\\
    &= \lambda p_i(1+\lambda)-\frac{\lambda^3 p_i^2}{2p_j} + \frac{p_i}{2}(3 - 2\lambda - \lambda^2)-\frac{p_i^2}{6p_j}(1 - \lambda^3)\\
    &= \frac{3 + \lambda^2}{2} \cdot p_i - \frac{(1 + 2\lambda^3)}{6} \frac{p_i^2}{p_j}.
\end{align*}

\textbf{Case 2:} Assume that $p_j \le \frac{p_i}{\lambda}$. In this case, it might happen that $j$ completes before $i$ and $j$ might be boosted before $i$ completes. As before,  we continue by analyzing the pairwise delay between $i$ and $j$ for different configurations of $S_i$ and $S_j$ (cf.~\Cref{fig:strong-signals:delay:cases:2} for an illustration of the configurations). First, we can observe that if $S_i \le \lambda p_j$, then we end up with essentially the same five cases as before:
\begin{figure}[t]
    \centering
    \begin{tikzpicture}
        \draw[thick] (0,0) rectangle (9,6);
        \draw[thick] (0,0) -- (6,6);
        \draw[thick] (0,4.5) -- (9,4.5);
        \draw[thick] (4.5,4.5) -- (4.5,6);
        \draw[thick] (3,0) -- (3,4.5);
        \draw[thick] (3,3) -- (9,3);

\node at (-1,3){\large$S_i$};
    
          \node at (-0.25,-0.25){0};
          \node at (-0.25,6){$p_i$};
          \node at (-0.35,3){$\lambda p_i$};
          \node at (-0.35,4.5){$\lambda p_j$};

\node at (4.5,-1){\large$S_j$};
      
          \node at (6,-0.25){$p_i$};
          \node at (3,-0.25){$\lambda p_i$};
          \node at (4.5,-0.25){$\lambda p_j$};
          \node at (9,-0.25){$p_j$};

\node at (1.5,3){1.};
          \node at (2,1){2.};
          \node at (6,1.5){3.};
          \node at (6,3+0.75){4.};
          \node at (3+0.5,3+1){5.};
          \node at (2.25,4.5+0.75){6.};
          \node at (4.5+0.5,4.5+1){7.};
          \node at (4.5+2.25,4.5+0.75){8.};
    \end{tikzpicture}
    \caption{Illustration of the different subcases for $p_j \le \frac{p_i}{\lambda}$ in the proof of~\Cref{lem:boosting-case-delays}. The figure assumes $\lambda = 0.5$ and $p_j = 1.5 \cdot p_i$.}
    \label{fig:strong-signals:delay:cases:2}
\end{figure}
\begin{enumerate}
    \item If $S_i \le \lambda \cdot p_i$ and $S_j \le S_i$, then $j$ is not processed again once it emits its signal until $i$ completes. Hence,
    $$
    d(i,j)+d(j,i) \le p_i + S_j.
    $$
    \item If $S_i \le \lambda \cdot p_i$ and $ \lambda \cdot p_i \ge S_j \ge S_i$,  then $j$ is not processed again once it emits its signal until $i$ completes. Hence,
    $$
    d(i,j)+d(j,i) \le p_i + S_j.
    $$
    \item If $S_i \le \lambda \cdot p_i$ and $ \lambda \cdot p_i \le S_j$, then we can observe that $j$ is not processed after its elapsed time reaches $\lambda \cdot p_i$ until $i$ completes: If $j \in U(t)$, $e_j(t) \ge \lambda \cdot p_i$ and $i \in K(t)$, then $i$ (or smaller known jobs than $i$) will be boosted to completion before the algorithm returns to round-robin. 
    Hence,
    $$
    d(i,j)+d(j,i) \le p_i + \lambda p_i.
    $$

    \item If $\lambda \cdot p_j \ge S_i \ge \lambda \cdot p_i$ and $S_i \le S_j$, then $j$ will not be processed after $i$ emits its signal until $i$ completes.  Hence,
    $$
    d(i,j) + d(j,i) = p_i + S_i.
    $$

    \item If $\lambda \cdot p_j \ge S_i \ge \lambda \cdot p_i$ and $S_i \ge S_j$, then $j$ will not be processed after it emits its signal until $i$ completes. Hence,
    $$
    d(i,j) + d(j,i) = p_i + S_j.
    $$

\end{enumerate}
It remains to analyze the pairwise delay if $S_i \ge \lambda p_j$. In this case, whichever job emits its signal first will also complete first:
\begin{enumerate}
    \setcounter{enumi}{5}
    \item If $S_i \ge \lambda p_j$ and $S_j \le \lambda p_j$, then $i$ is not processed after its elapsed time reaches $\lambda p_j$ until $j$ completes: If $i \in U(t)$, $e_i(t) \ge \lambda \cdot p_j$ and $j \in K(t)$, then $j$ (or smaller known jobs than $j$) will be boosted to completion before the algorithm returns to round-robin. 
    In this case, the pairwise delay is
    $$
    d(i,j)+d(j,i) = p_j + \lambda p_j.
    $$
    \item If $S_i \ge \lambda p_j$ and $S_i \ge S_j \ge \lambda p_j$, then $i$ is not processed after its elapsed time reaches $S_j$ until $j$ completes.  In this case, the pairwise delay is
    $$
    d(i,j)+d(j,i) = p_j + S_j.
    $$
    \item If $S_i \ge \lambda p_j$ and $S_i \le S_j$, $j$ is not processed after its elapsed time reaches $S_i$ until $i$ completes. In this case, the pairwise delay is
    $$
    d(i,j)+d(j,i) = p_i + S_i.
    $$
\end{enumerate}

The expected pairwise delay is upper-bounded by
\begin{align*}
    \ex[d(i,j) + d(j,i)] &\le \int_{S_i = 0}^{\lambda p_i} \frac{1}{p_i} \cdot \left( \int_{S_j = 0}^{\lambda p_i } \frac{p_i + S_j}{p_j} dS_j +  \int_{S_j = \lambda p_i}^{p_j} \frac{p_i + \lambda p_i}{p_j} dS_j  \right) dS_i\\
    &+ \int_{S_i = \lambda p_i}^{\lambda p_j} \frac{1}{p_i} \left( \int_{S_j = 0}^{S_i} \frac{p_i + S_j}{p_j} dS_j +  \int_{S_j = S_i}^{p_j} \frac{p_i + S_i}{p_j} dS_j  \right) dS_i\\
    &+ \int_{S_i = \lambda p_j}^{p_i} \frac{1}{p_i} \cdot \left(\int_{S_j = 0}^{\lambda p_j} \frac{p_j + \lambda p_j}{p_j} dS_j + \int_{S_j = \lambda p_j}^{S_i} \frac{p_j+S_j}{p_j} dS_j + \int_{S_j = S_i}^{p_j} \frac{p_i+S_i}{p_j} dS_j\right) dS_i\\
    &=  \lambda p_i(1+\lambda)-\frac{\lambda^3 p_i^2}{2p_j} + \lambda (p_j - p_i)
    + \frac{\lambda^2}{2p_i}(p_j^2 - p_i^2)
    - \frac{\lambda^3}{6p_i p_j}(p_j^3 - p_i^3)\\
    &+ 2 p_i + (\lambda^2 - \lambda)p_j - (\lambda^2 + \frac{\lambda^3}{3})\frac{p_j^2}{p_i} - \frac{2}{3}\frac{p_i^2}{p_j}\\
    &= \lambda^2 p_j+\frac{4 + \lambda^2}{2}p_i-\frac{\lambda^2 + \lambda^3}{2}\frac{p_j^2}{p_i}-\frac{\lambda^3 + 2}{3}\,\frac{p_i^2}{p_j}  
\end{align*}
\end{proof}

The following lemma is implied by~\Cref{lem:boosting-case-delays} and states the expected pairwise delay as a function of $\nicefrac{p_i}{p_j}$.

\begin{restatable}{lemma}{lemStringSignalDelays}
    \label{lem:revealing-signals:delays}
    Consider two jobs $i$ and $j$ with $p_i \le p_j$. Then, for any $0<\lambda\le 1$, $\ex[d^{\LBOOSTING}(i,j)+d^{\LBOOSTING}(j,i)] \leq p_i g_\lambda\left(\frac{p_i}{p_j}\right)$, where $g_\lambda : [0,1]\rightarrow \Rbb_+$ satisfies
    \begin{itemize}
        \item[] If $x \le \lambda$, then  $g_\lambda(x) = \frac{3 + \lambda^2}{2} - \frac{(1 + 2\lambda^3)}{6} x$.
        \item[] If $\lambda \le x$, then  $g_\lambda(x) = -\frac{\lambda^3 + 2}{3}x+\frac{4 + \lambda^2}{2}+\lambda^2 x^{-1}-\frac{\lambda^2 + \lambda^3}{2}x^{-2}$.
    \end{itemize}
\end{restatable}

Our goal is to apply~\Cref{lem:strong-delay} to prove~\Cref{thm:alg-strong}. However, $g_{\lambda}$ is not nonincreasing and concave on $[0,1]$:

\begin{itemize}
    \item If $x \leq \lambda$, then $g_\lambda(x)$ is linear and decreasing.
    \item If $x \geq \lambda$, then we can observe that  $g_\lambda(x)$ is strictly concave in $[\lambda,1]$ and has exactly one maximum point: It has derivative
    \[
     g'_\lambda(x) = - \frac{\lambda^2}{x^2} + \frac{\lambda^2 + \lambda^3}{x^3} - \frac{\lambda^3 + 2}{3} \ . 
    \]
    and second derivative
    \[
     g''_\lambda(x) = \frac{2\lambda^2}{x^3} - \frac{3(\lambda^2 + \lambda^3)}{x^4} = \frac{\lambda^2(2x-3-3\lambda)}{x^4}
    \]
    Note that $g''_\lambda(x) < 0$ for all $x \in [\lambda,1]$, so $g_{\lambda}$ is strictly concave in that interval.

    Moreover $g'_\lambda(\lambda) > 0$ and $g'_\lambda(1) < 0$, so $g_{\lambda}$ has exactly one maximum point $x^* \in [\lambda,1]$.
\end{itemize}

Next, consider the intercept of the tangent of $g_\lambda(x)$ at $x=0$, which can be written as
\[
    T_\lambda(x) := g_\lambda(x) - x g'_\lambda(x)
\]
Since $T'_\lambda(x) = -x g''_\lambda(x) > 0$, $T_\lambda(x)$ is strictly increasing in $x$.

We now have two cases. If $g_\lambda(x^*) \geq a:= g_\lambda(0)$, the best tangent point that leads to a nonincreasing concave function is $x_{\mathrm{tan}} = x^*$.
If $g_\lambda(x^*) \leq a$, let $x_{\mathrm{tan}}$ be the solution to $T_\lambda(x) = g_\lambda(0)$, which simplifies to
\[
    x^2 + 4\lambda^2x - 3\lambda^2(1+\lambda) = 0 
\]
so
\[
    x_{\mathrm{tan}} = \lambda \big( \sqrt{4\lambda^2 + 3\lambda + 3} - 2\lambda \big) \ .
\]
Note that $x_{\mathrm{tan}} \leq 1$ if 
\begin{align*}
\lambda \big( \sqrt{4\lambda^2 + 3\lambda + 3} - 2\lambda \big) &\leq 1 \\
\sqrt{4\lambda^2 + 3\lambda + 3} &\leq 1/\lambda + 2\lambda \\
    4\lambda^2 + 3\lambda + 3   &\leq 1 / \lambda^2 + 4 + 4\lambda^2 \\
    3\lambda^3 - \lambda^2 &\leq 1
\end{align*}
which holds if $\lambda \leq 0.824123$.

So in both cases the corresponding majorant is
\[
\hat g_\lambda(x) = \begin{cases}
    g_\lambda(0) + g'_\lambda(x_{\mathrm{tan}}) x & \text{if } x \leq x_\mathrm{tan} \\
     g_{\lambda}(x) \quad& \text{if } x \geq x_\mathrm{tan} \\
\end{cases}
\]
where
\[
x_{\mathrm{tan}} = \begin{cases}
    x^* & \text{if } g_\lambda(x^*) \geq g_\lambda(0) \\
    \lambda \big( \sqrt{4\lambda^2 + 3\lambda + 3} - 2\lambda \big) & \text{if } g_\lambda(x^*) \leq g_\lambda(0) \\
\end{cases}
\]

Numerically, it appears that the optimum is again attained when $g_\lambda(0) = g_\lambda(x^*)$. In that case, the flat first piece is optimal. \Cref{fig:strong-signals} illustrates representative choices of $\lambda$; in the left and middle panels the support point is $x^*$, and the marked segment is flat.

As argued above, $\hat{g}_{\lambda}$ is concave and nonincreasing. Hence, we can prove~\Cref{thm:alg-strong} with the following lemma.

\begin{restatable}{lemma}{strongsignal}
\label{lem:strong-signal}
For $\lambda = 0.4699$, $g_\lambda(0) \leq g_\lambda(x_\lambda)$ and for $\hat g_\lambda$ defined in Equation \eqref{eq:hat2}, we have $\int_0^1 \hat g_{\lambda}(t^2)\,dt < 1.59999$.
\end{restatable}

\begin{proof}
Fix $\lambda=0.4699$. For $x\in[\lambda,1]$, write
\[
g_\lambda(x)
=
Ax+B+Cx^{-1}+Dx^{-2},
\]
where
\[
A=-\frac{2+\lambda^3}{3},\qquad
B=\frac{4+\lambda^2}{2},\qquad
C=\lambda^2,\qquad
D=-\frac{\lambda^2+\lambda^3}{2}.
\]
For this value of $\lambda$ these constants are
\[
A=-0.701252248033,\quad
B=2.110403005,\quad
C=0.22080601,\quad
D=-0.1622813770495 .
\]

The maximizer $x_\lambda$ of $g_\lambda$ on $[\lambda,1]$ is the unique root of
\[
q_\lambda(x)
:=
(2+\lambda^3)x^3+3\lambda^2x-3\lambda^2(1+\lambda)=0.
\]
Indeed,
\[
g_\lambda'(x)
=
-\frac{q_\lambda(x)}{3x^3},
\]
and
\[
q_\lambda'(x)=3(2+\lambda^3)x^2+3\lambda^2>0.
\]
Thus $q_\lambda$ is strictly increasing, so its root is unique. Direct exact rational
evaluation gives
\[
q_\lambda(0.63946184)
=
-2.7983272752892265\cdot 10^{-8}<0,
\]
and
\[
q_\lambda(0.63946185)
=
4.4484142438288805\cdot 10^{-9}>0.
\]
Therefore
\[
x_\lambda\in [a,b]
:=
[0.63946184,\,0.63946185].
\]

We next verify the condition $g_\lambda(0)\le g_\lambda(x_\lambda)$. Since
$q_\lambda(a)<0$, we have $g_\lambda'(a)>0$, and therefore $g_\lambda$ is increasing
on $[a,x_\lambda]$. Hence
\[
g_\lambda(x_\lambda)\ge g_\lambda(a).
\]
A direct exact evaluation gives
\[
g_\lambda(0)=\frac{3+\lambda^2}{2}=1.610403005,
\]
whereas
\[
g_\lambda(a)
=
1.6104167866212743\ldots
>
1.610403005.
\]
Thus $g_\lambda(0)\le g_\lambda(x_\lambda)$.

It remains to bound the integral. Since $\hat g_\lambda(x)=g_\lambda(x_\lambda)$
for $x\le x_\lambda$ and $\hat g_\lambda(x)=g_\lambda(x)$ for $x\ge x_\lambda$,
we have
\[
\rho_\lambda
:=
\int_0^1 \hat g_\lambda(t^2)\,dt
=
\sqrt{x_\lambda}\,g_\lambda(x_\lambda)
+
\int_{\sqrt{x_\lambda}}^1 g_\lambda(t^2)\,dt .
\]
Using the formula for $g_\lambda$ on $[\lambda,1]$ and integrating explicitly,
\[
\rho_\lambda
=
K
+
\frac{2A}{3}x_\lambda^{3/2}
+
2C x_\lambda^{-1/2}
+
\frac{4D}{3}x_\lambda^{-3/2},
\]
where
\[
K:=\frac{A}{3}+B-C-\frac{D}{3}
=
1.7099400380055.
\]

We now bound this expression using the interval
\[
x_\lambda\in[a,b]=[0.63946184,0.63946185].
\]
The following square-root enclosures can be checked numerically:
\[
\sqrt a\in[0.79966357926,0.79966357927],
\qquad
\sqrt b\in[0.79966358551,0.79966358552].
\]
Since $A<0$, $C>0$, and $D<0$, we obtain the upper bound
\[
\begin{aligned}
\rho_\lambda
&\le
K
+
\frac{2A}{3}\,a\cdot 0.79966357926
+
\frac{2C}{0.79966357926}
+
\frac{4D}{3}\cdot
\frac{1}{b\cdot 0.79966358552}  \\
&<
1.599987032 \\
&< 1.599988 \ .
\end{aligned}
\]
This proves the claimed bound.
\end{proof}

\begin{figure}
	\centering
	\includegraphics[width=0.8\textwidth]{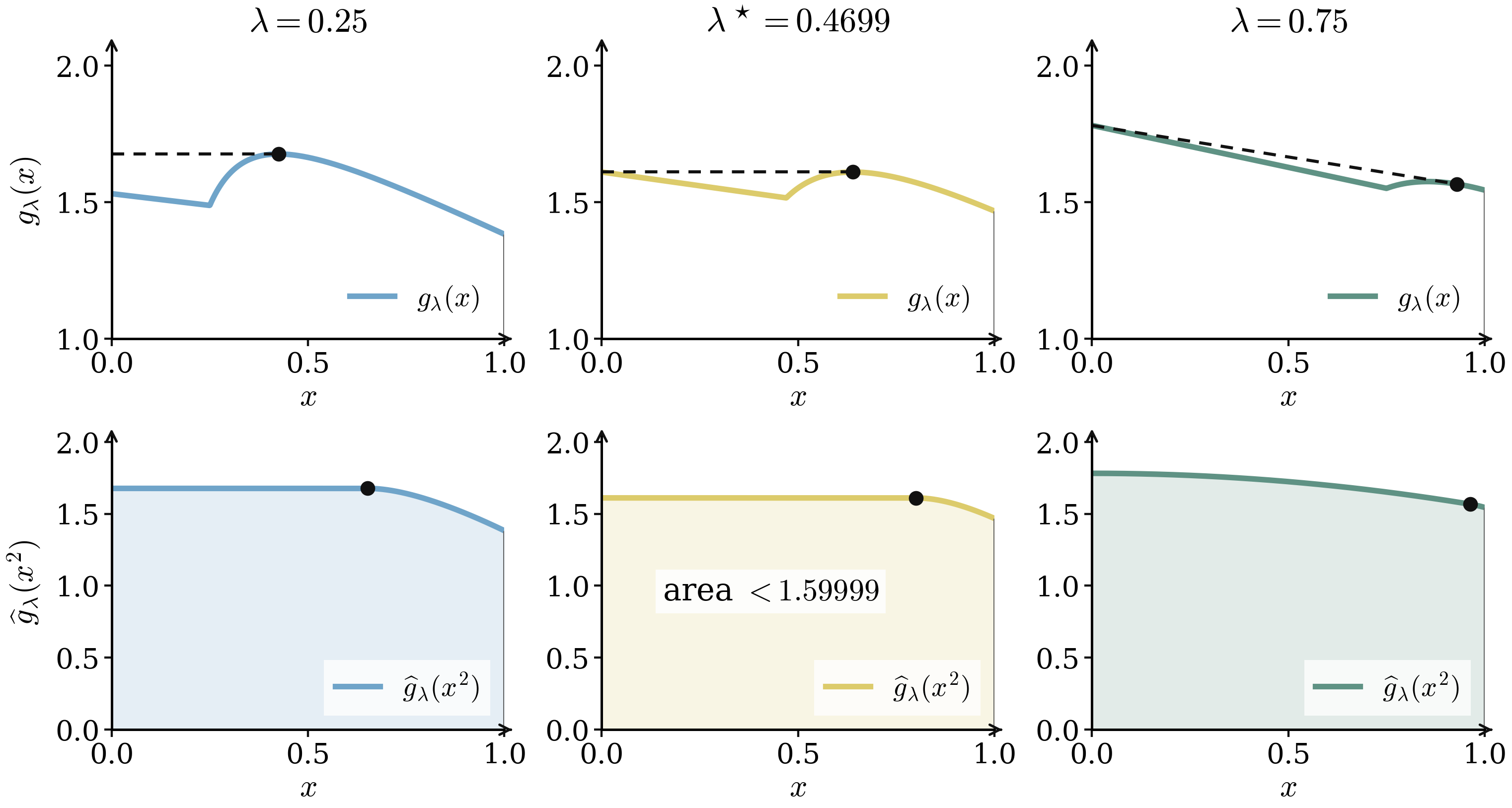}
	\caption{Visualization of the normalized pairwise expected delay of $\LBOOSTING$ and its nonincreasing concave majorant for different values of $\lambda$.}
	\label{fig:strong-signals}
\end{figure}

\begin{proof}[We are now ready to prove~\Cref{thm:alg-strong}]
The normalized expected pairwise delay is bounded by the function $g_\lambda:[0,1]\rightarrow \Rbb_+$ defined in \Cref{lem:revealing-signals:delays}. The function $g_\lambda$ is linear and decreasing for $x\in[0,\lambda]$ and strictly concave for $x\in[\lambda,1]$, maximized by $x_\lambda$, the solution of $(2 + \lambda^3) x_\lambda^3 + 3\lambda^2x_\lambda - 3\lambda^2(1+\lambda) = 0$. Thus the function is decreasing on $(0,\lambda]$, increasing on $[\lambda,x_\lambda]$, decreasing on $[x_\lambda,1]$.  We define 
\begin{equation}\label{eq:hat2}
    \hat g_\lambda(x) = \begin{cases}
    g_{\lambda}(x_\lambda) \quad& \text{if } x \leq x_\lambda \\
    g_{\lambda}(x) \quad& \text{if } x \geq x_\lambda \\
\end{cases}
\end{equation}
which is nonincreasing and concave. Furthermore, the function satisfies $\forall x\in [0,1] : g_\lambda(x) \leq \hat g_\lambda(x)$ assuming that $g_\lambda(0)\leq g_\lambda(x_\lambda)$. Hence, by~\Cref{lem:strong-delay}, the competitive ratio of $\LBOOSTING$ satisfies $\rho_\lambda \le \int_0^1 \hat g_{\lambda}(t^2)\,dt$ under the same assumption.
\Cref{lem:strong-signal} then implies the claimed bound.
An illustration of these functions is given in \Cref{fig:strong-signals}.
\end{proof}

\paragraph{Analysis of $\LBOOSTING$ via Local Delay Theorem.} For a fixed value of $\lambda$, if we want to analyze the competitiveness of $\LBOOSTING$ via the local delay theorem, then the best bound we can obtain is $\alpha_{\lambda} := \max_{x \in [0,1]} g_{\lambda}(x)$. Optimizing over $\lambda$, we can then obtain the bound $\alpha^* := \min_{\lambda \in (0,1]} \alpha_{\lambda}$. Numerical computations show that $\alpha^* \approx 1.6104$ is obtained at $\lambda \approx 0.4699$, indicating that the Global Delay Theorem is necessary to obtain bounds $< 1.6$.

\section{Appendix for \Cref{sec:lower-bounds} (Lower Bounds)}\label{sec:lb-support}
\subsection{Basic results on the Gamma distribution}\label{sec:basic-gamma}
In this section, we review some basic results about the Gamma distribution.

\subparagraph*{Exponential distribution.} For $\lambda > 0$, we say that a real-valued random variable $X$ has an exponential distribution with parameter $\lambda$, and write $X \sim \Ecal(\lambda),$ 
if its density is given for $x\geq 0$ by $
f_X(x) = \lambda e^{-\lambda x}$. Its expectation is $\Ebb[X] = \frac{1}{\lambda}$. The exponential distribution is memoryless in the sense that, for all $s,t\ge 0$,
\[
\Pbb(X\ge s+t \mid X\ge s)
=\frac{\Pbb(X\ge s+t)}{\Pbb(X\ge s)}
=\frac{e^{-\lambda(s+t)}}{e^{-\lambda s}}
=e^{-\lambda t}
=\Pbb(X\ge t).
\]
Equivalently, the conditional distribution of $X$ given $X\ge s$ is an exponential distribution shifted by $s$, i.e., $s+\Ecal(\lambda)$. 

\subparagraph*{Gamma distribution.}
For $k>0$ and $\lambda>0$, we say that a random variable $X$ has a Gamma distribution with shape parameter $k$ and rate parameter $\lambda$, and write $X \sim \Gamma(k,\lambda),$ if its density is given for $x\geq 0$ by
$f_X(x)=\frac{\lambda^k}{\Gamma(k)}x^{k-1}e^{-\lambda x},$
where $\Gamma(k)=\int_0^\infty t^{k-1}e^{-t}\,dt$ denotes the Gamma function. Its expectation is $\Ebb[X]=\frac{k}{\lambda}$. In the special case $k=2$, this simplifies to $f_X(x)=\lambda^2 x e^{-\lambda x}$ and $\Ebb[X]=\frac{2}{\lambda}$. For $k\in \Nbb$, the sum of $k$ i.i.d. random variables following $\Ecal(\lambda)$ follows a $\Gamma(k,\lambda)$ distribution.

\subparagraph*{Minimum of two Gamma random variables.} Let $X_1,X_2$ be two independent $\Gamma(k=2,\lambda)$ random variables, and define $M = \min\{X_1,X_2\}$. We consider the survival function $F_{X_1}(t)=\Pbb(X_1\geq t) = \int_{t}^\infty \lambda^2x\exp(-\lambda x)dx = \exp(-\lambda t)(1+\lambda t)$. Thus,
$$\Ebb(M) = \int_{0}^\infty\Pbb(M\geq t)dt =\int_{0}^\infty F_{X_1}(t)^2dt = \int_0^\infty\exp(-2\lambda t)(1+\lambda t)^2dt = \frac{5}{4\lambda} \ ,$$
where the last identity uses that $\int_0^\infty t^n\exp(-at)dt = \frac{n!}{a^{n+1}}$ so that we can decompose into the sum $\frac{1}{2\lambda} + \frac{2\lambda}{(2\lambda)^2}+\frac{\lambda^2 2!}{(2\lambda)^3} = \frac{2+2+1}{4\lambda}$.

\subparagraph*{Conditioning of the Gamma distribution.} Let $X\sim \Gamma(k=2,\lambda)$ and let $U$ be sampled uniformly from $[0,X]$. Then $U\sim \Ecal(\lambda)$. Indeed, the marginal density of $U$ is 
$$f_U(u) = \int_{x\in\Rbb_+}f_{U|X}(u|x)f_X(x)dx = \int_{x\geq u} \frac{1}{x}\times \lambda^2 x\exp(-\lambda x)dx = \lambda\exp(-\lambda u).$$
Also, computing the conditional distribution 
\[
f_{X|U}(x|u) = \frac{f_{U,X}(u,x)}{f_{U}(u)} = \frac{\lambda^2\exp(-\lambda x)}{\lambda\exp(-\lambda u)} = \lambda\exp(-\lambda(x-u)),
\]
one obtains that the conditional distribution of $X$ given $U=u$ is an exponential law shifted by $u$, i.e., $u+\Ecal(\lambda)$. Since the exponential law is memoryless, we further have that the conditional distribution of $X$ given $U=u$ and $X\geq x$ is the exponential law $x+\Ecal(\lambda)$.

\subparagraph*{Order statistics of the exponential distribution.} We assume $X_1,\dots, X_n$ are i.i.d. random variables following $\Ecal(\lambda)$ and denote by $X_{(i)}$ the $i$-th smallest variable. We have that $X_{(1)}\sim \Ecal(n\lambda)$ as the minimum of $n$ exponential random variables. Then, observing that the process is memoryless, and conditioning on all other variables to exceed $X_{(1)}$, we get that $X_{(2)}-X_{(1)}\sim \Ecal((n-1)\lambda)$, and more generally, $X_{(i+1)}-X_{(i)}\sim \Ecal((n-i)\lambda)$. Thus, by linearity of expectation we get $\Ebb(X_{(i)}) = \frac{1}{\lambda}\left(\frac{1}{n}+\dots+\frac{1}{n-i+1}\right) = \frac{1}{\lambda}\left(H_n-H_{n-i}\right)$ where $H_n$ is the $n$-th harmonic number. This allows us to compute
\begin{align*}
    \sum_{i=1}^n \sum_{j=1}^i \Ebb(X_{(j)}) = \sum_{j=1}^n(n-j+1)\Ebb(X_{(j)}) = \sum_{r=1}^n r \sum_{k=r}^n \frac{1}{k} = \sum_{k=1}^n\frac{1}{k} \sum_{r=1}^k r 
    &= \sum_{k=1}^n\frac{1}{k} \frac{k(k+1)}{2} \\
    &= \frac{n(n+3)}{4}.
\end{align*}

\subsection{Characterizing optimal online algorithms}
The goal of this section is to prove the following result, which characterizes optimal online algorithms for the input distribution presented above. 
\begin{lemma}\label{lem:opt-online}
    For the scheduling with one weak uniform signal problem, where all processing times are sampled from independent Gamma random variables $p_i \sim \Gamma(k=2,\lambda)$, optimal online algorithms are those that satisfy the following two conditions:\footnote{These conditions could be violated on a set of times of Lebesgue measure zero or on an event of probability zero. In formal terms, we should say that the two conditions are satisfied `almost everywhere'.}
\begin{enumerate}
    \item \label{c:1}They process an unfinished job that has emitted a signal if there is one
    \item \label{c:2}They process arbitrary unfinished jobs otherwise
\end{enumerate}
\end{lemma}
\begin{proof}
We note that in the distributional setting (i.e., the distribution of all job sizes is known to the algorithm), the problem reduces to a continuous-time Markov Decision Process (MDP), where the policy corresponds to the jobs being processed, and the instantaneous cost of the algorithm equals the number of unfinished jobs. The value function for this Markov Decision Process can be made explicit via dynamic programming, as can the optimal online algorithms. Because introducing a continuous-time Markov Decision Process would require substantial formalism, we turn to a more direct argument, while noticing that further extensions of our lower bound would likely rely on analyzing a refined Markov Decision Process.

We first argue that all algorithms that satisfy conditions \ref{c:1} and \ref{c:2} have the same expected cost. We say that an event occurs if either (a) a job emits its signal or (b) a job is completed. We recall that, by arguments introduced in \Cref{sec:basic-gamma}, the time required for a job $j$ to emit its signal $S_j$ follows an exponential random variable $\Ecal(\lambda)$, and, further, that once a job $j$ has emitted its signal, the remaining processing time $P_j-S_j$ follows an exponential random variable $\Ecal(\lambda)$. We also recall that the exponential random variable is memoryless. Thus, any algorithm that has property \ref{c:2} satisfies that the elapsed time between two events follows an exponential random variable. If the algorithm further satisfies property \ref{c:1}, the events must be an alternation of (a) a signal emission and (b) a job completion. 

We now argue that algorithms that do not satisfy both conditions \ref{c:1} and \ref{c:2} are suboptimal. 
It is clear that no optimal algorithm idles unnecessarily (otherwise the total completion time can be strictly improved).
Thus, optimal algorithms must have the property that the elapsed time between two events is exponentially distributed.
Furthermore, it now suffices to show that optimal online algorithms satisfy condition~\ref{c:1} as this, together with the fact that such algorithms do not idle unnecessarily, implies that optimal online algorithms satisfy both conditions. Hence, we now argue that algorithms that do not idle unnecessarily but also do not satisfy condition \ref{c:1} are suboptimal. 

To this end, for $i\in [2n]$, denote by $A_i$ the time at which the $i$-th event occurs, and let $A_0=0$. Then 
\[
\ALG = \sum_{i=0}^{2n-1}(A_{i+1}-A_{i})n(i),
\]
where $n(i)$ is the number of unfinished jobs in the interval $[A_{i},A_{i+1})$, and must satisfy $\forall i\in [2n]: n(i) \geq n-\lfloor i/2 \rfloor$, with equality when condition \ref{c:1} is respected. Assume, by contradiction, that condition \ref{c:1} is not met. Then, with positive probability, the sequence of events is not an alternation of (a) a signal emission and (b) a job completion. Therefore, we have $n(i)>n-\lfloor i/2 \rfloor$ with positive probability for some $i\in [2n]$, implying the sub-optimality of the algorithm. 
\end{proof}

\section{Appendix for \Cref{sec:beta} (Beta Distributions and Multiple Signals)}
\label{ap:beta}

\subsection{Proof of the Local Delay Baseline (\Cref{prop:beta-weak})}
\label{ap:beta-weak}

We restate the expected pairwise delay $g(x)$ under Beta signals here for easy reference. With $X_k := S_k/p_k \sim \mathrm{Beta}(\alpha,\beta)$ and $F$ and $f$ denoting its CDF and PDF, 
\begin{equation}
    g(x) = 1 + \int_0^1 (1-F(t))(1-F(xt))\, dt + \left(\frac{1}{x} - 1\right) \int_0^1 F(xt)\, f(t)\, dt,\label{eq:g-beta}
\end{equation}
where $x = p_i/p_j$.

\begin{proof}[Proof of \Cref{prop:beta-weak}]

Fix two jobs $i,j$ with $p_i \leq p_j$ and let $x = p_i/p_j$. 
Let $D = d^\GREEDY(i,j) + d^\GREEDY(j,i)$, and let $X_i, X_j \sim \mathrm{Beta}(\alpha,\beta)$ be i.i.d.\ so that the signal times are $S_i = p_i X_i$ and $S_j = p_j X_j$.
Thus, $S_i \leq S_j$ if and only if $x X_i \leq X_j$. By the same case analysis as in the proof of \Cref{lem:pairwise-delay-greedy},
\[
\frac{D}{p_i} = \begin{cases}
    1 + X_i & \text{ if } x X_i \leq X_j, \\
    \dfrac{1+ X_j}{x} & \text{ otherwise.}
\end{cases}
\]
Let $\Ecal = \{x X_i > X_j\}$. Then
\[
\frac{D}{p_i} = 1 + X_i + \left( \frac{1+ X_j}{x} - (1+ X_i) \right) \cdot \one[\Ecal].
\]
Under $\Ecal$, we have
\[
\frac{1+ X_j}{x} - (1+ X_i) < \frac{1}{x} + X_i - (1+ X_i) = \frac{1}{x} - 1.
\]
Thus,
\[
    \frac{\ex[D]}{p_i}
    \leq 1 + \ex[X_i] + \left(\frac1x - 1 \right) \cdot \pr(\Ecal) 
    = 1 + \frac{\alpha}{\alpha+\beta} + \left(\frac1x - 1 \right) \cdot \pr(\Ecal).
\]

We analyze the last term of this bound. Under $\Ecal$ we have
\[
\frac{1}{x} - 1 \leq \frac{X_i}{X_j} - 1 = \frac{X_i - X_j}{X_j}.
\]
Thus, for all $x \in (0,1]$,
\[
\left(\frac1x - 1 \right) \cdot \pr(\Ecal)
= \ex \!\left[ \left(\frac1x - 1 \right) \cdot \one[\Ecal] \right]
\leq \ex \!\left[ \frac{X_i - X_j}{X_j} \cdot \one[\Ecal] \right]
\leq \ex \!\left[ \frac{\max\{X_i - X_j,\,0\}}{X_j} \right].
\]
By Cauchy–Schwarz, this is at most
\[
\left( \ex[(\max\{X_i - X_j,\,0\})^2] \cdot \ex[X_j^{-2}] \right)^{1/2}.
\]
Since $X_i$ and $X_j$ are i.i.d.\ with the same mean,
\[ \ex[(\max\{X_i - X_j,0\})^2] =  \frac12 \ex[(X_i - X_j)^2] = \Var(X_i) = \frac{\alpha \beta}{(\alpha+\beta)^2(\alpha+\beta+1)},
\]
and for all $\alpha > 2$,
\[
\ex[X_i^{-2}]
= \frac{B(\alpha-2,\beta)}{B(\alpha,\beta)}
= \frac{(\alpha+\beta-1)(\alpha+\beta-2)}{(\alpha-1)(\alpha-2)}.
\]

So in total,
\[
 \frac{\ex[D]}{p_i} 
 \leq 1 + \frac{\alpha}{\alpha+\beta}
 + \sqrt{\frac{\alpha\beta(\alpha+\beta-1)(\alpha+\beta-2)}{(\alpha+\beta)^2(\alpha+\beta+1)(\alpha-1)(\alpha-2)}}.
\]
Since this bound is independent of $x$, it translates to a bound on the competitive ratio even via the Local Delay Theorem (\Cref{lem:weak-delay}). 
\end{proof}

\subsection{The Polynomial Envelope $U(x)$}
\label{ap:beta-envelope}

\begin{lemma}\label{lem:cdf-bounds}
    For any $t \in [0,1]$, $x \in (0,1]$, and parameters $\alpha \geq 1$, $\beta \geq 1$, the CDF of $\mathrm{Beta}(\alpha,\beta)$ satisfies
    \begin{align}
        F(xt) &\leq \frac{x^\alpha t^\alpha}{\alpha\, B(\alpha,\beta)}, \label{eq:cdf-upper}\\
        F(xt) &\geq x^\alpha F(t). \label{eq:cdf-lower}
    \end{align}
\end{lemma}

\begin{proof}
    For~\eqref{eq:cdf-upper}, since $\beta \geq 1$, $(1-u)^{\beta-1} \leq 1$ for $u \in [0,xt]$, so
    \[
        F(xt) = \int_0^{xt} \frac{u^{\alpha-1}(1-u)^{\beta-1}}{B(\alpha,\beta)}\, du \leq \int_0^{xt} \frac{u^{\alpha-1}}{B(\alpha,\beta)}\, du = \frac{x^\alpha t^\alpha}{\alpha\, B(\alpha,\beta)}.
    \]
    For~\eqref{eq:cdf-lower}, the substitution $u = xv$ gives $F(xt) = x^\alpha \int_0^{t} \frac{v^{\alpha-1}(1-xv)^{\beta-1}}{B(\alpha,\beta)}\, dv$. Since $x \leq 1$, $1-xv \geq 1-v$, and $\beta \geq 1$ implies $(1-xv)^{\beta-1} \geq (1-v)^{\beta-1}$. The integral is thus lower-bounded by $F(t)$, so $F(xt)$ is lower-bounded by $x^\alpha F(t)$.
    Both inequalities are tight at $\beta = 1$, where $(1-u)^{\beta-1} \equiv 1$.
\end{proof}

\begin{lemma}\label{lem:envelope}
    For any $\alpha \geq 1$ and $\beta \geq 1$, the normalized expected pairwise delay $g(x)$ defined in~\eqref{eq:g-beta} satisfies
    \[
        g(x) \leq U(x) := 1 + \mu + C x^{\alpha-1} - (C + \Delta)\, x^\alpha,
    \]
    where $\mu = \frac{\alpha}{\alpha+\beta}$, $C = \frac{B(2\alpha,\beta)}{\alpha\, B(\alpha,\beta)^2}$, and $\Delta = \int_0^1 F(t)(1-F(t))\, dt$.
\end{lemma}

\begin{proof}
    Let $I_1(x) = \int_0^1 (1-F(t))(1-F(xt))\, dt$ and $I_2(x) = \int_0^1 F(xt) f(t)\, dt$. Using~\eqref{eq:cdf-upper}:
    \[
        I_2(x) \leq \int_0^1 \frac{x^\alpha t^\alpha}{\alpha\, B(\alpha,\beta)} f(t)\, dt = \frac{x^\alpha}{\alpha\, B(\alpha,\beta)} \int_0^1 t^\alpha f(t)\, dt = C x^\alpha.
    \]
    Using~\eqref{eq:cdf-lower}, $1 - F(xt) \leq 1 - x^\alpha F(t)$, hence
    \[
        I_1(x) \leq \int_0^1 (1-F(t))(1 - x^\alpha F(t))\, dt = \mu - x^\alpha \Delta.
    \]
    Substituting back into $g(x)$, and using $(1/x - 1) \geq 0$ for $x \in (0,1]$:
    \[
        g(x) \leq 1 + (\mu - x^\alpha \Delta) + \left(\frac{1}{x} - 1\right) C x^\alpha = 1 + \mu + C x^{\alpha-1} - (C+\Delta) x^\alpha = U(x). \qedhere
    \]
\end{proof}

\paragraph{Shape of $U(x)$ for $\alpha > 1$.}

\begin{lemma}\label{lem:U-shape}
    Let $\alpha > 1$ and $\beta \geq 1$, and define
    \[
        x_U^* := \frac{\alpha-1}{\alpha} \cdot \frac{C}{C+\Delta}, \qquad x_{\mathrm{inf}} := \frac{\alpha-2}{\alpha} \cdot \frac{C}{C+\Delta}.
    \]
    Then $x_U^*$ is the unique root of $U'(x) = 0$ on $(0,1]$ and $U''(x_U^*) < 0$, so $x_U^*$ is the unique maximizer of $U$. If $\alpha > 2$, $x_{\mathrm{inf}}$ is the unique inflection point of $U$ on $(0,1]$ and $x_{\mathrm{inf}} < x_U^*$; if $\alpha \leq 2$, $U$ is concave on $(0,1]$. In particular, $U$ is nonincreasing and concave on $[x_U^*, 1]$.
\end{lemma}

\begin{proof}
    The first and second derivatives are
    \[
        U'(x) = C(\alpha-1)\, x^{\alpha-2} - (C+\Delta)\alpha\, x^{\alpha-1}, \quad
        U''(x) = C(\alpha-1)(\alpha-2)\, x^{\alpha-3} - (C+\Delta)\alpha(\alpha-1)\, x^{\alpha-2}.
    \]
    Setting $U'(x) = 0$ and dividing by $x^{\alpha-2} > 0$ yields $x_U^* = \frac{\alpha-1}{\alpha}\cdot\frac{C}{C+\Delta} \in (0,1)$. Using $(C+\Delta)\alpha\, x_U^* = C(\alpha-1)$,
    \[
        U''(x_U^*) = (\alpha-1)(x_U^*)^{\alpha-3}\bigl[C(\alpha-2) - C(\alpha-1)\bigr] = -(\alpha-1)(x_U^*)^{\alpha-3}\, C < 0.
    \]
    Similarly, $U''(x) = 0$ yields $x_{\mathrm{inf}} = \frac{\alpha-2}{\alpha}\cdot\frac{C}{C+\Delta}$, which is positive only when $\alpha > 2$, and $x_U^* - x_{\mathrm{inf}} = \frac{C}{\alpha(C+\Delta)} > 0$.
\end{proof}

By \Cref{lem:U-shape}, the function
\begin{equation}\label{eq:ghat-from-U}
    \hat g(x) := \begin{cases} U(x_U^*) & \text{if } x \in [0, x_U^*], \\ U(x) & \text{if } x \in (x_U^*, 1], \end{cases}
\end{equation}
is nonincreasing and concave on $[0,1]$, and satisfies $\hat g(x) \geq U(x) \geq g(x)$ for all $x \in (0,1]$.

\begin{proof}[Proof of \Cref{thm:beta-case2-U}]
For $\alpha = 1$ the maximizer is $x_U^* = 0$, so the flat piece in~\eqref{eq:ghat-from-U} is empty and $\hat g = U$. Substituting $\alpha=1$ into \Cref{lem:envelope} gives $\mu = \frac{1}{\beta+1}$, $C = \frac{B(2,\beta)}{B(1,\beta)^2} = \frac{\beta}{\beta+1}$, and $\Delta = \int_0^1[(1-t)^\beta - (1-t)^{2\beta}]\,dt = \frac{\beta}{(\beta+1)(2\beta+1)}$, so $U(x) = 2 - \frac{2\beta}{2\beta+1}\,x$ is linear, nonincreasing, and concave. Applying \Cref{lem:strong-delay} with $\hat g = U$,
\[
    \rho_U = \int_0^1 U(t^2)\, dt = 2 - \frac{2\beta}{3(2\beta+1)},
\]
which recovers $16/9$ at $\beta=1$ and matches the general formula below at $x_U^* = 0$.

For $\alpha > 1$, applying \Cref{lem:strong-delay} with $\hat g$ from~\eqref{eq:ghat-from-U},
\[
    \rho_U = \int_0^1 \hat g(t^2)\, dt = \int_0^{\sqrt{x_U^*}} U(x_U^*)\, dt + \int_{\sqrt{x_U^*}}^1 U(t^2)\, dt = \sqrt{x_U^*}\cdot U(x_U^*) + \int_{\sqrt{x_U^*}}^1 U(t^2)\, dt.
\]
The second integral admits a closed form:
\[
    \int_{\sqrt{x_U^*}}^1 U(t^2)\, dt = (1+\mu)(1-\sqrt{x_U^*}) + \frac{C\bigl(1 - (x_U^*)^{\alpha-1/2}\bigr)}{2\alpha-1} - \frac{(C+\Delta)\bigl(1 - (x_U^*)^{\alpha+1/2}\bigr)}{2\alpha+1}.
\]
The formula for $U(x_U^*)$ follows from
\[
    U(x_U^*) = 1 + \mu + (x_U^*)^{\alpha-1}\bigl[C - (C+\Delta)\, x_U^*\bigr] = 1 + \mu + \frac{C}{\alpha}\,(x_U^*)^{\alpha-1},
\]
using $(C+\Delta)x_U^* = \frac{\alpha-1}{\alpha} C$. \qedhere
\end{proof}

\subsection{Tighter Bounds via Direct Analysis of $g$}
\label{ap:beta-direct-g}

The envelope $U(x)$ incurs a relaxation gap relative to $g(x)$, which becomes substantial for large $\alpha,\beta$. Here we bypass the envelope and construct $\hat g$ directly from $g$.

\subsubsection{Derivatives of $g$ and a Closed Form for $g''(1)$}

Differentiating~\eqref{eq:g-beta},
\begin{equation}\label{eq:g-prime}
    g'(x) = -\int_0^1 t\,(1-F(t))\, f(xt)\, dt - \frac{1}{x^2}\int_0^1 F(xt)\, f(t)\, dt + \left(\frac{1}{x} - 1\right)\int_0^1 t\, f(xt)\, f(t)\, dt,
\end{equation}
and
\begin{equation}\label{eq:g-double-prime}
    g''(x) = -\!\int_0^1 \! t^2 (1-F(t))\, f'(xt)\, dt + \frac{2}{x^3}\!\int_0^1\! F(xt) f(t)\, dt - \frac{2}{x^2}\!\int_0^1\! t\, f(xt) f(t)\, dt + \Bigl(\tfrac{1}{x}-1\Bigr)\!\int_0^1\! t^2 f'(xt) f(t)\, dt.
\end{equation}
These expressions do not admit closed-form simplifications for general $\alpha, \beta$; however, $g''(1)$ does.

\begin{lemma}\label{lem:g-double-prime-at-1}
    For $\alpha \geq 1$ and $\beta \geq 1$,
    \begin{equation}\label{eq:gpp1}
        g''(1) = 1 + \frac{\alpha}{\alpha+\beta} - \Delta - \frac{2\,B(2\alpha, 2\beta-1) + B(2\alpha+1, 2\beta-1)}{B(\alpha,\beta)^2}.
    \end{equation}
\end{lemma}

\begin{proof}
    At $x = 1$, the last term in~\eqref{eq:g-double-prime} vanishes, leaving
    
    \[
        g''(1) = -\int_0^1 t^2(1-F(t))\, f'(t)\, dt + 2\int_0^1 F(t)f(t)\, dt - 2\int_0^1 t\,f(t)^2\, dt.
    \]
    
    We evaluate each of the three integrals.

    \emph{First integral.} Using $f'(t) = \frac{(\alpha-1)t^{\alpha-2}(1-t)^{\beta-1} - (\beta-1)t^{\alpha-1}(1-t)^{\beta-2}}{B(\alpha,\beta)}$ and integration by parts with the boundary term $[t^2(1-F(t))f(t)]_0^1 = 0$,
    \[
        \int_0^1 t^2 (1-F(t))\, f'(t)\, dt = -2\!\int_0^1 t(1-F(t))f(t)\, dt + \int_0^1 t^2 f(t)^2\, dt.
    \]
    Write $\int_0^1 t(1-F(t))f(t)\, dt = \int_0^1 t f(t)\, dt - \int_0^1 t F(t) f(t)\, dt = \mu - \int_0^1 t F(t) f(t)\, dt$. Since $\frac{d}{dt}[F(t)^2/2] = F(t)f(t)$, integration by parts yields
    \[
        \int_0^1 t\, F(t)\, f(t)\, dt = \left[t\, F(t)^2/2\right]_0^1 - \frac{1}{2}\int_0^1 F(t)^2\, dt = \frac{1}{2} - \frac{1}{2}\int_0^1 F(t)^2\, dt.
    \]
    By the tail-integral identity, $\int_0^1 F(t)\, dt = 1 - \mu$; therefore
    \[
        \Delta = \int_0^1 F(t)(1-F(t))\, dt = (1-\mu) - \int_0^1 F(t)^2\, dt,
        \qquad\text{i.e.,}\qquad
        \int_0^1 F(t)^2\, dt = 1 - \mu - \Delta.
    \]
    Combining,
    \[
        2\int_0^1 t(1-F(t))f(t)\, dt = 2\mu - 2\!\left[\tfrac{1}{2} - \tfrac{1}{2}(1 - \mu - \Delta)\right] = 2\mu - 1 + (1 - \mu - \Delta) = \mu - \Delta,
    \]
    so the contribution of the first integral to $g''(1)$ is
    \[
        -\!\int_0^1 t^2 (1-F(t))\, f'(t)\, dt = (\mu - \Delta) - \int_0^1 t^2 f(t)^2\, dt.
    \]

    \emph{Second integral.} $\int_0^1 F(t)f(t)\, dt = \left[F(t)^2/2\right]_0^1 = \tfrac{1}{2}$, so $2\int_0^1 F(t)f(t)\, dt = 1$.

    \emph{Third integral.} Using the identity $\int_0^1 t^a (1-t)^b\, dt = B(a+1, b+1)$,
    \[
        \int_0^1 t\, f(t)^2\, dt = \frac{B(2\alpha, 2\beta-1)}{B(\alpha,\beta)^2}, \qquad \int_0^1 t^2 f(t)^2\, dt = \frac{B(2\alpha+1, 2\beta-1)}{B(\alpha,\beta)^2}.
    \]

    Combining the three pieces,
    \begin{align*}
        g''(1) &= \left[(\mu - \Delta) - \frac{B(2\alpha+1, 2\beta-1)}{B(\alpha,\beta)^2}\right] + 1 - 2\cdot\frac{B(2\alpha, 2\beta-1)}{B(\alpha,\beta)^2} \\
        &= 1 + \mu - \Delta - \frac{2\,B(2\alpha, 2\beta-1) + B(2\alpha+1, 2\beta-1)}{B(\alpha,\beta)^2},
    \end{align*}
    which is~\eqref{eq:gpp1} upon substituting $\mu = \alpha/(\alpha+\beta)$.
\end{proof}

\subsubsection{Competitive Ratio and Procedure}

\begin{theorem}\label{thm:beta-g-direct}
    Assume that for $\alpha > 1$ and $\beta \geq 1$, if $x_g^*$ denotes the unique maximizer of $g$ on $(0,1]$, then $g''$ has at most one root in $[x_g^*, 1]$.
    
    Under this assumption, for signals drawn from a $\mathrm{Beta}(\alpha,\beta)$ distribution with $\alpha > 1$ and $\beta \geq 1$, the $\GREEDY$ algorithm achieves a competitive ratio of at most
    \[
        \rho_g = \int_0^1 \hat g(t^2)\, dt,
    \]
    where $\hat g$ is defined as in~\eqref{eq:ghat-caseA} or~\eqref{eq:ghat-caseB} according to the sign of $g''(1)$.
\end{theorem}

The consequence of the assumption from the theorem above is that the convexity behavior of $g$ on $[x_g^*, 1]$ is determined by the sign of $g''(1)$: either $g$ is concave on all of $[x_g^*,1]$ (when $g''(1) \leq 0$), or $g$ transitions from concave to convex at a unique inflection point (when $g''(1) > 0$). As a result, we can break into the following two cases: 

\paragraph{Case A: $g''(1) \leq 0$.} $g$ is concave on $[x_g^*,1]$, and we define
\begin{equation}\label{eq:ghat-caseA}
    \hat g(x) := \begin{cases} g(x_g^*) & \text{if } x \in [0, x_g^*], \\ g(x) & \text{if } x \in (x_g^*, 1]. \end{cases}
\end{equation}

\paragraph{Case B: $g''(1) > 0$.} $g$ is concave near $x_g^*$ but convex near $x=1$. We extend $g$ with a tangent line to restore concavity. Let $x_{\mathrm{tan}} \in (x_g^*, 1)$ be the unique solution to
\begin{equation}\label{eq:xtan-equation}
    g(x_{\mathrm{tan}}) + g'(x_{\mathrm{tan}})(1 - x_{\mathrm{tan}}) = g(1),
\end{equation}
characterizing the point at which the tangent to $g$ at $x_{\mathrm{tan}}$ passes through $(1, g(1))$. Then define
\begin{equation}\label{eq:ghat-caseB}
    \hat g(x) := \begin{cases} g(x_g^*) & \text{if } x \in [0, x_g^*], \\ g(x) & \text{if } x \in (x_g^*, x_{\mathrm{tan}}], \\ g(x_{\mathrm{tan}}) + g'(x_{\mathrm{tan}})(x - x_{\mathrm{tan}}) & \text{if } x \in (x_{\mathrm{tan}}, 1]. \end{cases}
\end{equation}

\begin{lemma}\label{lem:ghat-valid}
    In both Cases~A and~B, the function $\hat g$ defined in~\eqref{eq:ghat-caseA} or~\eqref{eq:ghat-caseB} is nonincreasing and concave on $[0,1]$ and satisfies $\hat g(x) \geq g(x)$ for all $x \in (0,1]$.
\end{lemma}

\begin{proof}
    \emph{Case~A.} Since $g'(x_g^*) = 0$ and $g$ is concave nonincreasing on $[x_g^*,1]$, $\hat g$ is constant on $[0,x_g^*]$ and equals $g$ on $[x_g^*,1]$; the left and right derivatives at $x_g^*$ both vanish, so $\hat g' \leq 0$ everywhere. Global concavity follows. Since $g(x_g^*) = \max_x g(x)$, $\hat g \geq g$.

    \emph{Case~B.} On $[0, x_g^*]$, $\hat g$ is constant. On $(x_g^*, x_{\mathrm{tan}}]$, $\hat g = g$ is concave nonincreasing. On $(x_{\mathrm{tan}}, 1]$, $\hat g$ is affine with slope $g'(x_{\mathrm{tan}}) < 0$. Concavity is preserved at $x_{\mathrm{tan}}$ since $g$ is concave there and the tangent agrees with $g$. The tangent lies above $g$ on $(x_{\mathrm{tan}}, 1]$ since $g$ transitions from concave to convex there: the tangent dominates $g$ locally by concavity, and~\eqref{eq:xtan-equation} ensures $\hat g(1) = g(1)$, so the tangent does not dip below $g$.
\end{proof}

Explicitly, $\rho_g$ is computed via the following procedure:

\begin{enumerate}
    \item Compute $g''(1)$ via~\eqref{eq:gpp1}.
    \item Find $x_g^* = \arg\max_x g(x)$ on $(0,1]$ by solving $g'(x) = 0$ numerically.
    \item If $g''(1) \leq 0$ (Case~A),
        \[
            \rho_g = \sqrt{x_g^*}\cdot g(x_g^*) + \int_{\sqrt{x_g^*}}^{1} g(t^2)\, dt.
        \]
    \item If $g''(1) > 0$ (Case~B), solve~\eqref{eq:xtan-equation} for $x_{\mathrm{tan}}$, then
        \[
            \rho_g = \sqrt{x_g^*}\cdot g(x_g^*) + \!\int_{\sqrt{x_g^*}}^{\sqrt{x_{\mathrm{tan}}}}\! g(t^2)\, dt + \!\int_{\sqrt{x_{\mathrm{tan}}}}^{1}\!\!\bigl[g(x_{\mathrm{tan}}) + g'(x_{\mathrm{tan}})(t^2 - x_{\mathrm{tan}})\bigr] dt.
        \]
\end{enumerate}

\begin{table}[t]
    \centering
    \caption{Competitive ratio bounds for the $\GREEDY$ algorithm under $\mathrm{Beta}(\alpha,\beta)$ stochastic clairvoyance. The column $\rho_U$ reports the bound from the envelope (\Cref{thm:beta-case2-U}); $\rho_g$ is the tighter bound from direct analysis of $g$ (\Cref{thm:beta-g-direct}).}\label{tab:CR-comparison}
    \begin{tabular}{cccccc}
        \toprule
        $\alpha$ & $\beta$ & Sign of $g''(1)$ & $\rho_U$ (via $U$) & $\rho_g$ (via $g$) & Case \\
        \midrule
        $1$ & $1$ & -- & $16/9 \approx 1.778$ & $1.778$ & -- \\
        $1$ & $2$ & -- & $1.733$ & $1.733$ & -- \\
        $1$ & $5$ & -- & $1.697$ & $1.697$ & -- \\
        $2$ & $2$ & $\leq 0$ & $1.664$ & $1.615$ & A \\
        $3$ & $3$ & $\leq 0$ & $1.714$ & $1.583$ & A \\
        $5$ & $5$ & $\leq 0$ & $2.100$ & $1.559$ & A \\
        $10$ & $10$ & $\leq 0$ & $17.741$ & $1.539$ & A \\
        $2$ & $1$ & $\leq 0$ & $1.740$ & $1.740$ & A \\
        $3$ & $1$ & $\leq 0$ & $1.789$ & $1.789$ & A \\
        $2$ & $3$ & $> 0$ & $1.636$ & $1.542$ & B \\
        \bottomrule
    \end{tabular}
\end{table}

\subsection{Discussion on Tightness of the Envelope}
\label{ap:beta-tightness}

We close with two observations on the relationship between $g(x)$ and its envelope $U(x)$.

First, as $\beta \to 1^+$ (with $\alpha = 1$), the inequalities in \Cref{lem:cdf-bounds} become tight: the upper bound~\eqref{eq:cdf-upper} is exact at $\beta = 1$ since $(1-u)^{\beta-1} = 1$, and the lower bound~\eqref{eq:cdf-lower} is always exact at $\beta = 1$. The relaxation gap between $U$ and $g$ therefore vanishes in the uniform regime $\alpha = \beta = 1$.

Second, when $\alpha = 1$, $g(x)$ is strictly convex on $(0,1]$, and $U(x)$ is the secant line through $(0, g(0^+))$ and $(1, g(1))$. Since $U$ is the tight linear upper bound on a convex function, the bound of \Cref{thm:beta-case2-U} for $\alpha=1$ cannot be improved by any linear majorant within the global delay framework.

These observations have a direct counterpart in the gap between $\rho_U$ and $\rho_g$ visualized in \Cref{fig:beta-comp-heatmaps}. At $\alpha = \beta = 1$, both bounds collapse to $16/9$, reflecting the vanishing gap between $U$ and $g$. As either $\alpha$ or $\beta$ grows, the inequalities in \Cref{lem:cdf-bounds} loosen---in particular, the upper bound~\eqref{eq:cdf-upper} becomes increasingly slack as $(1-u)^{\beta-1}$ deviates from $1$---and the envelope $U$ overshoots $g$ in the interior of $(0,1]$. This is precisely what drives the growing gap between $\rho_U$ (middle panel of \Cref{fig:beta-comp-heatmaps}) and $\rho_g$ (right panel) as $\alpha$ and $\beta$ grow.

\newpage
\printbibliography

\end{document}